\documentclass[aps,prd,twocolumn,groupedaddress,nofootinbib]{revtex4-1}

\usepackage{amsmath,amssymb,bm}
\usepackage{graphicx}
\usepackage{epstopdf}
\usepackage{amsfonts}
\usepackage{amssymb}
\usepackage{amsbsy}
\usepackage{amsthm}
\usepackage{amsmath}
\usepackage{mathtools}
\usepackage{latexsym}
\usepackage{tikz}
\usetikzlibrary{calc,angles,positioning,intersections,quotes,decorations.markings,fit}
\usepackage{verbatim}
\usepackage{sansmath}
\usepackage{mathrsfs}
\usepackage{bm}
\usepackage{color}
\usepackage{comment}
\usepackage{csquotes}								
\usepackage{physics}
\usepackage{eufrak}
\usepackage{enumitem}
\usepackage{accents}
\usepackage[colorlinks=true,breaklinks]{hyperref}
\usepackage[caption=false]{subfig}
\usepackage[utf8]{inputenc}

\newtheorem{theorem}{Theorem}
\newtheorem{definition}{Definition}

\newtheorem{lemma}[theorem]{Lemma}

\newcommand{\di}{\partial}
\newcommand{\res}{\omega}

\newcommand{\T}{\footnotesize{\text{T}}}

\renewcommand{\H}{\footnotesize{\textup{H}}}

\def\be{\begin{equation}}
\def\ee{\end{equation}}
\def\bea {\begin{eqnarray}}
\def\eea {\end{eqnarray}}

\begin{document}
\title{Velocity-dependent deformations of the energy spectrum of a quantum cavity from Lorentz symmetry violations}

\author{Jarod George Kelly} \email{jarod.kelly@unb.ca} \affiliation{Department of Mathematics and Statistics, University of New Brunswick, Fredericton, NB, Canada E3B 5A3}

\author{Sanjeev S.\ Seahra} \email{sseahra@unb.ca} \affiliation{Department of Mathematics and Statistics, University of New Brunswick, Fredericton, NB, Canada E3B 5A3}

\begin{abstract}

Several approaches to quantum gravity suggest that Lorentz invariance will be broken at high energy.  This can lead to modified dispersion relations for wave propagation, which can be concretely realized in effective field theories where the equation of motion involves higher order spatial derivatives in a preferred frame.  We consider such a model in the presence of a finite cavity whose walls follow parallel inertial trajectories of speed $v$ with respect to the preferred frame.  We find evidence that when the cavity wall speed exceeds the phase velocity, the system becomes classically unstable.  For dispersion relations that do not lead to an instability, the energy levels of the cavity are non-trivial functions of $v$.  In other words, an observer could in principle measure their velocity with respect to the preferred frame by studying the energy spectra of a quantum cavity, which is a stark violation of the principle of relativity.  We also find that the energy levels of the cavity become infinitely large as its velocity approaches light speed.  

\end{abstract}

\maketitle

\section{Introduction}

Lorentz symmetry is one of the cornerstones of modern physics.   The mathematical and physical properties of fundamental theories are severely constrained by the simple axioms that there are no preferred inertial observers in a sufficiently small patch of spacetime, and that there is a maximum universal speed.  However, there is an inherent tension between Lorentz invariance and quantum theories that hypothesize that spacetime is fundamentally discrete, or that there exists a preferred rest frame for physical theories.  This has led many authors to speculate that Lorentz symmetry is an emergent phenomena that only holds approximately at large scales, and it is therefore possible for very high energy or small-scale experiments to detect Lorentz invariance violations (see refs.~\cite{Mattingly:2005re,Vucetich:2005ra,Russell:2006ge,Liberati:2012tb,Liberati:2013xla,Liberati:2015dja} for reviews).  Specific models exhibiting Lorentz symmetry violation have been obtained in the context of string theory \cite{Kostelecky:1988zi,Kostelecky:1989jp,Kostelecky:1989jw,Kostelecky:1991ak,Kostelecky:1995qk,Kostelecky:1999mu,Sudarsky:2002zy,Jenkins:2003hw}, loop quantum gravity \cite{Gambini:1998it,Alfaro:1999wd,Alfaro:2001rb,Alfaro:2002xz}, and non-commutative geometry \cite{Seiberg:1999vs, Hayakawa:1999yt, Mocioiu:2000ip, Guralnik:2001ax,Douglas:2001ba,Jurco:2001rq,Carlson:2001sw,Carroll:2001ws,Anisimov:2001zc} (for other examples see \cite{Mattingly:2005re} and references therein).

A useful technique for modelling the effects of small scale Lorentz symmetry violation in a model independent way is effective field theory \cite{Polchinski:1992ed}, which involves adding Lorentz-symmetry breaking terms to the Lagrangian of an otherwise Lorentz invariant theory (see Refs.~\cite{Ng:1993jb,Shiokawa:2000em,Jacobson:2000xp,AmelinoCamelia:2000ge,AmelinoCamelia:2000mn,Bruno:2001mw,Dowker:2003hb,KowalskiGlikman:2004qa,Lim:2004js,Jacobson:2004ts} for other phenomenological approaches). These symmetry breaking terms can be categorized by their mass dimension; for example, the standard kinetic term in scalar field theory $\di_{\alpha} \phi \di^{\alpha}\phi$ is dimension 4, while a term of the form $M_{\star}^{-2}\di_{\alpha}\di_{\beta} \phi \, \di^{\alpha} \di^\beta\phi$ is dimension 6, where $M_{\star}$ is some fixed mass scale.  From a phenomenological point of view, the higher the dimension of a given term (or operator) in an effective field theory, the less important it is at low energies.  For this reason, much of the work looking for experimental signatures of Lorentz invariance violations has concentrated on constraining effective field theories with symmetry-breaking terms of dimension $\lesssim 4$.  


One of the most popular examples of a Lorentz violating effective field theory is the Standard Model Extension (SME) \cite{Colladay:1996iz,Colladay:1998fq,Kostelecky:2003fs,Bluhm:2005uj,Greenberg:2002uu,Greenberg:2011cw}.  A version of the SME incorporating such symmetry-breaking terms is called the ``minimal SME'' \cite{Colladay:1998fq}.  The photonic sector of the theory is described by the Lagrangian $\mathcal{L} = -\frac{1}{4} F_{\mu\nu} F^{\mu\nu} - \frac{1}{4} (k_{F})_{\kappa\lambda\mu\nu}F^{\kappa\lambda}F^{\mu\nu}$.  Here, $(k_{F})_{\kappa\lambda\mu\nu}$ is a non-dynamical background tensor field that parameterized Lorentz symmetry breaking.  For a comprehensive summary of current theoretical and experimental limits on Lorentz violating coefficients in the mSME see \cite{Kostelecky:2008ts}.


One can also construct effective field theories with unequal numbers of temporal and spatial derivatives (see for e.g. \cite{Colladay:1998fq,Jacobson:2000xp,ArkaniHamed:2003uy,Lim:2004js,Jacobson:2004ts,Kostelecky:2009zp}).  In such models, the wave 4-vector $k^{\alpha}$ appearing in plane wave solutions $\phi \sim e^{i k_{\alpha} x^{\alpha}}$ for massless fields does not in general satisfy the Lorentz-invariant dispersion relation $k_{\alpha} k^{\alpha} = 0$.  Scalar field models of this type have been used to analyze the effects of small scale Lorentz invariance violations on the Hawking radiation from black holes \cite{PhysRevD.51.2827,Corley:1996ar} and the Unruh effect \cite{Husain:2015tna}.  Modified dispersion relations can be directly constrained by astrophysical observations, since they imply that different frequency components of an astrophysical signal will travel at different speeds and hence arrive at a detector at different times.  Because astrophysical signals travel large distances to reach us, this frequency dependent time delay is in principle detectable even if velocity differences are very small.  This idea has been used to analyze high-frequency electromagnetic radiation from gamma ray bursts and hence place limits on non-Lorentz invariant dispersion relations \cite{Galaverni:2008yj,Stecker:2011ps,Vasileiou:2013vra,Kislat:2016ehx,Wei:2017zuu,Lang:2018yog}.

While astrophysical observations represent the best means we have to constrain higher dimensional Lorentz-violating operators, it is obviously useful to have independent and reproducible laboratory-based techniques for probing these effects.  These include clock-comparison experiments \cite{Kostelecky:1999mr,Mewes:2003mw,Russell:2004ne,Kostelecky:2018fmc}, Penning traps \cite{Bluhm:1997ci,Bluhm:1997qb,Russell:2004ne} and doppler shift experiments  (for a more comprehensive discussion, see \cite{Mattingly:2005re,Vucetich:2005ra,Russell:2006ge,Liberati:2012tb,Liberati:2013xla,Liberati:2015dja}).


Below, we concentrate on tests involving microwave and optical cavities.  Such experiments are essentially modern versions of the classic Michelson-Morley and Kennedy-Thorndike experiments.  The basic idea is to determine if the resonant frequencies of radiation confined to a cavity depend on the cavity's orientation in space or state of motion.  Some of the earliest experiments of this type directly addressed whether the speed of light was independent of direction \cite{Brillet:1979zz} or the state of motion of a cavity \cite{Hils:1990nrd,Braxmaier:2001wu}.  The change in resonant frequency of a cavity in the minimal SME was calculated to leading order in Ref.~\cite{Kostelecky:2002hh}.  Subsequent experiments  \cite{Muller:2002uk,Muller:2003zzc,Muller:2004zp,Wolf:2004gg,Herrmann:2005qe,Antonini:2005yb,Stanwix:2006jb,Herrmann:2006zza,Herrmann:2009zzb,Nagel:2014aga} used this result to directly constrain the components of the $(k_{F})_{\kappa\lambda\mu\nu}$ tensor in the Sun's rest frame.  The effects of higher mass dimension terms in the SME on cavity experiments were derived in \cite{Kostelecky:2009zp,Mewes:2012sm} in the context of perturbation theory and considering the effects of small boosts with $v/c \sim 10^{-4}$; the effects were subsequently probed experimentally in \cite{Parker:2015ena}.


In this work, we carefully calculate the energy levels of a quantum cavity in an effective field theory with higher derivatives, and hence modified dispersion.  We work with a scalar field for simplicity, and pay close attention to boundary conditions.  Unlike previous studies \cite{Kostelecky:2002hh,Kostelecky:2009zp,Mewes:2012sm}, we concentrate on the determination the spectrum of the cavity as a function of its inertial velocity with respect to a preferred frame.   We do not \emph{a priori} assume the cavity's velocity is small, and we present a fully quantum and self-consistent calculation.

In a Lorentz invariant theory, one would expect a cavity's energy spectrum to be velocity independent; i.e., any observer comoving with the cavity must measure the same discrete energy spectrum for longitudinal modes.  This is a direct consequence of the principle of relativity, which states that the results of all local experiments in a non-accelerating laboratory are independent of the laboratory's velocity.  We find that the opposite is basically true for scalar fields with non-Lorentz invariant dispersion relations arising from an effective field theory.  In order to maintain invariance under parity, we only consider effective actions with even dimension terms (see for e.g. \cite{Kostelecky:2000mm,Greenberg:2002uu,Greenberg:2011cw} for models incorporating CPT violations). We find evidence that some dispersion relations lead to the classical instability of the system.  For dispersion relations giving rise to stable dynamics, we see that higher dimension operators in the scalar field action imply that an observer can measure their velocity with respect to the preferred frame by measuring the energy levels of a quantum cavity.

In section \ref{sec:classical}, we present the classical formalism of our model including the action, Hamiltonian, and equation of motion.  In section \ref{sec:Fourier}, we present a Fourier-like mode decomposition of the scalar field that is used to quantize the model in section \ref{sec:quantization}.  The problem of quantization is made significantly more complicated by the fact that the Fourier modes of section \ref{sec:Fourier} are not orthogonal under the conserved inner product for the problem.  

In section \ref{sec:Bogliobov}, we rectify this by deducing the conditions under which we can find Bogliobov transformations that result in an orthogonal mode basis and a diagonal Hamiltonian.  Roughly speaking, such a transformation exists if the cavity walls do not exceed the phase velocity of any of the Fourier modes in the scalar field spectrum.  If this condition is satisfied, the cavity is called ``subsonic'' and its energy spectrum is simply given by integer multiples of the resonant or normal mode frequencies of the cavity (which are essentially the diagonal entries of the Hamiltonian operator in the new basis).  If the condition is not satisfied, we call the cavity ``supersonic'' and we are not guaranteed the existence of the Bogliobov transformation discussed in section \ref{sec:Bogliobov}, which in turn implies that the energy spectrum of the cavity is not well-defined.  Whether or not it is even possible for a cavity to be supersonic depends on the scalar field dispersion relation, or, equivalently, the nature of the higher dimension terms added to the action.  Technical proofs of the results of section \ref{sec:Bogliobov} are given in appendix \ref{sec:existence}.

In section \ref{sec:subsonic}, we calculate the resonant frequencies of subsonic cavities in various limits and for various choices of Lorentz-symmetry breaking terms.  We find if the cavity's velocity $v$ with respect to the preferred frame is small, the perturbation to the energy levels is $\mathcal{O}(v^{2})$.  Conversely, for $|v| \rightarrow 1$ the resonant frequencies are generally divergent, implying that it takes an infinite amount of energy to excite the cavity above its ground state.  In section \ref{sec:supersonic}, we revisit the case of the supersonic cavity and present some numeric and analytic evidence that the cavity is classically unstable under these conditions (though we do not have a proof of this).  In section \ref{sec:experiments}, we discuss the implications of our results for rotating cavity experiments under the plausible (but unproven) assumption that our scalar field results will generalize to the electromagnetic case.   Despite the exceptionally low error in these experiments, we find that the associated constraints placed on higher dimension Lorentz-violating operators are not competitive with constraints derived astrophysical tests.  

Section \ref{sec:discussion} is reserved for a discussion of our results.

\section{Classical formalism}\label{sec:classical}

In this paper, we consider a scalar field propagating in flat space with metric
\begin{equation}
	ds^{2} = -dT^{2} + dX^{2} + dY^{2} + dZ^{2}. 
\end{equation}
In these coordinates, the scalar field is assumed to be governed by the action
\begin{equation}
	S = \frac{1}{2} \int d^{4}X [(\di_{T}\phi)^{2} - (\vec{D}\phi)^{2} ], 
\end{equation}
where
\begin{equation}
	\vec{D} = \vec\nabla ( c_{1} + c_{2} L_{\star}^{2} \vec{\nabla} ^{2} + \cdots + c_{M} L_{\star}^{2M-2} \vec{\nabla} ^{2M-2} ).
\end{equation}
Here, $M \ge 2$ is an integer, the coefficients $c_{i}$ are dimensionless, while $L_{\star}$ is a fixed length scale.  In this frame, the action is invariant under spatial parity $(X,Y,Z) \mapsto (-X,-Y,-Z)$, spatial rotations, and time inversion $T \mapsto -T$.  Below, we take $c_{1} = 1$.  In the language of effective field theory, we are considering CPT invariant deformations of the Lorentz invariant massless scalar action by the addition of dimension $4,6,\ldots,4M$ operators.

We assume that the scalar field is confined to a semi-infinite cavity $\Omega$ (shown in figure \ref{fig:geometry 1}) whose walls are the surfaces:
\begin{equation}
	\Sigma_\text{L}: \,\, X = -vT, \quad \Sigma_\text{R}: \,\, X = -vT+L.
\end{equation}
That is, the walls of the cavity follow inertial trajectories with velocity $dX/dT = -v$ and are separated by a distance of $L$ as measured along a $T =$ constant surface.  Therefore, the proper length $L_{0}$ of the cavity (i.e., it's length as measured in its rest frame) is
\begin{equation}\label{eq:Lorentz contraction}
	L_{0} = \gamma L, \quad \gamma = \frac{1}{\sqrt{1-v^{2}}}.
\end{equation}
\begin{figure*}
\begin{center}
\includegraphics[width=\textwidth]{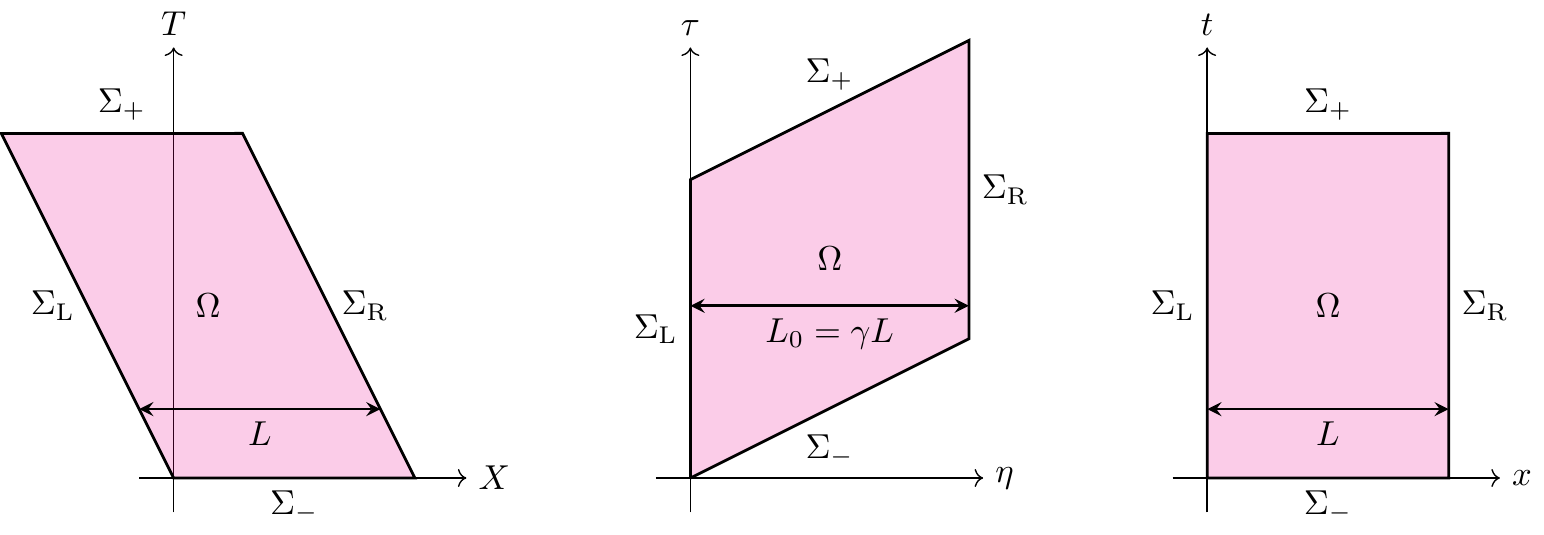}
\end{center}
\caption{Cavity geometry in various coordinate systems}\label{fig:geometry 1}
\end{figure*}

We will be primarily interested in the longitudinal modes of the cavity that are independent of $Y$ and $Z$.  The reduced action governing these modes is
\begin{equation}
	S = \frac{1}{2} \int_{\Omega} dT \, dX \left\{  (\di_{T}\phi)^{2} - [\hat{F}(\di_{X})\phi]^{2} \right\}, 
\end{equation}
where 
\begin{align}\nonumber
	\hat{F}(\di_{X}) & = \di_{X} + c_{2} L_{\star}^{2} \di_{X}^{3}+ \cdots + c_{M} L_{\star}^{2M-2}  \di_{X}^{2M-1} \\ &= \sum_{i=1}^{M} c_{i} L_{\star}^{2i-2} \di^{2i-1}_{X}.\label{eq:power F}
\end{align}
If $c_{2} = c_{3} = \cdots = c_{M} = 0$ (or if $L_{\star} = 0$), then the action takes the familiar form
\begin{equation}\label{eq:TX action}
	S = \frac{1}{2} \int_{\Omega} dT \, dX \left[  (\di_{T}\phi)^{2} - (\di_{X}\phi)^{2} \right]. 
\end{equation}

We now examine the behaviour of the reduced action under Lorentz transformations.  Specifically, we perform a boost to a frame where the walls $\Sigma_\text{L}$ and $\Sigma_\text{R}$ are at rest:
\begin{gather}\nonumber
	\tau = \gamma(T+vX), \quad \eta = \gamma(vT+X), \\ ds^{2} = -d\tau^{2}+d\eta^{2} + dY^{2} + dZ^{2}.
\end{gather}
The transformed action is 
\begin{equation}\label{eq:tau eta action}
	S = \frac{1}{2} \int_{\Omega} d\tau \, d\eta \left\{  [\gamma(\di_{\tau}+v\di_{\eta})\phi]^{2} - [\hat{F}(\gamma(v\di_{\tau}+\di_{\eta}))\phi]^{2} \right\}. 
\end{equation}
If $\hat{F}(\di_{X}) = \di_{X}$, then the action is
\begin{equation}
	S = \frac{1}{2} \int_{\Omega} d\tau \, d\eta \left[  (\di_{\tau}\phi)^{2} - (\di_{\eta}\phi)^{2} \right]; 
\end{equation}
i.e., the action is form-invariant under Lorentz transformations.  However, if $\hat{F}(\di_{X}) \ne \di_{X}$, then the action is not Lorentz invariant, and will involve higher order time derivatives.  Comparison of (\ref{eq:TX action}) and (\ref{eq:tau eta action}) clearly demonstrates that when $\hat{F}(\di_{X}) \ne \di_{X}$, the $(T,X)$ coordinates define a preferred frame.

While it is certainly possible to quantize theories with higher order time derivatives in the action, it is cumbersome.  To avoid this complication, we consider a Galilean change of coordinates
\begin{gather}\nonumber
	t = T, \quad x = X + vT \\
	ds^{2} = -dt^{2} + (dx-v \, dt)^{2}+dY^{2}+dZ^{2}.\label{eq:metric}
\end{gather}
In these coordinates, the reduced action is
\begin{equation}\label{eq:action}
	S = \frac{1}{2} \int dt \int_{0}^{L} dx  \left\{ [(\di_{t}+v \di_{x})\phi]^{2} - (\hat{F}\phi)^{2} \right\}. 
\end{equation}
Here and below, $\hat{F}$ is understood to refer to $\hat{F}(\di_{x})$.  Clearly, the action is not form invariant under Galilean transformations either.  However, this coordinate system has the advantage that the cavity walls are stationary and the action only involves first order time derivatives.

We now consider the variation of action with respect to $\phi$, subject to the usual assumption that the variation vanishes on the boundary of region of interest: $\delta\phi|_{\di\Omega} = 0$.  We obtain, after integration by parts
\begin{equation}\label{eq:variation}
	\delta S = - \int dt \int_{0}^{L} dx  \, [ \delta \phi (\di_{t}+v \di_{x})^{2}\phi 
	 +(\hat{F}\delta\phi) (\hat{F}\phi) ]. 
\end{equation}
Now, for any functions $U$ and $V$ satisfying the boundary conditions
\begin{equation}\label{eq:BCs}
	\di^{(2i)}_{x} U |_{x=0,L} = \di^{(2i)}_{x} V |_{x=0,L} = 0, \quad i = 0 \ldots M,
\end{equation}
we have the integration by parts identity\footnote{The identity also holds if all the odd derivatives of $u$ and $v$ vanish on $\Sigma_\text{L}$ and $\Sigma_\text{R}$.}
\begin{equation}\label{eq:integration by parts identity}
	\int_{0}^{L} dx \, (\hat{F}V) (\hat{F} U) = - \int_{0}^{L} dx \, U \hat{F}^{2}V.
\end{equation}
Assuming that $\phi$ and $\delta\phi$ satisfy the boundary conditions (\ref{eq:BCs}) and setting $\delta S = 0$, this identity yields the equation of motion
\begin{equation}\label{eq:wave equation}
	 (\di_{t}+v \di_{x})^{2}\phi-\hat{F}^{2}(\di_{x})\phi = 0, \quad  \di^{(2i)}_{x} \phi |_{x=0,L} = 0,
\end{equation}
where $i = 0 \ldots M$.  \footnote{As an aside, we note that if we set $v = 0$ and send $L \rightarrow \infty$, we recover the equation of motion of the scalar in the preferred frame in the absence of finite boundaries:
\begin{equation}
	 \di_{t}^{2}\phi-\hat{F}^{2}(\di_{x})\phi = 0.
\end{equation}
Plane wave solutions $\phi = e^{-i\omega t + i kx}$ satisfy the dispersion relation
\begin{equation}
	\omega^{2} = k^{2} - 2c_{2}L_{\star}^{2} k^{4}+\cdots
\end{equation}
We can compare this to a typical parameterization of modified dispersion relations for massless bosons \cite{Vasileiou:2013vra}:
\begin{equation}
	\omega^{2} = k^{2} \pm \frac{k^{4}}{E_\text{QG}^{2}} + \cdots
\end{equation}
where $E_\text{QG}$ is some quantum gravity energy scale.
Comparison of these two equations yields $E_\text{QG} = (2|c_{2}|)^{-1/2} L_{\star}^{-1}$.  We will use this in section \ref{sec:experiments} to compare our results to astrophysical constraints on Lorentz violations.}

We now define an inner product $(\cdot,\cdot)$ between two solutions $\phi$ and $\psi$ of (\ref{eq:wave equation}) as
\begin{align}
	\nonumber (\phi,\psi) & = i \int_{0}^{L} dx [\psi^{*} (\di_{t}+v \di_{x}) \phi - \phi (\di_{t}+v \di_{x}) \psi^* ]  \\ & = i \int_{0}^{L} dx \, \psi^{*} \overleftrightarrow{\di_{T}}\phi.\label{eq:inner product}
\end{align}
Useful properties of this inner product include
\begin{gather}\nonumber
	(\mu \phi,\nu\psi) = \mu\nu^{*}(\phi,\psi), \\ (\phi^{*},\psi^{*}) = -(\phi,\psi)^{*} = -(\psi,\phi), \label{eq:identities}
\end{gather}
where $\mu$ and $\nu$ are constant scalars.  The identity (\ref{eq:integration by parts identity}) and wave equation (\ref{eq:wave equation}) can be used to explicitly show that this inner product is conserved in time
\begin{equation}
	\frac{d}{dt}(\phi,\psi) = 0.
\end{equation}
Another useful identity is
\begin{multline}\label{eq:identity 2}
	(\phi,\dot\psi) = i \int dx [ (\di_{t}\psi^{*})(\di_{t}\phi) +  \\ (\hat{F}\psi^{*})(\hat{F}\phi)  - v^{2} (\di_{x}\phi)( \di_{x}\psi^{*})],
\end{multline}
where we use an overdot to denote $\di/\di t$ and have again assumed that $\phi$ and $\psi$ are solutions of the wave equation.  This leads to the facts:
\begin{equation}\label{eq:identity 3}
	(\phi,\dot\psi) = -(\psi,\dot\phi)^{*}=(\psi^{*},\dot\phi^{*}),
\end{equation}
which we will make use of below.

We now pass over to the Hamiltonian formalism.  From (\ref{eq:action}), we see that the Lagrangian of the system is
\begin{equation}
	 \mathcal{L} = \frac{1}{2} \left\{ [(\di_{t}+v \di_{x})\phi]^{2} - [\hat{F}(\di_{x})\phi]^{2} \right\}.
\end{equation}
The momentum conjugate to $\phi$ is
\begin{equation}
	\pi = \frac{\delta\mathcal{L}}{\delta(\di_{t}\phi)} =  (\di_{t}+v \di_{x})\phi = \di_{T} \phi,
\end{equation}
and we have the usual Poisson brackets
\begin{gather}
	\nonumber \{ \phi(t,x), \phi(t,x') \}=\{ \pi(t,x), \pi(t,x') \}=0, \\ \{ \phi(t,x), \pi(t,x') \} = \delta(x-x').
\end{gather}
After performing the standard Legendre transformation, we find that the full Hamiltonian is
\begin{align}
	\nonumber H & = \frac{1}{2} \int_{0}^{L} dx [ \pi^{2} + (\hat{F}\phi)^{2} - 2 v\pi \di_{x}\phi ]  \\ & = \frac{1}{2} \int_{0}^{L} dx [ (\di_{t}\phi)^{2} + (\hat{F}\phi)^{2} - v^{2} (\di_{x} \phi)^{2} ]. 
\end{align}
Comparison of this expression with (\ref{eq:identity 2}) yields the following compact formula:
\begin{equation}\label{eq:classical Hamiltonian}
	H = -\frac{i}{2} ( \phi,\dot\phi).
\end{equation}
Now, if $\phi$ is a solution of (\ref{eq:wave equation}), then $\dot\phi$ is also a solution satisfying the same boundary conditions from which it follows that the Hamiltonian is conserved $dH/dt = 0$.  Furthermore, if we set $\psi = \phi$ in (\ref{eq:identity 3}), then we see that $(\phi,\dot\phi)$ is imaginary, implying that $H$ is real, as required.

Before moving on to the quantization of this system, we note that the Hamiltonian defined above is the generator of $t$ evolution via the Poisson bracket:
\begin{equation}
	\frac{df}{dt} = \{ f,H\}.
\end{equation}
Now, the coordinate time $t$ is not the proper time $\tau$ for observers comoving with the cavity; i.e. $x=$ constant observers.  From the metric (\ref{eq:metric}), $\tau$ is related to $t$ by
\begin{equation}
	dt  = \gamma \, d\tau, \quad \frac{df}{d\tau} = \left\{ f,\gamma H \right\}.
\end{equation}
Hence the proper time Hamiltonian (i.e., the generator of $\tau$-evolution), is
\begin{equation}\label{eq:proper time Hamiltonian}
	H_{\tau} = \gamma H.
\end{equation}

\section{Fourier mode functions}\label{sec:Fourier}

Following the standard procedure, in order to quantize the system presented in the last section we need to specify a complete basis of solutions for the wave equation (\ref{eq:wave equation}).  The natural temptation is to choose mode functions that resemble those used for the standard wave equation in the presence of a reflecting boundary;  that is, mode functions satisfying Fourier-like initial conditions 
\begin{equation}\label{eq:initial data}
	\psi_{n} |_{\Sigma_{-}} = \frac{1}{\sqrt{\pi\zeta_{n}}} \sin (k_{n} x), \quad \di_{t}\psi_{n} |_{\Sigma_{-}} = -i \zeta_{n} \psi_{n} |_{\Sigma_{-}}, \quad 
\end{equation}
where $n = 1,2,3 \ldots$, $k_{n} = {\pi n}/{L}$, and $\zeta_{n} \ne 0$ are constants yet to be determined.  These explicitly satisfy the boundary conditions on the initial hypersurface
\begin{equation}
	\di_{x}^{2i}\psi_{n}|_{\Sigma_{-}\cap\Sigma_{L}} = \di_{x}^{2i}\dot\psi_{n}|_{\Sigma_{-}\cap\Sigma_{R}} = 0, \quad i = 0,1,2 \ldots
\end{equation}
Now, in order for $\{\psi_{n},\psi_{n}^{*}\}$ to comprise a complete basis, we cannot choose the $\zeta_{n}$ in a completely arbitrary fashion.  To see why, let $\phi$ be an arbitrary real solution of the wave equation (\ref{eq:wave equation}), and further suppose it can be decomposed as follows:
\begin{equation}\label{eq:Fourier decomposition 1}
	\phi = \frac{\pi}{L} \sum_{n=1}^{\infty}(a_{n}\psi_{n}+a_{n}^{*}\psi_{n}^{*}).
\end{equation}
We can use this to evaluate $\phi$ and $\dot\phi$ on $\Sigma_{-}$:
\begin{subequations}
\begin{align}
	\phi|_{\Sigma_{-}} & = \frac{2\pi}{L}\sum_{n} \text{Re} (A_{n}) \sin \left( \frac{n\pi x}{L} \right) , \\ \dot\phi|_{\Sigma_{-}} & = - \frac{2\pi}{L}\sum_{n} \text{Im} (\zeta_{n} A_{n}) \sin \left( \frac{n\pi x}{L} \right). 
\end{align}
\end{subequations}
with $ A_{n} = {a_{n}}/{\sqrt{\pi\zeta_{n}}}$.  If $\text{Re}(\zeta_{n}) \ne 0$, then these expressions represent independent Fourier sine series for $\phi|_{\Sigma_{-}}$ and $\dot\phi|_{\Sigma_{-}}$, respectively.  This implies that $\{\psi_{n},\psi_{n}^{*}\}$ form an $L_{2}$-complete basis for solutions of the wave equation (\ref{eq:wave equation}).  However, if $\text{Re}(\zeta_{n}) = 0$ for any $n$, then $\phi|_{\Sigma_{-}}$ and $\dot\phi|_{\Sigma_{-}}$ cannot be independently selected, and $\{\psi_{n},\psi_{n}^{*}\}$ will fail to be a basis.  In what follows, we assume that all the $\zeta_{n}$ are real and positive.

Using the initial data (\ref{eq:initial data}), it is easy to see that $\{\psi_{n},\psi_{n}^{*}\}$ is \emph{not} an orthogonal basis under the inner product (\ref{eq:inner product}) if $v \ne 0$:
\begin{subequations}\label{eq:explicit inner 1}
\begin{align}
	(\psi_{n},\psi_{m}) & = - (\psi^{*}_{n},\psi^{*}_{m})^{*} = \frac{L}{\pi} (\delta_{nm}+\sigma_{nm}), \\ (\psi_{n},\psi_{m}^{*}) & = (\psi^{*}_{n},\psi_{m}) = \frac{L}{\pi}\sigma_{nm},
\end{align}
\end{subequations}
where
\begin{equation}
\sigma_{nm} = \sigma^{*}_{mn} = \begin{cases} 0, & n = m, \\ \displaystyle \frac{2ivnm[1-(-1)^{n+m}]}{L\sqrt{\zeta_{n}\zeta_{m}}(m^{2}-n^{2})}, & n \ne m,  \end{cases}
\end{equation}
However, it turns out that these mode functions and their time derivatives do satisfy orthogonality relationships of the form
\begin{equation}\label{eq:dot orthogonality}
\begin{array}{r} (\psi_{n},\dot\psi_{m})  \\ (\psi_{n},\dot\psi_{m}^{*})  \end{array} \bigg\} = \frac{iL[\pm \zeta_{n}^{2}  + F^{2}(k_{n})-v^{2}k_{n}^{2} ] }{2\pi \zeta_{n}} \delta_{nm},
\end{equation}
where
\begin{equation}
	F(k) = -i\hat{F}(ik).
\end{equation}
Equation (\ref{eq:dot orthogonality}) does suggest that the Hamiltonian (\ref{eq:classical Hamiltonian}) will have a particularly simple form in terms of these mode functions.  However, it would be ideal to work with a basis for which both of $(\psi_{n},\psi_{m})$ and $(\psi_{n},\dot{\psi}_{m})$ are proportional to $\delta_{nm}$.   Unfortunately, there does not seem to be a closed form expression for initial data for such mode functions, so we will work with the non-orthogonal basis defined by (\ref{eq:initial data}) for the rest of this paper.

We have yet to fix the functional relationship between $\zeta_{n}$ and $k_{n}$.  Equation (\ref{eq:dot orthogonality}) suggests that a convenient choice is
\begin{equation}\label{eq:dispersion}
\zeta_{n} = \sqrt{|F^{2}(k_{n})-v^{2}k_{n}^{2}|},
\end{equation}
which simplifies the inner products (\ref{eq:dot orthogonality}) considerably.  We can assign a physical interpretation to this dispersion relation choice by considering solutions to the wave equation (\ref{eq:wave equation}) with $e^{ikx}$ spatial dependance in the case of an infinitely large cavity; i.e., in the $L \rightarrow \infty$ limit:
\begin{equation}\label{eq:plane wave sol}
	\phi = \mathcal{C}_{+} e^{ -i [F(k)t+vkt-kx]} + \mathcal{C}_{-} e^{ i [F(k)t-vkt+kx]},
\end{equation}
where $\mathcal{C}_{\pm}$ are constants.  This is superposition of plane wave modes with phase velocities
\begin{equation}
	u_{\pm}(k) = v \pm \frac{F(k)}{k}.
\end{equation} 
The geometric mean of the associated phase speeds is:
\begin{equation}
	u(k) = \sqrt{|u_{+}(k)u_{-}(k)|} = \frac{\sqrt{|F^{2}(k)-v^{2}k^{2}|}}{k}.
\end{equation}
Comparison with (\ref{eq:dispersion}) yields
\begin{equation}\label{eq:phase vel}
	\zeta_{n} = u_{n} k_{n}, \quad u_{n} = u(k_{n}).
\end{equation}
That is, the constant of proportionality between $\zeta_{n}$ and $k_{n}$ is the mean phase speed of plane waves modes with wavenumber $k_{n}$ when $L \rightarrow \infty$.

It is important to note that if
\begin{equation}\label{eq:supersonic condition}
	v^{2} > \frac{F^{2}(k)}{k^{2}}
\end{equation}
for a given choice of $k$, the phase velocities $u_{\pm}$ will have the same sign.  Now, in the $(T,X)$ coordinates the plane wave solution (\ref{eq:plane wave sol}) is
\begin{equation}
	\phi = \mathcal{C}_{+} e^{ -i [F(k)T-kX]} + \mathcal{C}_{-} e^{ i [F(k)T+kX]}.
\end{equation}
This is clearly a superposition of a left and right moving modes with identical phase speeds $|F(k)|/k$.  We therefore see that the condition (\ref{eq:supersonic condition}) implies that the $(t,x)$ frame is moving faster than the phase speed of modes with wavenumber $k$ as measured in the preferred frame.  That is, the $(t,x)$ frame is moving with a supersonic velocity with respect to such modes.  Now, when $L$ is finite, the cavity's walls will be measured to have speed $v$ in the preferred frame, which causes us to adopt the following terminology: if
\begin{equation}\label{eq:fast vs slow}
	F^{2}(k_{n})-v^{2}k_{n}^{2} > 0, \quad n = 1,2 \ldots,
\end{equation} 
we say that the cavity velocity is \emph{subsonic}; otherwise, we call the cavity velocity \emph{supersonic}.  Furthermore, we call modes for which the inequality (\ref{eq:fast vs slow}) holds \emph{fast} modes, all other modes are called \emph{slow}.  It follows that if $F^{2}(k) \ge k^{2}$ for all $k>0$, then there will exist no slow modes for all $|v| < 1$.  Conversely, if $F^{2}(k) < k^{2}$ for some $k>0$, then there will exist slow modes for certain cavity velocities.

To conclude this section, we note that when both $v$ and $L$ are finite, we do not know closed form expressions for the mode functions defined by (\ref{eq:initial data}) throughout $\Omega$.  This is not a fundamental difficulty since---as we will see explicitly below---all we use to find the spectrum of the Hamiltonian is explicit knowledge of the conserved inner products (\ref{eq:explicit inner 1}) and (\ref{eq:dot orthogonality}).

\section{Quantization using nonorthogonal modes}\label{sec:quantization}

We now quantize the system by promoting $\phi$ and $\pi$ to operators satisfying commutation relations
\begin{gather}
	 \nonumber [\hat{\phi}(t,x),\hat{\phi}(t,x')] =   [\hat{\pi}(t,x),\hat{\pi}(t,x')]  = 0, \\
	 [\hat{\phi}(t,x),\hat{\pi}(t,x')] = i \delta(x-x'), \label{eq:equal time}
\end{gather}
We consider the following mode decomposition for $\hat\phi$:
\begin{equation}\label{eq:mode decomposition}
	\hat\phi(t,x) = \frac{\pi}{L} \sum_{n=1}^{\infty} [\hat{a}_{n} \psi_{n}(t,x) + \hat{a}^{\dag}_{n} \psi^{*}_{n}(t,x) ],
\end{equation}
where $\{\psi_{n},\psi_{n}^{\star}\}$ are the Fourier mode functions introduced in the last section.\footnote{Most of the formulae in this section and the next are valid for an arbitrary choice of basis.}  This means that we cannot make the usual assumption that the modes satisfy orthogonality relations of the form
\begin{equation}\label{eq:algebra 0}
	 (\psi_{n},\psi_{m})= \frac{L}{\pi} \delta_{nm} = - (\psi_{n}^{*},\psi_{m}^{*}), \quad (\psi_{n},\psi_{m}^{*}) = 0.
\end{equation}

For various technical reasons encountered below, we will find it convenient to introduce an ultraviolet trucation in our mode expansion.  Specifically, we retain the first $N$ Fourier modes in the expansion (\ref{eq:mode decomposition}):
\begin{equation}\label{eq:mode decomposition}
	\hat\phi(t,x) = \frac{\pi}{L} \sum_{n=1}^{N} [\hat{a}_{n} \psi_{n}(t,x) + \hat{a}^{\dag}_{n} \psi^{*}_{n}(t,x) ],
\end{equation}
Physically, this corresponds to considering solutions of the wave equation generated by initial data on $\Sigma_{-}$ whose Fourier transform contains wave numbers $\le k_{N} = N\pi/L$.  Eventually, we will take the limit $N \rightarrow \infty$.

We also find it convenient to introduce some vector/matrix notation for the operators and mode functions appearing in (\ref{eq:mode decomposition}).  We write:
\begin{equation}\label{eq:a vector}
	\hat{\mathbf{a}} = \begin{pmatrix*}[c]  \hat{a}_{1} \\ \vdots \\ \hat{a}_{N} \\ \hline  \hat{a}^{\dag}_{1} \\ \vdots \\ \hat{a}^{\dag}_{N}  \end{pmatrix*}, \quad \bm{\psi} = \begin{pmatrix*}[c]  \psi_{1} \\ \vdots \\ \psi_{N} \\ \hline    \psi^{*}_{1} \\ \vdots \\  \psi^{*}_{N} \end{pmatrix*}.
\end{equation}
We use an ``H'' to denote the Hermitian transpose of a matrix, such that
\begin{align}
	\nonumber \bm{\psi}^{\H} = (\bm{\psi}^{*})^{\T}& = \begin{pmatrix*}[c]  \psi_{1}^{*} \cdots \psi_{N}^{*} \mid  \psi_{1} \cdots \psi_{N} \end{pmatrix*}, \\ \hat{\mathbf{a}}^{\H} = (\hat{\mathbf{a}}^{\dag})^{\T} & = \begin{pmatrix*}[c]  \hat{a}_{1}^{\dag} \cdots \hat{a}_{N}^{\dag} \mid \hat{a}_{1} \cdots \hat{a}_{N} \end{pmatrix*}. 
\end{align}
With this notation, we can write our mode decomposition concisely as
\begin{equation}
	\hat\phi = \frac{\pi}{L} \hat{\mathbf{a}}^{\H} \bm{\Gamma} \bm{\psi} = \frac{\pi}{L}  \bm{\psi}^{\H} \bm{\Gamma} \hat{\mathbf{a}},
\end{equation}
where
\begin{equation}
	\bm{\Gamma} =  \begin{pmatrix*}[r]  \mathbf{0} & \mathbf{I} \\  \mathbf{I} & \mathbf{0}  \end{pmatrix*}, \quad \bm{\Gamma}^{2} = \mathbf{I},
\end{equation}
and $\mathbf{I}$ is the identity matrix.  If we now define
\begin{equation}
	(\hat\phi,\bm{\psi}^{\H}) = \begin{pmatrix*}[c] (\hat\phi,\psi_{1}^{*}) \cdots (\hat\phi,\psi_{N}^{*}) \mid (\hat\phi,\psi_{1})  \cdots (\hat\phi,\psi_{N}) \end{pmatrix*},
\end{equation}
then we have
\begin{equation}
	(\hat\phi,\bm{\psi}^{\H}) = \hat{\mathbf{a}}^{\H} \bm{\Gamma} \mathbf{S},
\end{equation}
where the overlap matrix $\mathbf{S}$ is defined as
\begin{equation}\label{eq:S def}
\textbf{S} = \frac{\pi}{L} (\bm\psi,\bm\psi^{\H}) =  \frac{\pi}{L}  \begin{pmatrix*}[c] [(\psi_{n},\psi_{m}^{*})]  &  [(\psi_{n},\psi_{m})]  \\   [(\psi_{n}^{*},\psi_{m}^{*})]  & [(\psi_{n}^{*},\psi_{m})]  \end{pmatrix*}.
\end{equation}
For an orthogonal basis satisfying (\ref{eq:algebra 0}), the overlap matrix reduces to
\begin{equation}\label{eq:orthogonal}
	\mathbf{S} = \bm{\Sigma} \bm{\Gamma}, \quad \bm{\Sigma} =   \begin{pmatrix*}[r]  \mathbf{I} & \mathbf{0} \\ \mathbf{0} & -\mathbf{I}  \end{pmatrix*}, \quad \bm{\Sigma}^{2} = \mathbf{I}.
\end{equation}

Now, assuming $\mathbf{S}$ is invertible\footnote{We comment that part of the reason for introducing a finite cutoff in our mode expansion is that the issue of the invertibility is considerably more complicated in the infinite dimensional case.  If we take the matrix dimension $2N$ to be infinite \emph{a priori}, it is possible that only one of the left- or right-inverses of $\mathbf{S}$ might exist.  In this case, the statement that $\mathbf{S}$ is invertible means that both the left- and right-inverses exist and are equal.} and noting that $(\hat\phi,\bm{\psi}^{\H})^{\H} = -(\hat\phi,\bm{\psi})$, we have
\begin{align}\label{eq:a solutions}
	\hat{\mathbf{a}} = - \bm{\Gamma} (\mathbf{S}^{-1})^{\H} (\hat\phi,\bm{\psi}), \quad \hat{\mathbf{a}}^{\H} = (\hat\phi,\bm{\psi}^{\H})  \mathbf{S}^{-1} \bm{\Gamma}.
\end{align}
From (\ref{eq:equal time}) it follows that if $\varphi$ and $\psi$ are classical scalar quantities, then
\begin{equation}
	[ (\hat\phi, \varphi), (\hat\phi, \psi^{*}) ] = (\varphi^{*},\psi^{*})=-(\psi,\varphi).
\end{equation}
Using this and (\ref{eq:a solutions}), we find the following commutation relations for $\{\hat{a}_{n},\hat{a}_{n}^{\dag}\}$:
\begin{equation}\label{eq:algebra 1}
	[\hat{\mathbf{a}},\hat{\mathbf{a}}^{\H}] = \frac{L}{\pi} \mathbf{R}^{-1},
\end{equation}
where the (assumed invertible) matrix $\mathbf{R}$ is defined by
\begin{equation}\label{eq:R def}	
	\textbf{R} = -\mathbf{S}^{\H}\bm{\Gamma} = \frac{\pi}{L}  \begin{pmatrix*}[c] [(\psi_{n},\psi_{m})]  &  [(\psi_{n}^{*},\psi_{m})]  \\   [(\psi_{n},\psi_{m}^{*})]  & [(\psi_{n}^{*},\psi_{m}^{*})]\end{pmatrix*}^{*}.
\end{equation}
Now, if the basis $\{\psi_{n},\psi_{n}^*\}$ satisfies (\ref{eq:algebra 0}) then $\mathbf{R} = \mathbf{\Sigma}$.  This, in turn, implies
\begin{equation}\label{eq:algebra 2}
	\mathbf{R} = \bm{\Sigma} \quad \Leftrightarrow \quad [\hat{a}_{n},\hat{a}_{m}^{\dag}]  = \frac{L}{\pi} \delta_{nm}, \quad [\hat{a}_{n},\hat{a}_{m}] = 0.
\end{equation}
Therefore, $\{\hat{a}_{n},\hat{a}_{n}^{\dag}\}$ will only satisfy the algebra of creation and annihilation operators if $\{\psi_{n},\psi_{n}^*\}$ satisfies (\ref{eq:algebra 0}).

From the formula (\ref{eq:classical Hamiltonian}), we obtain the following expression for the Hamiltonian operator
\begin{equation}\label{eq:Hamiltonian 0}
	\hat{H} = \frac{\pi^{2}}{2L^{2}} \hat{\textbf{a}}^{\H} \bm{\Omega} \hat{\textbf{a}},
\end{equation}
with
\begin{equation}\label{eq:Omega def}
\bm{\Omega} = - i\bm{\Gamma} (\bm{\psi},\di_{t}\bm{\psi}^{\H})= -i \begin{pmatrix*}[c]  [(\psi_{n}^{*},\dot\psi_{m}^{*})] & [ (\psi_{n}^{*},\dot\psi_{m})] \\  [(\psi_{n},\dot\psi_{m}^{*})] & [(\psi_{n},\dot\psi_{m})] \end{pmatrix*}.
\end{equation}
Note that $\bm{\Omega}$ is Hermitian (i.e., $\bm{\Omega}^{\H} = \bm{\Omega}$), which means that $\hat{H}$ is self-adjoint, as required.

We conclude this section by giving formulae for $\mathbf{R}$, $\mathbf{R}^{-1}$ and $\bm{\Omega}$ for the Fourier-mode functions introduced in Section \ref{sec:Fourier}:
\begin{subequations}\label{eq:R and Omega formulae}
\begin{gather}
	\nonumber \textbf{R} =  \bm{\Sigma} - \begin{pmatrix*}[c]  \bm\sigma & \bm\sigma \\  \bm\sigma & \bm\sigma  \end{pmatrix*},  \quad \mathbf{R}^{-1} =  \bm{\Sigma} - \begin{pmatrix*}[r] - \bm\sigma & \bm\sigma \\  \bm\sigma & -\bm\sigma  \end{pmatrix*}, \\ \bm{\Omega} =   \begin{pmatrix*}[c]  \bm\rho & \bm\xi \\  \bm\xi & \bm\rho  \end{pmatrix*}.
\end{gather}
\end{subequations}
The entries of the various sub-matricies appearing above are:
\begin{gather}
	\nonumber \sigma_{nm}  = \sigma^{*}_{mn} = \begin{cases} 0, & n = m, \\ \displaystyle \frac{2iv\sqrt{nm}[1-(-1)^{n+m}]}{\pi\sqrt{u_{n} u_{m}}(m^{2}-n^{2})}, & n \ne m,  \end{cases}  \\ \rho_{nm} =  n u_{n} \varepsilon_{n} \delta_{nm} , \quad \xi_{nm}  =  - n u_{n} (1-\varepsilon_{n}) \delta_{nm}. \label{eq:rho and xi components}
\end{gather}
where $u_{n}$ is the mean phase velocity defined in equation (\ref{eq:phase vel}), and
\begin{equation}
	\varepsilon_{n} = \begin{cases} 1, & F^{2}(k_{n})-v^{2}k_{n}^{2}>0,  \\ 0, & F^{2}(k_{n})-v^{2}k_{n}^{2}<0. \end{cases}
\end{equation}
Note that for a cavity with subsonic velocity (i.e.\ all modes are ``fast''), $\varepsilon_{n} = 1$ for all $n$ and $\bm{\xi} = 0$.

\section{Bogoliubov transformation}\label{sec:Bogliobov}

The commutator algebra (\ref{eq:algebra 1}) and Hamiltonian (\ref{eq:Hamiltonian 0}) along with the matrix definitions (\ref{eq:R def}) and (\ref{eq:Omega def}) completely fix the quantum dynamics of the system.  However, if $\mathbf{R} \ne \bm{\Sigma}$ then the $\{\hat{a}_{n},\hat{a}_{n}^{\dag}\}$ operators cannot be used to conduct a Fock basis; i.e., we cannot write the Hamiltonian as a function of a number operator that corresponds to the occupation number of a given mode.  In this section, we therefore attempt to find a Bogoliubov transformation from $\{\hat{a}_{n},\hat{a}_{n}^{\dag}\}$ to a new set of operators $\{\hat{b}_{n},\hat{b}_{n}^{\dag}\}$ that do satisfy the creation/annihilation algebra.  We also attempt to impose the constraint that the Hamiltonian is diagonal in this new operator basis.  Our treatment loosely follows the formalism of \citet{2009arXiv0908.0787X}, which treats the problem of diagonalizing a Hamiltonian expressed in an orthogonal basis.

Our Bogoliubov transformation is explicitly given by
\begin{equation}
	\hat{a}_{n} = \sum_{m=1}^{N} ( A_{nm} \hat{b}_{m} + B_{nm} \hat{b}_{m}^{\dag}), \quad \hat{\mathbf{a}} = \mathbf{T}  \hat{\mathbf{b}},
\end{equation}
with
\begin{equation}\label{eq:T structure 1}	
	\mathbf{T} =  \begin{pmatrix*}[l]  \mathbf{A} & \mathbf{B} \\  \mathbf{B}^{*} & \mathbf{A}^{*}  \end{pmatrix*}.
\end{equation}
Here $A_{nm}$ and $B_{nm}$ are the Bogoliubov coefficients.  We note that equation (\ref{eq:T structure 1}) is equivalent to demanding
\begin{equation}\label{eq:T structure 2}
	\mathbf{T} = \bm{\Gamma} \mathbf{T}^{*} \bm{\Gamma}.
\end{equation}
A matrix satisfying either (\ref{eq:T structure 1}) or (\ref{eq:T structure 2}) is said to be a ``Bogoliubov-Valatin'' (BV) transformation.  Now, (\ref{eq:algebra 1}) gives
\begin{equation}
	 \mathbf{T} [\hat{\mathbf{b}},\hat{\mathbf{b}}^{\H}] \mathbf{T}^{\H}= \frac{L}{\pi} \mathbf{R}^{-1} .
\end{equation}
We demand that 
\begin{equation}\label{eq:T condition 1}
	\mathbf{T}^{\H} \mathbf{R} \mathbf{T} = \bm{\Sigma},
\end{equation}
which implies that  $\{\hat{b}_{n},\hat{b}_{n}^{\dag}\}$ satisfies the creation/annihilation algebra:
\begin{equation}\label{eq:b algebra}
	[\hat{\mathbf{b}},\hat{\mathbf{b}}^{\H}] = \frac{L}{\pi} \bm{\Sigma}.
\end{equation}
Now, in terms of the $\hat{\mathbf{b}}$ operators, the Hamiltonian is
\begin{equation}\label{eq:Hamiltonian 1}
	\hat{H} = \frac{\pi^{2}}{2L^{2}} \hat{\textbf{b}}^{\H} \mathbf{T}^{\H} \bm{\Omega} \mathbf{T} \hat{\textbf{b}}.
\end{equation}
Equation (\ref{eq:T condition 1}) implies that the inverse of $\mathbf{T}$ exists and is given by $ \mathbf{T}^{-1} = \bm{\Sigma} \mathbf{T}^{\H} \mathbf{R}$, which means this can be re-written as
\begin{equation}\label{eq:Hamiltonian 2}
	\hat{H} = \frac{\pi^{2}}{2L^{2}} \hat{\textbf{b}}^{\H} \bm{\Sigma} \mathbf{T}^{-1} \mathbf{D} \mathbf{T} \hat{\textbf{b}},
\end{equation}
where
\begin{equation}
	\mathbf{D} = \mathbf{R}^{-1} \mathbf{\Omega}.
\end{equation}
Let us now make the further assumption that we can choose $\mathbf{T}$ such that it diagonalizes the $\mathbf{D}$ matrix in the following manner:
\begin{equation}
	 \mathbf{T}^{-1} \mathbf{D} \mathbf{T}  = \text{diag}(\mu_{1},\ldots,\mu_{N},-\mu_{1},\ldots,-\mu_{N}),
\end{equation}
where $\mu_{1} \ldots \mu_{N}$ are positive numbers.  Then, the Hamiltonian and operator algebra take the familiar form
\begin{equation}\label{eq:diagonal H}
	\hat{H} = \frac{\pi^{2}}{2L^{2}} \sum_{n=1}^{N} \mu_{n} ( \hat{b}_{n}^{\dag} \hat{b}_{n} + \hat{b}_{n} \hat{b}_{n}^{\dag}  ), \quad  [\hat{b}_{n},\hat{b}_{m}^{\dag}]  = \frac{L}{\pi} \delta_{nm},
\end{equation}
with all other $\hat{b}_{n}$ commutators vanishing.  A number operator that returns the number of field quanta associated with the $n^\text{th}$ mode for a given quantum state is defined in the usual way: $\hat{\mathcal{N}}_{n} = (\pi/L) \hat{b}_{n}^{\dag} \hat{b}_{n}$.  Making note of (\ref{eq:Lorentz contraction}) and (\ref{eq:proper time Hamiltonian}), we find that the proper time Hamiltonian is given by
\begin{equation}\label{eq:normal mode H}
	\hat{H}_{\tau} = \sum_{n=1}^{N} \res_{n} ( \hat{\mathcal{N}}_{n} + \tfrac{1}{2} ), \quad \res_{n} = \frac{\pi\mu_{n}}{L_{0}(1-v^{2})}.
\end{equation}
We call the $\res_{n}$ the ``normal mode'' frequencies of the cavity due to the similarity of (\ref{eq:normal mode H}) and the Hamiltonian of a mechanical mass-spring system expressed in normal coordinates.  The energy levels (i.e.\ eigenvalues of $\hat{H}_{\tau}$) of the system will be the sum of the ground state energy plus integer multiples of the normal mode frequencies (as measured by an observer comoving with the cavity).

To summarize, if there exists a $2N$-dimensional square matrix $\mathbf{T}$ such that
\begin{enumerate}[label={[\roman*]}]
	\item \label{eq:condition 2} $\mathbf{T}^{\H} \textbf{R} \mathbf{T} = \bm{\Sigma}$,
	\item \label{eq:condition 0} $\mathbf{T} = \bm{\Gamma} \mathbf{T}^{*} \bm{\Gamma}$, and
	\item \label{eq:condition 1} $\mathbf{T}^{-1}\mathbf{D}\mathbf{T}  =  \text{diag}(\mu_{1},\ldots,\mu_{N},-\mu_{1},\ldots,-\mu_{N}),$
\end{enumerate}
then we can write the Hamiltonian in the diagonal form (\ref{eq:diagonal H}) via a Bogliobov transformation.  In such cases, we call $\mathbf{T}$ a ``normal mode transformation''.  Furthermore, if the above conditions hold for all $N$, then in the limit $N\rightarrow \infty$ we will have
\begin{equation}
	\hat{H}_{\tau} = \sum_{n=1}^{\infty} \res_{n} ( \hat{\mathcal{N}}_{n} + \tfrac{1}{2} ), \quad \res_{n} = \frac{\pi\mu_{n}}{L_{0}(1-v^{2})},
\end{equation}
where $\{\pm\mu_{n}\}$ are the eigenvalues of the infinite matrix $\mathbf{D}_{\infty} = \lim_{N\rightarrow\infty} \mathbf{D}$.  In Appendix \ref{sec:existence}, we demonstrate that such a normal mode transformation exists if and only if $\bm\Omega$ is positive definite for all $N$.

We conclude this section by noting that, in addition to the $\hat{\mathbf{b}}$ operators, $\mathbf{T}$ can be used to define normal mode functions $\{\varphi_{n},\varphi^{*}_{n}\}$.  We write
\begin{equation}
	\bm{\varphi} = \begin{pmatrix*}[c]  \varphi_{1} \\ \vdots \\ \hline    \varphi^{*}_{1} \\ \vdots  \end{pmatrix*} = \bm{\Gamma} \mathbf{T}^{\H} \bm{\Gamma} \bm{\psi},
\end{equation}
which implies that
\begin{equation}
	\hat\phi = \frac{\pi}{L} \hat{\mathbf{b}}^{\H} \bm{\Gamma} \bm{\varphi}.
\end{equation}
Repeating the calculations of section \ref{sec:quantization} with this mode expansion and noting both (\ref{eq:b algebra}) and (\ref{eq:diagonal H}), we see that the normal mode functions satisfy:
\begin{gather}\nonumber
	 (\varphi_{n},\varphi_{m})= (L/\pi) \delta_{nm} = - (\varphi_{n}^{*},\varphi_{m}^{*}), \quad (\varphi_{n},\varphi_{m}^{*}) = 0, \\ (\varphi_{n},\dot\varphi_{m}) = i\mu_{n}\delta_{nm}, \quad  (\varphi_{n}^{*},\dot\varphi_{m}) = 0.
\end{gather}
From this, it follows that:
\begin{equation}
	\left( \varphi_{n}, \frac{d\varphi_{m}}{d\tau} + i\omega_{m}\varphi_{m} \right) = \left( \varphi_{n}^{*}, \frac{d\varphi_{m}}{d\tau} + i\omega_{m}\varphi_{m} \right) = 0,
\end{equation}
for all $n$ and $m$.  Since $\{\varphi_{n},\varphi_{n}^{*}\}$ form a complete basis, this means
\begin{equation}
	\frac{d\varphi_{m}}{d\tau} + i\omega_{m}\varphi_{m} = 0 \implies \varphi_{m} = e^{-i\omega_{m}\tau} \Phi_{m}(x),
\end{equation}
where $\Phi_{m}(x)$ is a function determined by the details of the normal mode transformation or, equivalently, a solution of the ODE boundary value problem
\begin{gather}\nonumber
	 \left(-\frac{i \omega_{m}}{\gamma}+v \di_{x}\right)^{2}\Phi_{m}-\hat{F}^{2}(\di_{x})\Phi_{m} = 0, \\ \di^{(2i)}_{x} \Phi_{m} |_{x=0,L} = 0.\label{eq:ODE BVP}
\end{gather}
In other words, the $\varphi_{m}$ functions oscillate sinusoidally with frequency $\omega_{m}$ according to an observer comoving with the cavity.  This is consistent with the classical behaviour of a normal mode, and hence provides additional justification for calling $\omega_{m}$ the ``normal mode'' frequencies of the system.

\section{Resonant frequencies for subsonic cavity velocity}\label{sec:subsonic}

In the last section, we saw how the spectrum of $\mathbf{D}$ in the $N \rightarrow \infty$ limit gives the energy eigenvalues of the system, also known as the normal mode frequencies, provided that a normal mode transformation exists.  In this section, we consider the situation when the cavity velocity is subsonic; or, equivalently, the case when all Fourier modes are considered to be ``fast''.  In this case, we have $\bm\xi = \bm{0}$, which means that $\bm\Omega$ is positive-definite.  By Lemma \ref{lem:Omega positive} in Appendix \ref{sec:existence}, this means that a normal mode transformation does indeed exist.  In this section, we calculate the normal mode frequencies using several different methods and assumptions.

\subsection{Zero velocity case}\label{sec:zero velocity}

When the cavity has zero velocity with respect to the preferred frame, the matrix $\mathbf{D}$ is diagonal and its positive eigenvalues are simply:
\begin{equation}
	\mu_{n} = \frac{\pi}{L_{0}} \left| F\left( \frac{n\pi}{L_{0}} \right) \right|.
\end{equation}	
This leads to normal mode frequencies of the form
\begin{equation}
	\omega_{n} = \left| F\left( \frac{n\pi}{L_{0}} \right) \right| = \frac{n\pi}{L_{0}} \left| 1 - \sum_{i=2}^{M} c_{i} (-1)^{i} \left( \frac{n\pi L_{\star}}{L_{0}} \right)^{2i-2} \right|.
\end{equation}
If we make the physical assumption that the cavity is much bigger than the exotic physics length scale $L_{0} \gg n L_{\star}$, then we have
\begin{equation}
	\omega_{n} \approx \omega_{n}^{(0)} + \delta\omega_{n}^\text{static} , 
\end{equation}
where
\begin{equation}
	\omega_{n}^{(0)} = \frac{n\pi}{L_{0}}, \quad \delta\omega_{n}^\text{static} = -c_{2}  [\omega_{n}^{(0)} L_{\star} ]^{2} \omega_{n}^{(0)}.
\end{equation}
Here, $\omega_{n}^{(0)}$ are the frequencies of the cavity in the Lorentz invariant case, and $\delta\omega_{n}^\text{static}$ is the leading order correction to the frequencies when the cavity is at rest compared to the preferred frame.
	
\subsection{High velocity limit}\label{sec:high v}

We first examine the limit in which $|v| \rightarrow 1$ while the proper cavity length $L_{0}$ is held fixed.  In this limit, the wavenumbers $k_{n}$ are divergent:
\begin{equation}
	k_{n} = \frac{n\pi}{L} = \frac{n\pi}{L_{0}\sqrt{1-v^{2}}} \rightarrow \infty.
\end{equation}
Recall that the highest order spatial derivative in the wave equation (\ref{eq:wave equation}) is $\di_{x}^{4M-2}$.  If $M > 1$, it follows that the mean phase velocity is given by
\begin{equation}\label{eq:phase velocity high v}
	u_{n} \approx  \frac{|c_{M}|}{2^{M-1}} \left( \frac{n\pi L_{\star}}{L_{0}} \right)^{2M-2} (1-|v|)^{1-M},
\end{equation}
for $|v| \rightarrow 1$.  Now, we can decompose the $\mathbf{D}$ matrix as
\begin{gather}\nonumber
	\mathbf{D} = \mathbf{D}_{0} + \mathbf{D}_{1}, \\ \mathbf{D}_{0} = \begin{pmatrix*}[r]  \bm\rho & \mathbf{0}  \\  \mathbf{0}   & -\bm\rho  \end{pmatrix*}  , \quad \mathbf{D}_{1} = \begin{pmatrix*}[r] \bm\sigma \bm\rho &  - \bm\sigma \bm\rho  \\  - \bm\sigma\bm\rho  & \bm\sigma\bm\rho  \end{pmatrix*}.
\end{gather}
Putting (\ref{eq:phase velocity high v}) into (\ref{eq:rho and xi components}), we see that $\|\bm\rho\| = \mathcal{O}[(1-|v|)^{1-M}]$ while $\|\bm\sigma\| = \mathcal{O}[(1-|v|)^{M-1}]$, where $\| \cdot \|$ indicates some suitable matrix norm.  Hence, we can take $\mathbf{D} \approx \mathbf{D}_{0}$.  Since $\mathbf{D}_{0}$ is diagonal its eigenvalues and eigenvectors are trivial and can be divided into two sets.  First, we have eigenvalue-eigenvector pairs:
\begin{equation}\label{eq:D0 eigenvalue 1}
	(\mu_{n},\mathbf{v}_{n}) = (nu_{n},\mathbf{e}_{n}), \quad n = 1 \ldots N,
\end{equation}
where $\mathbf{e}_{n}$ is the $n^\text{th}$ standard basis vector in $\mathbb{C}^{2N}$.  The second set of eigenvalue-eigenvector pairs are
\begin{equation}\label{eq:D0 eigenvalue 2}
	(-\mu_{n},\mathbf{u}_{n}) = (-nu_{n},\bm\Gamma\mathbf{e}_{n}), \quad n = 1 \ldots N.
\end{equation}
Neglecting the contribution of $\mathbf{D}_{1}$ to $\mathbf{D}$, we then find that the normal mode frequencies are given by
\begin{equation}\label{eq:high v omega}
	\omega_{n} \approx \frac{n\pi |c_{M}|}{2^{M}L_{0}} \left( \frac{n\pi L_{\star}}{L_{0}} \right)^{2M-2} (1-|v|)^{-M}, \quad M > 1.
\end{equation}
This result is independent of matrix size, so we can trivially take the $N \rightarrow \infty$ limit.  Therefore, the energy levels of the system \emph{diverge} like $(1-|v|)^{-M}$ for high velocities.

If all the coefficients appearing in (\ref{eq:power F}) are of order unity, we expect the phase velocity approximation (\ref{eq:phase velocity high v}) to hold for
\begin{equation}
	k_{n} L_{\star} = \frac{n\pi L_{\star}}{L_{0}\sqrt{1-v^{2}}} \gg 1.
\end{equation} 
Hence, this is also the condition for which we expect the high-$v$ normal mode formula (\ref{eq:high v omega}) to be applicable.  We can rewrite this in terms of the Lorentz invariant frequency $\omega_n^{(0)} = n\pi/L_{0}$:
\begin{equation}
	\frac{\omega_{n}^{(0)} L_{\star}}{\sqrt{1-v^{2}}} \gg 1.
\end{equation} 
We can further interpret this physically by noting that in this approximation, the $\mathbf{D}$ matrix is diagonal and therefore the Fourier modes of section \ref{sec:Fourier} are the normal modes of the cavity.  Then, the characteristic wavelength of each of these modes as measured in the preferred frame is
\begin{equation}\label{eq:wavelength}
	\lambda^\text{pref}_{n} = \frac{L_{0}\sqrt{1-v^{2}}}{2n},
\end{equation}
and the condition for the approximation (\ref{eq:high v omega}) to hold is simply that mode wavelength in the preferred frame is less than the exotic physics length scale:
\begin{equation}
	\lambda^\text{pref}_{n} \ll L_{\star}.
\end{equation}

\subsection{Small velocity limit}\label{sec:low v}

We now turn our attention to the small velocity case $|v| \ll 1$.  We note that the matrix $\mathbf{D}_{1}$ is proportional to $v$, so it makes sense to treat it as a perturbative correction to $\mathbf{D}_{0}$.  This allows us to find the approximate the eigenvalues of $\mathbf{D}$ using an algorithm highly analogous to time-independent perturbation theory in quantum mechanics.  However, there is one key difference: in quantum mechanics, one is trying to estimate the eigenvalues of a Hermitian linear operator (the Hamiltonian), whereas here the relevant operator $\mathbf{D}$ is not self-adjoint.  Fortunately, it is not difficult to generalize the quantum mechanical formula to non-Hermitian linear operators.

We write the eigenvalues of $\mathbf{D}$ as
\begin{equation}
	\nu_{n} = \nu_{n,0}+\nu_{n,1}+\nu_{n,2}+ \cdots,
\end{equation}
where $\nu_{n,0}$ indicates the eigenvalue of $\mathbf{D}_{0}$ associated with an eigenvector $\mathbf{w}_{n}$, $\nu_{n,1}$ is the first order perturbative correction due to $\mathbf{D}_{1}$, $\nu_{n,2}$ is the second order correction, etc.  The first order correction is given by
\begin{equation}
\nu_{n,1} = \mathbf{w}_{n}^{\H} \mathbf{D}_{1} \mathbf{w}_{n}.
\end{equation}
Now, just as in the high velocity case above, the eigenvectors and eigenvalues of $\mathbf{D}_{0}$ are given by (\ref{eq:D0 eigenvalue 1}) and (\ref{eq:D0 eigenvalue 2}).  Since the eigenvectors are essentially the standard basis vectors, we see that $\nu_{n,1}$ is just the $(n,n)$ component of $\mathbf{D}_{1}$.  But $\mathbf{D}_{1}$ has zeroes on its main diagonal, so $\nu_{n,1}=0$.

Turning to the second order correction, we have
\begin{equation}\label{eq:perturbative lambda}
\nu_{n,2} = \sum_{k \ne n} \frac{  (\mathbf{w}_{n}^{\H} \mathbf{D}_{1} \mathbf{w}_{k} )( \mathbf{w}_{k}^{\H} \mathbf{D}_{1} \mathbf{w}_{n})}{\nu_{n,0} -\nu_{k,0} }.
\end{equation}
After some algebra, this expression yields the following formula for the positive eigenvalues of $\mathbf{D}$:
\begin{multline}
 \mu_{n} = nu_{n} \left[ 1 +   \frac{16 v^{2}n^{2}}{\pi} \sum_{k \ne n} \frac{ k^{2}[1-(-1)^{n+k}]^{2} }{ (k^{2}-n^{2})^{2}(n^{2}u_{n}^{2} -k^{2}u_{k}^{2})} \right. \\ \left. + \mathcal{O}(v^{4}) \right]
\end{multline}
We note that this expression holds for all $N$ and that the mean phase velocity $u_{n}$ is a function of $v$.

For certain choices of dispersion relation, the sum appearing in (\ref{eq:perturbative lambda}) is analytically calculable in the $N \rightarrow \infty$ limit.  For example, if we take 
\begin{gather}\nonumber
	\hat{F}(\di_{x}) = \di_{x} - L_{\star}^{2} \di_{x}^{3}, \quad F(k) = k + L_{\star}^{2} k^{3}, \\ u_{n} = \sqrt{\left(1 + \frac{n^{2} \pi^{2} L_{\star}^{2}}{L_{0}^{2}}\right)^{2} - v^{2}},\label{eq:cubic-squared}
\end{gather}
and retain terms of order $v^{2}$ in (\ref{eq:perturbative lambda}) we can derive a complicated closed form expression for $\mu_{n}$.\footnote{Note that $F(k)/k > 1$ for the choice (\ref{eq:cubic-squared}), so we are guaranteed that all modes are subsonic.}  This expression becomes much simpler if we make the physical assumption that the cavity's rest length is much larger than the length scale of exotic physics $L_{0} \gg L_{\star}$, yielding
\begin{equation}\label{eq:perturbative eigenvalues}
	\mu_{n} = n(1-v^{2}) + \frac{\pi^{2}L_{\star}^{2}(1+7v^{2})n^{3}}{L_{0}^{2}} + \mathcal{O}\left( v^{4}, \frac{L_{\star}^{3}}{L_{0}^{3}} \right).
\end{equation}
From this, we obtain the normal mode frequencies
\begin{equation}
	\omega_{n} = \frac{n\pi}{L_{0}} \left[ 1 + \frac{\pi^{2} n^{2} L_{\star}^{2} (1+8v^{2})}{L_{0}^{2}} \right] + \mathcal{O}\left( v^{4}, \frac{L_{\star}^{3}}{L_{0}^{3}} \right).
\end{equation}
we can rewrite this as
\begin{equation}
	\omega_{n} \approx \omega_{n}^{(0)} + \delta\omega_{n}^\text{static} + \delta\omega_{n}^\text{dynamic},
\end{equation}
where $\omega_{n}^{(0)} = n\pi/L_{0}$ is the Lorentz invariant frequency, $\delta\omega_{n}^\text{static}$ is the leading order correction to the frequency of a stationary cavity as derived in section \ref{sec:zero velocity} (with $c_{2}=-1$), and
\begin{equation}
	\delta\omega_{n}^\text{dynamic} = 8v^{2} [ \omega_{n}^{(0)} L_{\star}]^{2}  \omega_{n}^{(0)},
\end{equation}
is the leading order velocity correction to the normal mode frequencies.  Hence, we see that both the leading order static and dynamic corrections to a given normal mode's frequency scale like $[ \omega_{n}^{(0)} L_{\star}]^{2}$.  
  
\subsection{Numerical approximation}

We can estimate the normal mode frequencies numerically using a variation of the Rayleigh-Ritz method.  The idea is to calculate the eigenvalues of $\mathbf{D}$ for matrix sizes $N \le N_{\max}$.  We then approximate the a subset of the eigenvalues of $\mathbf{D}_{\infty}$ by eigenvalues of $\mathbf{D}$ which appear to converge to fixed values as $N \rightarrow N_{\max}$. The associated eigenvectors are approximations to the eigenvectors of $\mathbf{D}_{\infty}$ that are ``mostly'' confined to the subspace spanned by the columns of $\mathbf{D}$.\footnote{Technically speaking, here ``columns of $\mathbf{D}$'' refer to the $\mathbb{C}^{\infty}$ vectors obtained from the actual columns of $\mathbf{D}$ by appropriately appending zeroes.} 

Physically, we can understand why this method works by recalling that by taking $N$ to be finite, we are essentially excluding initial data involving wavelengths $\lesssim L/N$.  Hence, the normal modes identified by this method are the solutions of the wave equation involving wavelengths $\ll L/N_{\max}$; i.e.\ ``infrared'' modes are best approximated by this algorithm.

In figure \ref{fig:converge}, we show how the smallest positive eigenvalues of $\mathbf{D}$ behave as the matrix size is increased for $\hat{F}$ given by (\ref{eq:cubic-squared}).\footnote{The spectrum of $\mathbf{D}$ is symmetric under $\mu \rightarrow -\mu$ for this choice of $\hat{F}$, so we do not show the negative eigenvalues.}  We see that for a rather moderate choice of matrix size ($N \sim 30$), the first ten eigenvalues appear to be approaching their asymptotic values.  In the plots that follow, we take $N = 200$ which implies that the lowest $\sim 25$ eigenvalues and well-converged.
\begin{figure}
\begin{center}
	\includegraphics[width=\columnwidth]{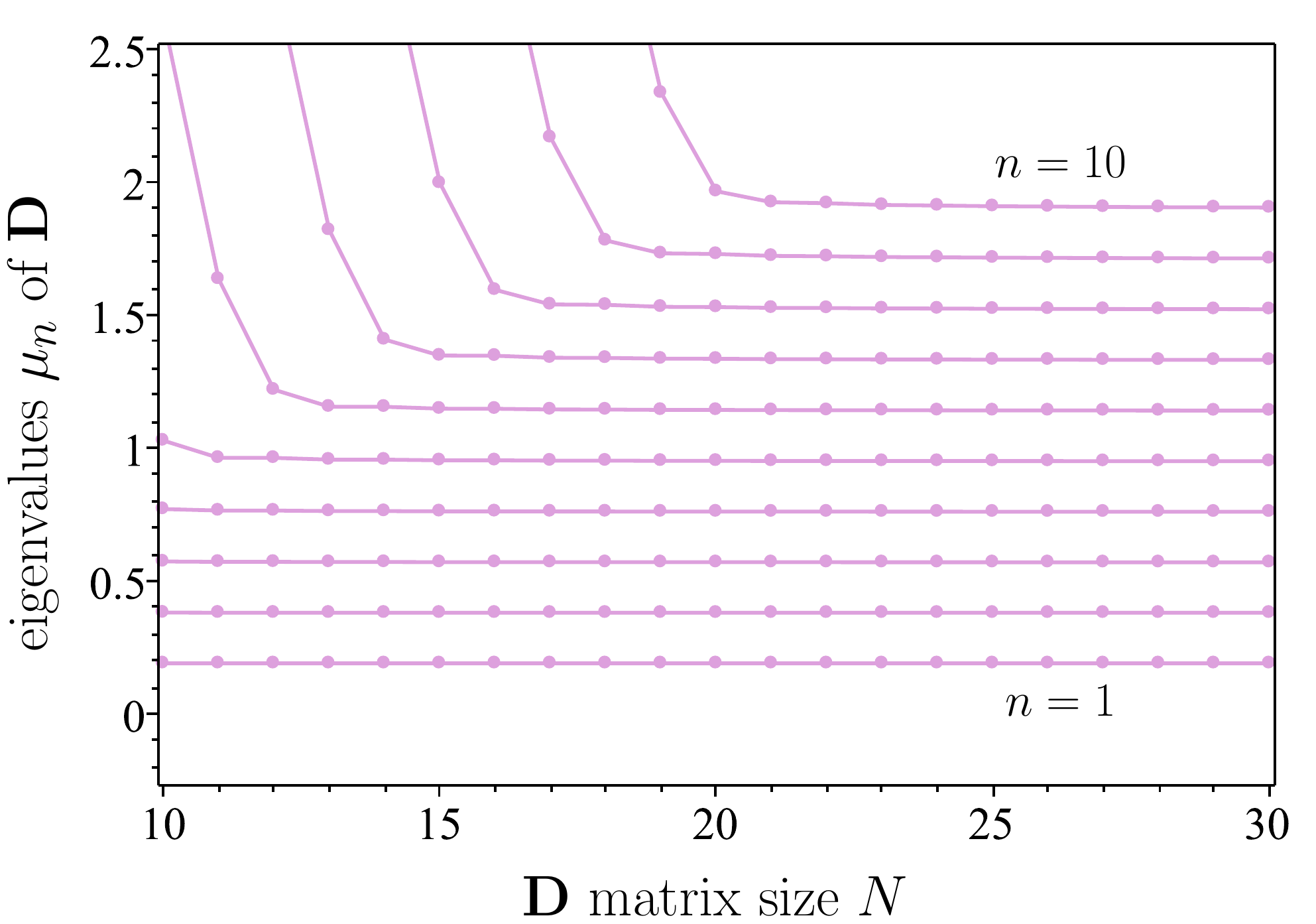}
\end{center}
\caption{Plot demonstrating the convergence of the lowest ten eigenvalues of $\mathbf{D}$ as $N$ is increased.  We have taken $\hat{F}$ to be the cubic polynomial given in (\ref{eq:cubic-squared}) and selected $L_{0}= 20 \pi L_{\star}$ with $v = 0.9$.}\label{fig:converge}
\end{figure}

In figure \ref{fig:epsilon=1}, we show the dependence of the eigenvalues and normal mode frequencies on the cavity velocity for the cubic $\hat{F}$ given in (\ref{eq:cubic-squared}).  We see that at low velocities and for large cavities ($v^{2}\ll 1$ and $L_{0} \gg L_{\star}$) the eigenvalues are $\approx n (1-v^{2})$ and the standard spectrum $\omega \approx n\pi/L_{0}$ is recovered, as expected from section \ref{sec:low v}.  However, for $v^{2} \rightarrow 1$ we see that the $\omega_{n}$ frequencies diverge, as predicted by the analytic analysis of section \ref{sec:high v}.  Qualitatively, the spectra associated with larger cavity resemble the standard case more strongly than the spectra of smaller cavities.
\begin{figure*}
\begin{center}
	\includegraphics[width=\textwidth]{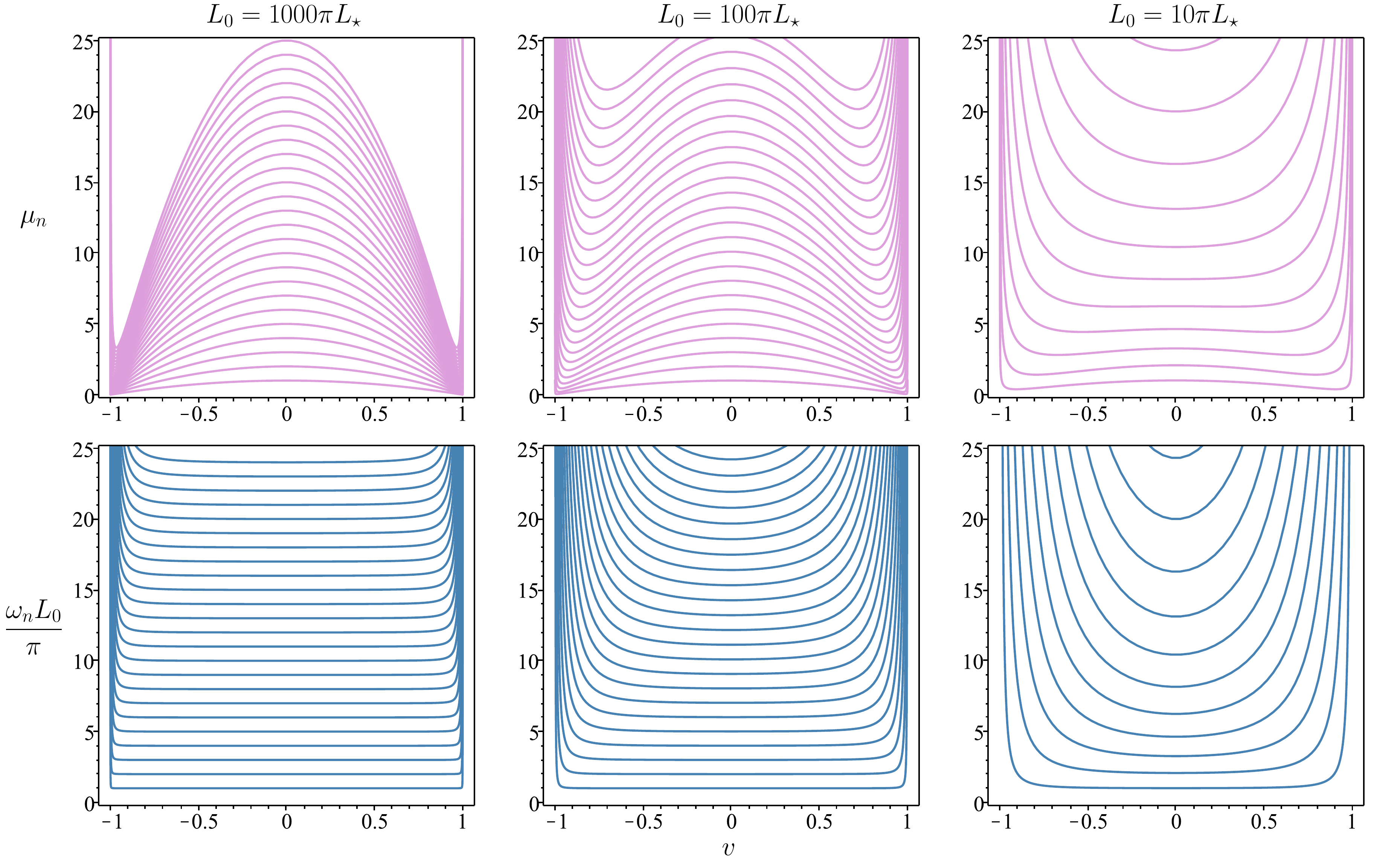}
\end{center}
\caption{Numerical approximations to the positive eigenvalues $\mu_{n}$ of $\mathbf{D}_{\infty}$ and the normal mode frequency $\omega_{n}$ as a function of cavity velocity $v$.  We have take $\hat{F}$ to the cubic polynomial given in (\ref{eq:cubic-squared}).}\label{fig:epsilon=1}
\end{figure*}

In figure \ref{fig:high-v}, we compare the high velocity behaviour of our numeric approximation to the results of section \ref{sec:high v}.  We see excellent agreement between the analytic and numerical results when the mode wavelength $\lambda_{n}^\text{pref}$ is less than $L_{\star}$.
\begin{figure}
\begin{center}
	\includegraphics[width=\columnwidth]{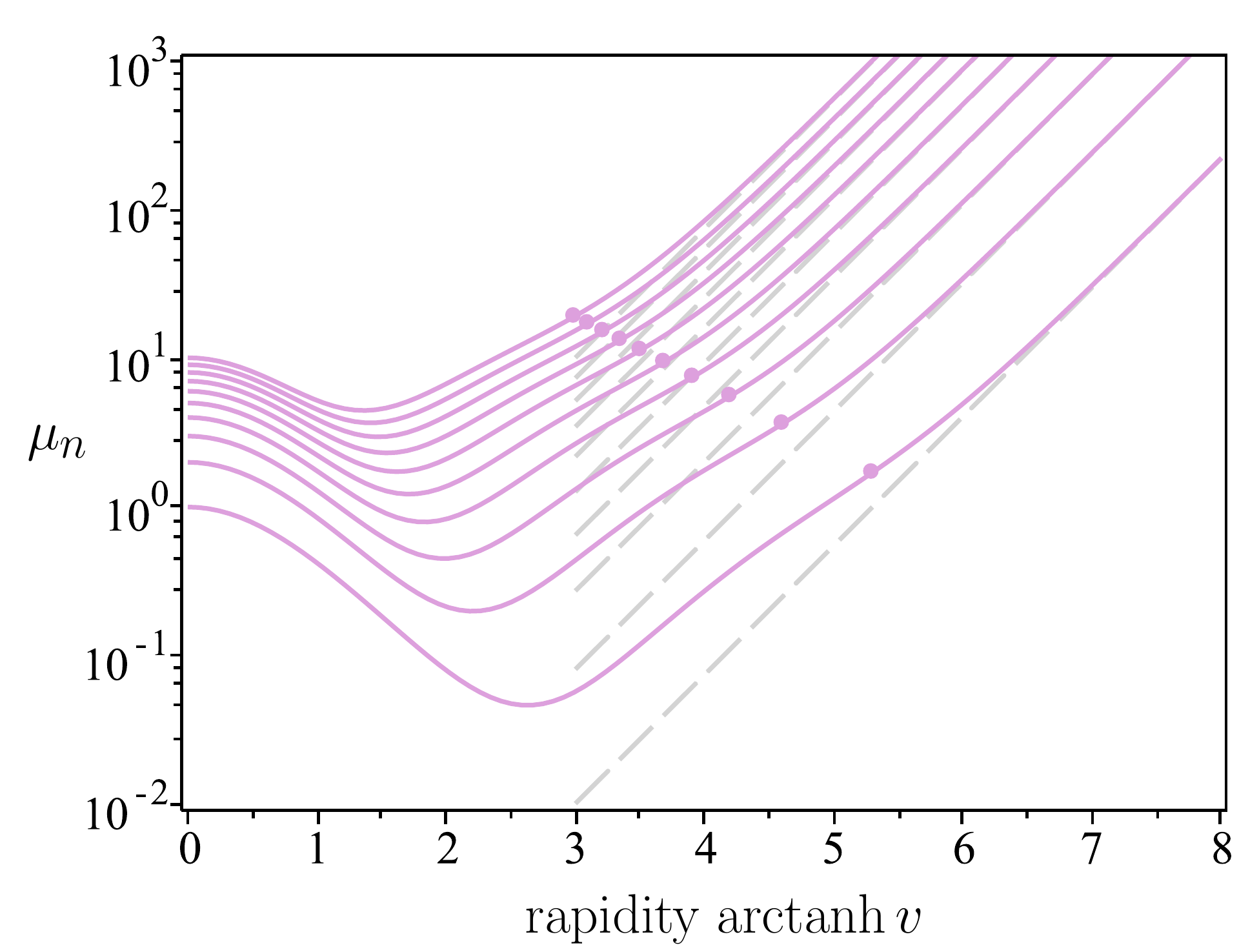}
	\caption{Numerical approximations to the positive eigenvalues $\mu_{n}$ of $\mathbf{D}_{\infty}$ as a function of cavity rapidity $v$ (\emph{solid lines}).  Also plotted is the analytic approximations for the high velocity eigenvalue behaviour derived in section \ref{sec:high v} (\emph{dashed lines}).  The squares on the solid lines where the characteristic wavelength of a given mode equals the exotic physics length scale.  This plot assumes that $\hat{F}$ is given by (\ref{eq:cubic-squared}) with $L_{0} = 100\pi L_{\star}$.}\label{fig:high-v}
\end{center}
\end{figure}

\subsection{Aside: recovering the standard result in the Lorentz invariant case}

Before leaving this section, we comment that all of the above formalism applies to the Lorentz invariant choice $\hat{F}(\di_{X}) = \di_{X}$.  We know that in this case, the normal mode frequencies as measured by comoving observers cannot depend on $v$ and must be given by
\begin{equation}
	\omega_{n} = \frac{n\pi}{L_{0}},
\end{equation}
which in turn implies that the positive eigenvalues of the $\mathbf{D}$ matrix have to be
\begin{equation}\label{eq:LI eigenvalues}
	\mu_{n} = n(1-v^{2}),
\end{equation}
in the $N \rightarrow \infty$ limit.

Now, we note that even when $\hat{F}(\di_{X}) = \di_{X}$, the matrix $\mathbf{D}$ has a non-trivial structure:
\begin{gather}\nonumber
	\mathbf{D} =  \sqrt{1-v^{2}} \begin{pmatrix*}[r]  \bm\aleph & \mathbf{0}  \\  \mathbf{0}   & -\bm\aleph  \end{pmatrix*} + \frac{4iv}{\pi} \begin{pmatrix*}[r] \bm\chi  &  - \bm\chi  \\  - \bm\chi  & \bm\chi  \end{pmatrix*},
\end{gather}
where
\begin{equation}
	\bm{\aleph} = \text{diag}(1,2,3\ldots),
\end{equation}
and the components of $\bm\chi$ are given by
\begin{equation}
	\chi_{nm} = \begin{cases}n^{1/2}m^{3/2}(m^{2}-n^{2})^{-1}, & (n-m) \in \text{odd}, \\ 0, & (n-m) \in \text{even}. \end{cases}
\end{equation}
Written out explicitly, the upper-left portion of the $\bm\chi$ matrix is
\begin{equation}
	\bm{\chi} = \begin{pmatrix*}[c]
	 0 & \frac{2\sqrt{2}}{3} & 0 & \frac{8}{15} & 0 & \cdots \\
	 -\frac{\sqrt{2}}{3} & 0 & \frac{3\sqrt{6}}{5} & 0 & \frac{5\sqrt{10}}{21} & \cdots \\
	 0 & - \frac{2\sqrt{6}}{5} & 0 & \frac{8\sqrt{3}}{7} & 0 & \cdots \\
	 -\frac{2}{15} & 0 & -\frac{6\sqrt{3}}{7} & 0 & \frac{10\sqrt{5}}{9} & \cdots \\
	 0 & - \frac{2\sqrt{10}}{21} & 0 & -\frac{8\sqrt{5}}{9} & 0 & \cdots \\
	 \vdots & \vdots & \vdots & \vdots & \vdots & \ddots
	\end{pmatrix*}
\end{equation}
From this, we see that $\mathbf{D}$ is a rather complicated matrix, and it is not at all clear how to analytically demonstrate that its positive eigenvalues are given by (\ref{eq:LI eigenvalues}).

Since we cannot analytically determine the eigenvalues of $\mathbf{D}_{\infty}$, we instead calculate them numerically.  The results are shown in figure \ref{fig:LI-case}.  We find that the numerical eigenvalues are indeed consistent with $\mu_{n} = n(1-v^{2})$, which is an important consistency check for both our analytic and numerical arguments.
\begin{figure}
\begin{center}
	\includegraphics[width=\columnwidth]{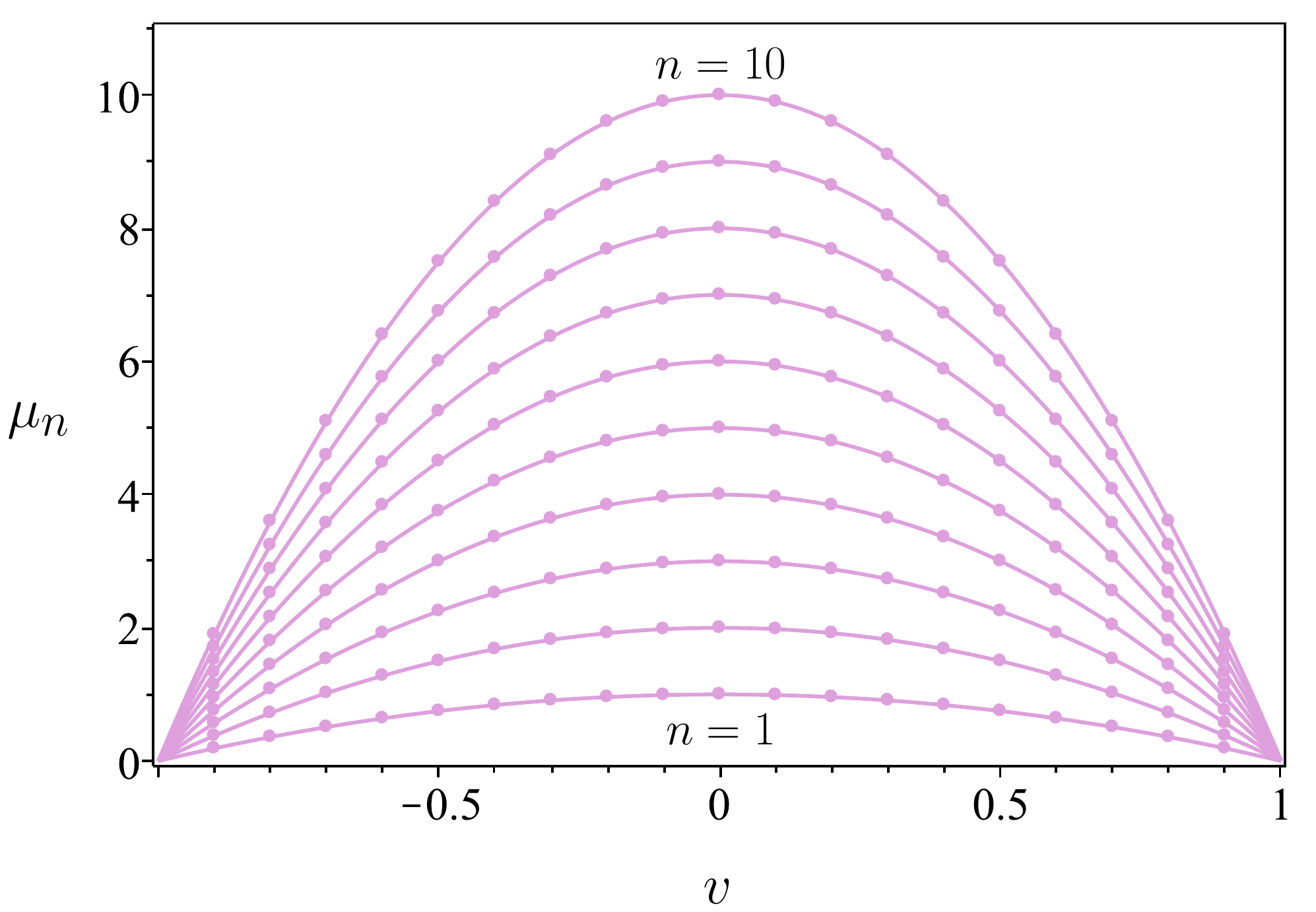}
	\caption{Numerical approximations to the positive eigenvalues $\mu_{n}$ of $\mathbf{D}_{\infty}$ in the Lorentz-invariant case $\hat{F} = \di_{X}$ (\emph{points}).  Shown for comparison are the curves $\mu_{n} = n(1-v^{2})$, which represent our expectations from Lorentz symmetry considerations (\emph{lines}).}\label{fig:LI-case}
\end{center}
\end{figure}

\section{Classical instabilities for supersonic cavity velocity}\label{sec:supersonic}

In this section, we consider the case of a supersonic cavity; i.e., one for which certain modes have phase speed less than the velocity of the cavity's walls (as measured in the preferred frame in the $L \rightarrow \infty$ limit).  In other words, the cavity's walls act as a sonic horizon for some Fourier modes.  In this case, the $\bm\Omega$ matrix given in (\ref{eq:R and Omega formulae}) is not diagonal or positive-definite, so we are guaranteed that a normal mode transformation does not exist by Lemma \ref{lem:5} in Appendix \ref{sec:existence}.

In fact, the classical Hamiltonian is not bounded from below in this case.  To see this, we can plug the classical Fourier mode expansion (\ref{eq:Fourier decomposition 1}) into (\ref{eq:classical Hamiltonian}) to obtain
\begin{equation}
	H = \frac{\pi^{2}}{2L^{2}} {\textbf{a}}^{\H} \bm{\Omega} \textbf{a}, 
\end{equation}
where $\mathbf{a}$ is the vector formed by the classical expansion coefficients $a_{n}$ in a manner similar to (\ref{eq:a vector}).  Let us define two sets of positive integers:
\begin{align}
	\mathbb{S} & = \left\{ n \in \mathbb{Z}^{+} \mid F^{2}(k_{n}) - v^{2} k_{n}^{2} \ge 0 \right\}, \\
	\bar{\mathbb{S}} & = \left\{ n \in \mathbb{Z}^{+} \mid F^{2}(k_{n}) - v^{2} k_{n}^{2} < 0 \right\} = \mathbb{Z}^{+} / \mathbb{S}.
\end{align}
The set $\mathbb{S}$ contains the integer labels of all the modes with mean phase velocity greater than $v$ (``fast'' modes), while the set $\bar{\mathbb{S}}$ contains all the integer labels of the modes with mean phase velocity less than $v$ (``slow'' modes).  Then, the classical Hamiltonian can be written as 
\begin{equation}
	H = \frac{\pi}{L} \sum_{n\in \mathbb{S}} \zeta_{n} |a_{n}|^{2} - \frac{\pi}{L} \sum_{n\in \bar{\mathbb{S}}} \zeta_{n} \text{Re}(a_{n}^{2}). 
\end{equation}
Recalling that $\zeta_{n} \ge 0$, we see that the first sum is associated with the amplitudes of the fast modes and is positive semi-definite.  On the other hand, the second sum is due to the slow modes modes and does not have a fixed sign.  Hence, the classical Hamiltonian is unbounded from below if there are any slow modes present ($\bar{\mathbb{S}} \ne \varnothing$); i.e. if the cavity velocity is supersonic.

\begin{figure*}
\begin{center}
	\includegraphics[width=\textwidth]{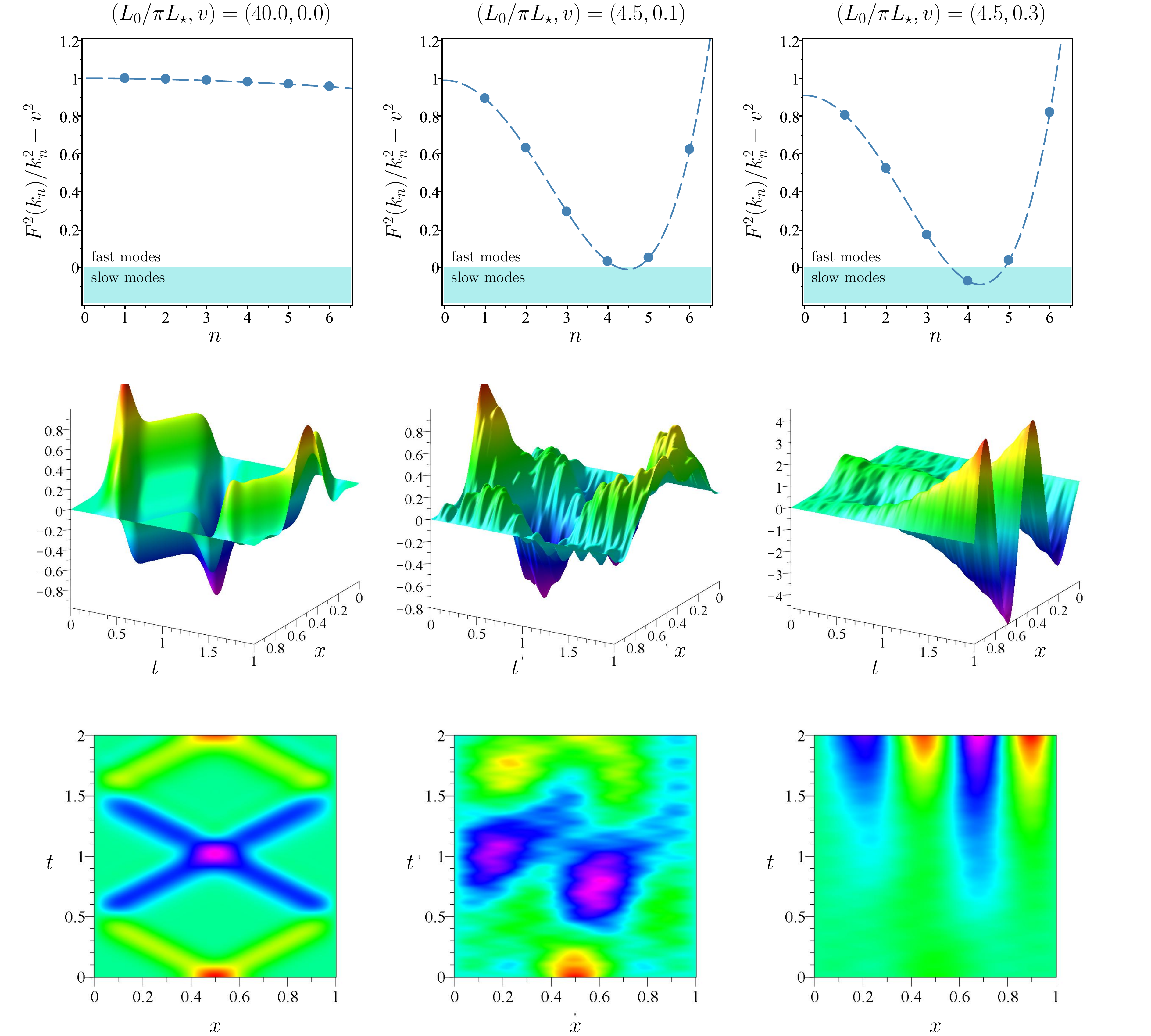}
\end{center}
\caption{Numeric simulations of the solution of the wave equation (\ref{eq:wave equation}) for the choice (\ref{eq:cubic-squared alt}).  The top panels show $F^{2}(k_{n})/k_{n}^{2} - v^{2}$ versus $n$ for the first few Fourier modes of the system and three different sets of parameter values.  The sign of this quantity determines whether or not a given mode is fast or slow.   If there exist any slow modes, the classical Hamiltonian to be unbounded from below and we expect the system to be classically unstable.  The lower panels show numeric simulations for the various parameter choices shown in the top row.  We indeed see that when the system supports even one slow mode, the numerical solution the wave equation appears to exhibit a classical instability.}\label{fig:simulations}
\end{figure*}
The unboundedness of $H$ would seem to imply a classical instability of the system.  To test this, we numerically solved the wave equation with the choice
\begin{gather}\nonumber
	\hat{F}(\di_{x}) = \di_{x} + L_{\star}^{2} \di_{x}^{3}, \quad F(k) = k - L_{\star}^{2} k^{3}, \\ u_{n} = \sqrt{\left| \left(1 - \frac{n^{2} \pi^{2} L_{\star}^{2}}{L_{0}^{2}}\right)^{2} - v^{2} \right| }.\label{eq:cubic-squared alt}
\end{gather}
Depending on the values of $v$ and $L_{0}$, the cavity may or may not possess slow modes.  In figure \ref{fig:simulations}, we plot $F^{2}(k_{n})/k_{n}^{2} - v^{2}$ versus $n$ for several parameter choice along with the associated simulation results.  We clearly see in the figure that if the cavity's velocity is subsonic for every mode (i.e.\ all modes are ``fast''), then the numerical evolution appears to be stable.  On the other hand, if even one mode has mean phase velocity less than the wall velocity (i.e.\ some modes are ``slow''), the numerical solution is exponentially growing.

Based on the unboundedness of the Hamiltonian and numerical simulations, we may be tempted to conclude that the cavity exhibits a classical instability whenever its wall velocity is supersonic in the preferred frame.  However, we caution that the circumstantial evidence presented here does not constitute a proof of this conjecture.

We conclude this section by comparing our results to the discussion in Ref.~\cite{Kostelecky:2000mm}.  In that paper, it was pointed out that if a field has a non-Lorentz-invariant dispersion relation such that the 4-momentum $k^{\alpha}$ of a given mode is spacelike, then then energy of that mode  $E = -k_{\alpha}u^{\alpha}$, as measured by an inertial observer with 4-velocity $u^{\alpha}$, can be either positive or negative.  In \cite{Kostelecky:2000mm}, frames in which the energy is positive or negative were termed ``concordant'' or ``non-concordant'', respectively.  The existence of negative energy modes in non-concordant frames would suggest an instability in the same way that the unboundedness of the Hamiltonian from below does in our system.  

To see how this intuition works in our model, we note that a plane wave solution of the wave equation in the preferred frame in the absence of cavity walls is
\begin{equation}
	\phi = \mathcal{C}e^{ i k_{\alpha}x^{\alpha} } = \mathcal{C}e^{ -i [|F(k)|t-kx]}, \quad k^{\alpha} = [|F(k)|,k,0,0].
\end{equation}
A sufficient condition for the 4-momentum to be spacelike (i.e.\ $k_{\alpha} k^{\alpha}>0$) is that $F^{2}(k) < k^{2}$.  In the preferred frame, the 4-velocity $u^{\alpha}$ of an inertial observer with velocity $v$ in the $x$-direction is given by
\begin{equation}
	 \quad u^{\alpha} = \gamma [1,v,0,0].
\end{equation}
The energy of the mode according this observer is then
\begin{equation}
	E = \gamma[|F(k)|-kv].
\end{equation}	
We see then that $E<0$ is only possible if 
\begin{equation}
	F^{2}(k) - k^{2} v^{2} < 0.
\end{equation}
But this inequality is exactly the condition that slow modes exist in a frame moving at speed $v$ with respect to the preferred frame.  Furthermore, since $v^{2} < 1$ the inequality can only be satisfied if mode's 4-momentum is spacelike.  So, like the authors of \cite{Kostelecky:2000mm}, we find that if there exist mode solutions of the wave equation with spacelike momenta and positive energy in a given frame, there exist other frames where the mode energy is negative and the Hamiltonian is unbounded from below.  In such non-concordant frames, it is not surprising to find a classical instability.  In some sense, we have found an explicit example that supports the general argument found in \cite{Kostelecky:2000mm}.

\section{Constraints from cavity experiments}\label{sec:experiments}

In this section, we consider how experiments an be used to put constraints on Lorentz violations based on the effects described above.  Several  scenarios are possible, but here we concentrate on one particularly precise experiment: namely, rotating cavities (see, for example, \cite{Herrmann:2009zzb}).  The key observation is that, in the presence of Lorentz violation of the form discussed in this paper, the normal mode frequencies of a cavity moving with 3-velocity $\vec{v}$ with respect to the preferred frame will be different if the cavity's boundaries are parallel or perpendicular to $\vec{v}$.  We will need to assume that the results presented in the previous section for scalar fields generalize to spin-1 fields, since actual experiments almost exclusively measure electromagnetic radiation.  This does not seem like a radical assumption based on experience from ordinary quantum field theory, but it is important to keep in mind that electromagnetic field theory is considerably more complicated in models that break Lorentz invariance.  Therefore, the discussion in this section is some what heuristic.  In this section, we take $\hat{F}$ to be of the form (\ref{eq:cubic-squared}) for concreteness.

\begin{figure}
\begin{center}
	\includegraphics[width=\columnwidth]{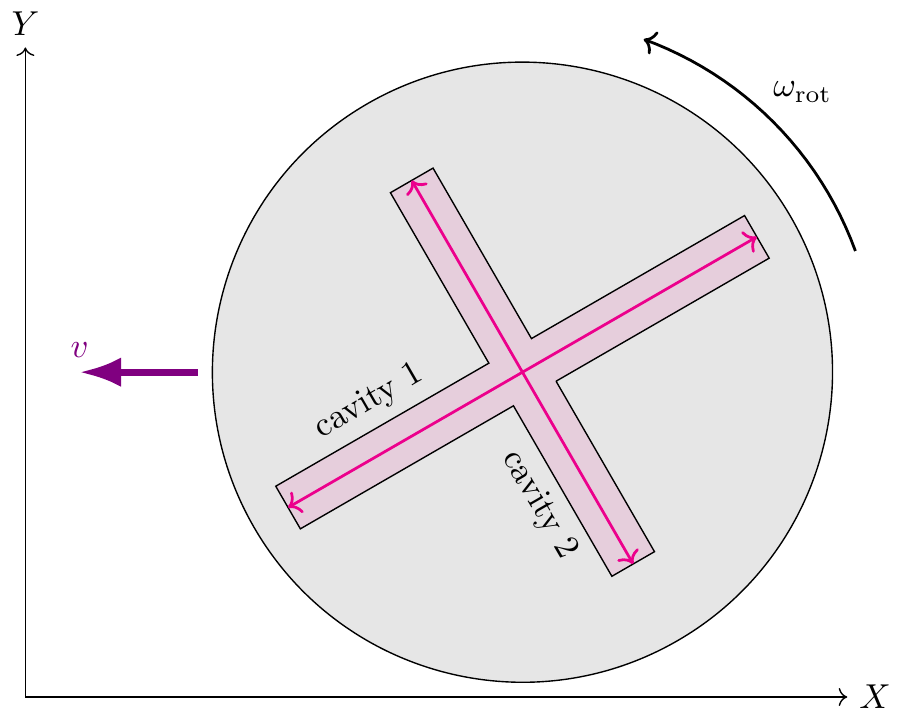}
\end{center}
\caption{A rotating cavity experiment}\label{fig:experiment}
\end{figure}
We are therefore motivated to consider the experimental set-up is shown in figure \ref{fig:experiment}.  As before, we denote cartesian coordinates in the preferred frame by $(X,Y,Z)$.  We consider two orthogonal cavities on a turntable rotating with frequency $\omega_\text{rot}$ in the $XY$-plane.  The centre of the turntable moves with speed $v$ in the negative $X$-direction.  It is easy to see that if $v \ll 1$ and we adopt the cubic-sqared dispersion relation (\ref{eq:cubic-squared}), the results of section \ref{sec:low v} give
\begin{align}
	\frac{\omega_{n}}{\omega_n^{(0)}} \approx 1 + [\omega_n^{(0)} L_{\star}]^{2} \times  \begin{cases} 1+8v^{2} ,  & \text{$X$ orientation}, \\   1 , &\text{$Y$ orientation}. 
\end{cases}
\end{align}
where the $\omega_{n}^{(0)} = n\pi/L_{0}$ are the normal mode frequencies in the absence of Lorentz violation.  Here, ``$X$ orientation'' and ``$Y$ orientation'' refer to the cases where the normals to the cavity boundaries are parallel to the $X$ and $Y$ axes, respectively.  We also assume that a pair of tuneable lasers are stabilized to normal modes of frequency $\omega_{1} > \omega_{2}$ in cavities 1 and 2 respectively.  The beat frequency between the two lasers $\Delta\omega = (\omega_{1} - \omega_{2})/2$ can then be easily measured by siphoning-off and combining a small amount of output from each laser.  As the table rotates, $\Delta\omega$ will be modulated with frequency $2\omega_\text{rot}$ as the cavities' orientation relative to the turnable velocity changes.  The amplitude $A$ of the modulation will be
\begin{equation}
	A \sim 4L_{\star}^{2}v^{2}(\omega_{1}+\omega_{2})(\omega_{1}^{2} - \omega_{1}\omega_{2} + \omega_{2}^{2}).
\end{equation}
Finally, we can make the simplifying assumption that the frequencies each laser are reasonably close to each other $\omega_{1} \approx \omega_{2} \approx \omega_{0}$:
\begin{equation}
	\frac{A}{\omega_{0}} \sim 8 \omega_{0}^{2}L_{\star}^{2}v^{2}.
\end{equation}

An terrestrial experiment of this type was done by \citet{Herrmann:2009zzb}, who found that $A/\omega_{0} \lesssim 10^{-16}$ using lasers tuned to frequency $\omega_{0} = 1.17\,\text{eV}$.  Assuming that the preferred frame is defined by the cosmic microwave background (CMB) and noting that the speed of the earth with respect to the CMB is $v \sim 2 \times 10^{-3}$, this gives $L_{\star}^{-1} \gtrsim 7 \times 10^{-4} \, \text{GeV}$.

We note that $L_{\star}$ for this dispersion relation can separately constrained by observations of gamma-ray bursts.  If we translate the results of \citet{Vasileiou:2013vra} into our notation, one finds the constraint $L_{\star}^{-1} \gtrsim 10^{10} \,\text{GeV}$ from astrophysical observations.  Clearly, this is a much more stringent result than can be obtained from the terrestrial experiment described above.

Is it possible to do better?  We saw in section \ref{sec:subsonic} that the relative changes to the resonant frequency of a given the cavity becomes significant (i.e. $\Delta \omega / \omega_{0} \gtrsim \mathcal{O}(1)$) if
\begin{equation}
	\frac{\omega_{0} L_{\star}}{\sqrt{1-v^{2}}} \gtrsim 1.
\end{equation}
Assuming $\omega_{0} L \ll 1$, this becomes
\begin{equation}
	|v| \gtrsim 1 - \tfrac{1}{2} \omega_{0}^{2} L_{\star}^{2} 
\end{equation}
Hence, in order to obtain better a better constraint on $L_{\star}$ than the one from gamma ray bursts from an experiment using a modes with $\omega_{0} \approx 1\,\text{eV}$, one would need to conduct the experiment in a lab travelling at a speed of $v \gtrsim 1 - 5 \times 10^{-39}$ with respect to the preferred frame.  Obviously, this represents a significant technical (and energetic) challenge if the preferred frame is the CMB.

\section{Discussion}\label{sec:discussion}

In this paper, we have considered scalar fields with non-Lorentz invariant dispersion relations confined in a quantum cavity.  If the walls of the cavity move with a speed exceeding the speed of the scalar modes, we have presented evidence that the system is classically unstable and, as mentioned in the introduction, a Bogoliubov transformation to a basis diagonalizing the Hamiltonian may not exist.  The fact that we are not guaranteed the existence of a Bogoliubov transformation diagonalizing the Hamiltonian, and in turn a well defined energy spectrum, is not a novel feature of our model; indeed similar phenomena can be found in the literature albeit in a different context. For example in \cite{Kostelecky:2000mm} the authors find that dispersion relations induced by Lorentz violation in the electromagnetic sector of SME can lead to negative energies in certain observer frames;  the nomenclature used for these frames is "non-concordant", where a "concordant frame" is one in which modifications due to Lorentz violations remain small.  Boosting to non-concordant frames leads to failure of the canonical quantization procedure and may lead to stability problems when interactions are introduced.  Additionally, violations of microcausality may also arise as one particle dispersion relations can develop group velocities exceeding $1$.  For further discussion on non-concordant frames see \cite{Kostelecky:2001xz,Lehnert:2002yi}; to see how these issues may be overcome via an extended Hamiltonian formalism see \cite{Colladay:2017bon,Colladay:2017cey}.

If however the walls move slower than the speed of all the scalar modes, we find that the energy spectrum of the cavity depends on the inertial velocity of the cavity with respect to a preferred frame in a non-trivial way.  Our results are significantly different from previously reported results for the mSME \cite{Muller:2004zp}, where the changes to a cavities resonant frequencies are effectively modelled by a velocity dependent index of refraction between the plates.  This implies that the spacing between adjacent resonant frequencies $\omega_{n+1} - \omega_{n}$ is independent of $n$ in the mSME (just as in standard theory).  In contrast, for our model involving higher dimensional operators,  $\omega_{n+1} - \omega_{n}$ will generally depend on $n$.  Because our results are directly tied to the existence of dimension 6 and higher operators in the effective action, the effects of Lorentz violation at low energies is exceedingly small if the speed of the cavity with respect to the preferred frame is not too large.  However, if the cavity speed approaches $1$, then the energy levels of the cavity become infinite.

This work can be extended in many ways:  To make the model more realistic, the behaviour of electromagnetic fields in a finite moving cavities should be considered.  It would also be interesting to determine the thermodynamic properties of a gas of quanta confined to a moving cavity following \cite{Husain:2013zda}; i.e., solve the blackbody problem for this system.  One can also extend the calculation of the Casimir effect in a cavity at rest in the preferred frame \cite{1751-8121-41-16-164021} to a cavity with non-zero velocity.  The dynamical Casimir effect, where the mismatch between velocity of the two cavity walls results in particle creation, could be interesting from both a theoretical and experimental point of view.  Finally, it would be useful to apply the formalism of this paper to the Unruh effect and black hole radiation.  In the former case, one would hope to get some physical insight into the main result presented in \citet{Husain:2015tna}, which was that certain Lorentz-violating dispersion relations lead to radically different behaviours for inertial particle detectors at low energy.

\appendix

\section{Existence of the Bogliobov transformation}\label{sec:existence}

In this appendix, we show that if $\bm\Omega$ is positive definite, then we can find a Bogliobov transformation $\mathbf{T}$ satisfying the conditions \ref{eq:condition 0}--\ref{eq:condition 2} from Section \ref{sec:Bogliobov}.  Here, the matrices $\mathbf{R}$ and $\mathbf{\Omega}$ are defined in equations (\ref{eq:R def}) and (\ref{eq:Omega def}), respectively.  Throughout this appendix, we assume all matrices are of finite size $2N \times 2N$, which allows us to make liberal use of elementary results from linear algebra.  Our analysis loosely follows the arguments presented in \citet{2009arXiv0908.0787X}, which deals with the Hamiltonian diagonalization problem given an already orthogonal basis.

\begin{lemma}\label{lem:basic results}

The matrices $\mathbf{R}$ and $\mathbf{\Omega}$ have the following properties
\begin{align}
	& \mathbf{R}=  \mathbf{R}^\text{\emph{H}}, \\ 
	& \mathbf{\Omega}=  \mathbf{\Omega}^\text{\emph{H}}, \\
	& \bm{\Gamma} \mathbf{R}^{*} \mathbf{\Gamma} = - \mathbf{R}, \label{eq:lem 1-3} \\
	& \bm{\Gamma} \mathbf{\Omega}^{*} \mathbf{\Gamma} = \mathbf{\Omega}.  \label{eq:lem 1-4}
\end{align}

\end{lemma}

\begin{proof}

These results can be verified from the definitions (\ref{eq:R def}) and (\ref{eq:Omega def}) making use of the identities (\ref{eq:identities}).

\end{proof}

\begin{definition}\label{def:inner}

The $\mathbf{R}$-inner product between two vectors is defined by
\begin{equation}
	\langle \mathbf{u} , \mathbf{v} \rangle = \mathbf{u}^{\H} \mathbf{R} \mathbf{v}.
\end{equation}

\end{definition}

\begin{lemma}\label{lem:eigenvectors}

If $\mathbf{v}_{n}$ is an eigenvector of $\mathbf{D}$ with eigenvalue $\mu_{n}$, then $\mathbf{u}_{n} = \mathbf{\Gamma} \mathbf{v}_{n}^{*}$ is also an eigenvector of $\mathbf{D}$ with eigenvalue $-\mu_{n}^{*}$.  Furthermore, the eigenvectors satisfy
\begin{equation}\label{eq:eigenvector inner products}
	\langle \mathbf{v}_{n},\mathbf{v}_{m}\rangle^{*} = -\langle \mathbf{u}_{n},\mathbf{u}_{m}\rangle, \quad \langle \mathbf{v}_{n}, \mathbf{u}_{m} \rangle = - \langle \mathbf{v}_{m}, \mathbf{u}_{n} \rangle.
\end{equation}

\end{lemma}
\begin{proof}

If $\mathbf{v}_{n}$ and $\mu_{n}$ are solutions to the eigenvalue problem for $\mathbf{D} = \mathbf{R}^{-1}\mathbf{\Omega}$, then it follows that
\begin{equation}
	(\mathbf{\Omega}-\mu_{n} \mathbf{R}) \mathbf{v}_{n} = 0.
\end{equation}
This can be re-arranged to read
\begin{equation}
	[\mathbf{\Gamma} (\mathbf{\Omega}-\mu_{n} \mathbf{R})^{*} \mathbf{\Gamma}]  (\mathbf{\Gamma}\mathbf{v}^{*}_{n})=0.
\end{equation}
Using (\ref{eq:lem 1-3}) and (\ref{eq:lem 1-4}), this can be re-cast as
\begin{equation}
	(\mathbf{\Omega}+\mu^{*}_{n} \mathbf{R}) \mathbf{u}_{n} = 0, \quad \mathbf{u}_{n} = \mathbf{\Gamma} \mathbf{v}_{n}^{*}.
\end{equation}
Hence, $\mathbf{u}_{n} = \mathbf{\Gamma} \mathbf{v}_{n}^{*}$ is an eigenvector of $\mathbf{D}$ with eigenvalue $-\mu^{*}_{n}$.  The relations (\ref{eq:eigenvector inner products}) follow directly from Lemma \ref{lem:basic results}, Definition \ref{def:inner}, and $\mathbf{u}_{n} = \mathbf{\Gamma} \mathbf{v}_{n}^{*}$.

\end{proof}

\begin{definition}\label{def:normal mode xform}

A normal mode transformation is a matrix $\mathbf{T}$ such that
\begin{enumerate}[label={\textup{[\roman*]}}]
	\item \label{eq:T condition} $\mathbf{T}^{\H} \textbf{\textup{R}} \mathbf{T} = \bm{\Sigma}$,
	\item  \label{eq:T condition 2} $\mathbf{T} = \bm{\Gamma} \mathbf{T}^{*} \bm{\Gamma}$, and
	\item \label{eq:T diagonal}  $\mathbf{T}^{-1}\mathbf{D}\mathbf{T}  = \textup{diag}(\mu_{1},\ldots,\mu_{N},-\mu_{1},\ldots,-\mu_{N}) \equiv \bm{\Lambda}$.
\end{enumerate}
Here,  $\mathbf{T}^{-1} = \mathbf{\Sigma} \mathbf{T}^\mathrm{H} \mathbf{R}$ and $\mu_{n} >0$.

\end{definition}

\begin{definition}

A set of vectors $S$ is mutually $\mathbf{R}$-orthogonal if for all $\mathbf{w}_{1},\mathbf{w}_{2} \in S$, $\mathbf{w}_{1} \ne \mathbf{w}_{2}$ implies $ \langle \mathbf{w}_{1} , \mathbf{w}_{2} \rangle = 0$.

\end{definition}

\begin{lemma}\label{lem:T} 

Suppose that

\begin{enumerate}[label={\textup{[\alph*]}}]
\item there exists a set of $N$ mutually $\mathbf{R}$-orthogonal eigenvectors of $\mathbf{D}$ given by ${B} = \{ {\mathbf{v}}_{n}, {\mathbf{u}}_{n} \}_{n=1}^{N}$ with ${\mathbf{u}}_{n} = \mathbf{\Gamma} {\mathbf{v}}_{n}^{*}$;
\item $\langle {\mathbf{v}}_{n} , {\mathbf{v}}_{n} \rangle = 1$; and
\item the eigenvalue associated with $ {\mathbf{v}}_{n} $ is positive.
\end{enumerate} 
Then there exists normal mode transformation $\mathbf{T}$ with the $\mu_{n}$ given by the eigenvalues of $\mathbf{D}$ associated the $ {\mathbf{v}}_{n} $ eigenvectors. 

\end{lemma}
\begin{proof}

We need to show that conditions \ref{eq:T condition}--\ref{eq:T diagonal} in Definition \ref{def:normal mode xform} hold with these assumptions.  To prove \ref{eq:T condition}, we note that since $\mathbf{R}$ is Hermitian, we have
\begin{equation}
	\langle \mathbf{w} , \mathbf{w} \rangle = \langle \mathbf{w} , \mathbf{w} \rangle^{*}.
\end{equation}
Then, by Lemma \ref{lem:eigenvectors}, we have $\langle {\mathbf{u}}_{n} , {\mathbf{u}}_{n} \rangle = -\langle {\mathbf{v}}_{n} , {\mathbf{v}}_{n} \rangle =  - 1$.

If we then construct $\mathbf{T}$ according to
\begin{equation}\label{eq:T soln}
\mathbf{T} = \begin{pmatrix*}[c] \big| & &\big|  & \\ \mathbf{v}_{1} & \cdots & \mathbf{u}_{1} & \cdots \\ \big| & &\big|  & \end{pmatrix*},
\end{equation}
then we see that \ref{eq:T condition} is satisfied.

The proof of \ref{eq:T condition 2} follows directly from the definition (\ref{eq:T soln}) and the the fact that $\mathbf{u}_{n} = \mathbf{\Gamma} \mathbf{v}_{n}^{*}$.

Finally, if we denote the positive eigenvalue associated with $\mathbf{v}_{n}$ as $\mu_{n}$, we have from Lemma \ref{lem:eigenvectors}
\begin{equation}
	\mathbf{D} \mathbf{T} = \mathbf{T} \mathbf{\Lambda}.
\end{equation}
Equation \ref{eq:T condition} implies that $\mathbf{T}^{-1}$ exists and is equal to $\mathbf{\Sigma}\mathbf{T}^\mathrm{H} \mathbf{R}$.  Hence, we have $\mathbf{T}^{-1} \mathbf{D} \mathbf{T} = \mathbf{\Lambda}$, which verifies \ref{eq:T diagonal}.

\end{proof}

\begin{lemma}\label{lem:Omega positive}

Suppose that $\bm{\Omega}$ is positive definite.  Then, a normal mode transformation exists.

\end{lemma}

\begin{proof}

Because $\mathbf{\Omega}$ is positive definite, we can write $\bm{\Omega} = \mathbf{F}^{\H} \mathbf{F}$ where $\mathbf{F}$ is an invertible matrix.  Since $\mathbf{F}$ is invertible, then $\mathbf{F}^{\H}$ is also invertible.  From this, it follows that $\textbf{D}$ is invertible and $\mathbf{D}^{-1} =  \mathbf{F}^{-1} (\mathbf{F}^{\H})^{-1}  \mathbf{R}$.  This implies that the spectrum of $\textbf{D}$ does not include zero.  

Hence, if $\mathbf{w}$ is an eigenvector of $\mathbf{D}$ with eigenvalue $\mu$, we can write:
\begin{equation}
	\mathbf{Q} \mathbf{y} = \mu^{-1} \mathbf{y}, \quad \mathbf{y} = \mathbf{F} \mathbf{x}, \quad \mathbf{Q} = (\mathbf{F}^{\H})^{-1} \mathbf{R} \mathbf{F}^{-1}.
\end{equation}
Since $\mathbf{Q}$ is obviously Hermitian, we are guaranteed that we can find linearly independent eigenvectors $\mathbf{y}_{n}$ of $\mathbf{Q}$ with real non-zero eigenvalues that are mutually orthogonal under the Euclidean inner-product.  This then implies that the eigenvectors of $\mathbf{D}$ (given by $\mathbf{w}_{n} = \mathbf{F}^{-1} \mathbf{y}_{n}$) satisfy relations
\begin{equation}
	\mathbf{w}_{n}^{\H} \bm{\Omega} \mathbf{w}_{m} = 0 \implies \langle \mathbf{w}_{n}, \mathbf{w}_{m} \rangle = 0,
\end{equation}
when $n \ne m$.  Therefore, $\mathbf{D}$ possesses a set $S$ of eigenvectors that are mutually orthogonal under the $\mathbf{R}$-inner product.  Furthermore, Lemma \ref{lem:eigenvectors} and the fact zero is not an eigenvalue of $\mathbf{D}$ implies that half of the eigenvectors in $S$ must be associated with positive eigenvalues, and the other half must be associated with negative eigenvalues. 

Now, if we denote the eigenvectors in $S$ with positive eigenvalues $\mu_{n}>0$ by $\tilde{\mathbf{v}}_{n}$, then we have
\begin{equation}
	 \langle \tilde{\mathbf{v}}_{n} , \tilde{\mathbf{v}}_{n} \rangle = \frac{\tilde{\mathbf{v}}_{n}^{\H} \mathbf{F}^{\H} \mathbf{F} \tilde{\mathbf{v}}_{n}}{\mu_{n} }.
\end{equation}
The righthand side is manifestly positive, so we conclude that $\mathbf{R}$-norm of $\tilde{\mathbf{v}}_{n}$ is positive.

Hence, we can construct a new set of eigenvectors $B = \{  \mathbf{v}_{n},  \mathbf{u}_{n} \}_{n=1}^{N}$ with $\mathbf{v}_{n} = {\tilde{\mathbf{v}}_{n}}/{\sqrt{\langle \tilde{\mathbf{v}}_{n} , \tilde{\mathbf{v}}_{n} \rangle}}$,
$\mathbf{u}_{n} = \mathbf{\Gamma} \mathbf{v}_{n}^{*}$, and
\begin{equation}
	 \langle {\mathbf{v}}_{n} , {\mathbf{v}}_{n} \rangle = 1.
\end{equation}
Furthermore, since the ${\mathbf{v}}_{n}$ are just scalar multiples of vectors in $S$, they are mutually orthogonal under the $\mathbf{R}$-inner product.  Lemma \ref{lem:eigenvectors} then implies that the $\mathbf{u}_{n}$ vectors are also mutually orthogonal.

The eigenvalue equations for ${\mathbf{v}}_{n}$ and ${\mathbf{u}}_{m}$ are
\begin{align}
	0 & =  (\mathbf{\Omega}-\mu_{n} \mathbf{R}){\mathbf{v}}_{n}, \\
	0 & = (\mathbf{\Omega}+\mu_{m} \mathbf{R}){\mathbf{u}}_{m}.
\end{align}
Combining these, we get
\begin{equation}
	(\mu_{n}+\mu_{m}^{*}) \langle {\mathbf{v}}_{n} , {\mathbf{u}}_{m} \rangle = 0.
\end{equation}
Since the $\mu_{n}$ are strictly real and positive, this gives $ \langle {\mathbf{v}}_{n} , {\mathbf{u}}_{m} \rangle = 0$.

Therefore, $B = \{  \mathbf{v}_{n},  \mathbf{u}_{n} \}_{n=1}^{N}$  is a set of mutually $\mathbf{R}$-orthogonal eigenvectors of $\mathbf{D}$, the $\mathbf{R}$-norm of $\mathbf{v}_{n}$ is +1, and the eigenvalue of $\mathbf{v}_{n}$ is positive; by Lemma \ref{lem:T}, a normal mode transformation exists.

\end{proof}

\begin{lemma}\label{lem:5}

Suppose that a normal mode transformation exists.  Then, $\mathbf{\Omega}$ is positive definite.

\end{lemma}

\begin{proof}

Since all the $\mu_{n}$ are positive when a normal mode transformation exists, we can define a matrix $\bm{\Theta}$ as follows:
\begin{gather}\nonumber
	\bm{\Theta}^{2} = \bm{\Sigma} \bm{\Lambda}, \\ \bm{\Theta} = \textup{diag}(\mu_{1}^{1/2},\ldots,\mu_{N}^{1/2},\mu_{1}^{1/2},\ldots,\mu_{N}^{1/2}).
\end{gather}
Furthermore, properties \ref{eq:T condition} and \ref{eq:condition 1} of a normal mode transformation can be combined to show
\begin{equation}
	\mathbf{T}^{\H} \bm{\Omega} \mathbf{T} = \bm{\Theta}^{2}.
\end{equation} 
This may be re-arranged to yield 
\begin{equation}
	 \bm{\Omega} = \mathbf{F}^{\H} \mathbf{F}, \quad \mathbf{F} = \bm{\Theta} \mathbf{T}^{-1}.
\end{equation}
$\mathbf{F}$ is obviously invertible, which implies that $\mathbf{\Omega}$ is positive definite.
	
\end{proof}

\begin{lemma}

A normal mode transformation exists if and only if $\mathbf{\Omega}$ is positive definite.

\end{lemma}

\begin{proof}

This follows directly from Lemmas \ref{lem:Omega positive} and \ref{lem:5}.

\end{proof}

\begin{acknowledgments}

We are supported by NSERC of Canada.  In addition, this research was supported in part by Perimeter Institute for Theoretical Physics. Research at Perimeter Institute is supported by the Government of Canada through Innovation, Science and Economic Development Canada and by the Province of Ontario through the Ministry of Research, Innovation and Science.

\end{acknowledgments}

\vfill

 \bibliography{cavity}

\begin{thebibliography}{86}%
\makeatletter
\providecommand \@ifxundefined [1]{%
 \@ifx{#1\undefined}
}%
\providecommand \@ifnum [1]{%
 \ifnum #1\expandafter \@firstoftwo
 \else \expandafter \@secondoftwo
 \fi
}%
\providecommand \@ifx [1]{%
 \ifx #1\expandafter \@firstoftwo
 \else \expandafter \@secondoftwo
 \fi
}%
\providecommand \natexlab [1]{#1}%
\providecommand \enquote  [1]{``#1''}%
\providecommand \bibnamefont  [1]{#1}%
\providecommand \bibfnamefont [1]{#1}%
\providecommand \citenamefont [1]{#1}%
\providecommand \href@noop [0]{\@secondoftwo}%
\providecommand \href [0]{\begingroup \@sanitize@url \@href}%
\providecommand \@href[1]{\@@startlink{#1}\@@href}%
\providecommand \@@href[1]{\endgroup#1\@@endlink}%
\providecommand \@sanitize@url [0]{\catcode `\\12\catcode `\$12\catcode
  `\&12\catcode `\#12\catcode `\^12\catcode `\_12\catcode `\%12\relax}%
\providecommand \@@startlink[1]{}%
\providecommand \@@endlink[0]{}%
\providecommand \url  [0]{\begingroup\@sanitize@url \@url }%
\providecommand \@url [1]{\endgroup\@href {#1}{\urlprefix }}%
\providecommand \urlprefix  [0]{URL }%
\providecommand \Eprint [0]{\href }%
\providecommand \doibase [0]{http://dx.doi.org/}%
\providecommand \selectlanguage [0]{\@gobble}%
\providecommand \bibinfo  [0]{\@secondoftwo}%
\providecommand \bibfield  [0]{\@secondoftwo}%
\providecommand \translation [1]{[#1]}%
\providecommand \BibitemOpen [0]{}%
\providecommand \bibitemStop [0]{}%
\providecommand \bibitemNoStop [0]{.\EOS\space}%
\providecommand \EOS [0]{\spacefactor3000\relax}%
\providecommand \BibitemShut  [1]{\csname bibitem#1\endcsname}%
\let\auto@bib@innerbib\@empty
\bibitem [{\citenamefont {Mattingly}(2005)}]{Mattingly:2005re}%
  \BibitemOpen
  \bibfield  {author} {\bibinfo {author} {\bibfnamefont {D.}~\bibnamefont
  {Mattingly}},\ }\href {\doibase 10.12942/lrr-2005-5} {\bibfield  {journal}
  {\bibinfo  {journal} {Living Rev. Rel.}\ }\textbf {\bibinfo {volume} {8}},\
  \bibinfo {pages} {5} (\bibinfo {year} {2005})},\ \Eprint
  {http://arxiv.org/abs/gr-qc/0502097} {arXiv:gr-qc/0502097 [gr-qc]}
  \BibitemShut {NoStop}%
\bibitem [{\citenamefont {Vucetich}(2005)}]{Vucetich:2005ra}%
  \BibitemOpen
  \bibfield  {author} {\bibinfo {author} {\bibfnamefont {H.}~\bibnamefont
  {Vucetich}},\ }\href@noop {} {\  (\bibinfo {year} {2005})},\ \Eprint
  {http://arxiv.org/abs/gr-qc/0502093} {arXiv:gr-qc/0502093 [gr-qc]}
  \BibitemShut {NoStop}%
\bibitem [{\citenamefont {Russell}(2008)}]{Russell:2006ge}%
  \BibitemOpen
  \bibfield  {author} {\bibinfo {author} {\bibfnamefont {N.}~\bibnamefont
  {Russell}},\ }\bibfield  {booktitle} {\emph {\bibinfo {booktitle}
  {{International Workshop: From Quantum to Cosmos: Fundamental Physics
  Research in Space Washington, District of Columbia, May 22-24, 2006}}},\
  }\href {\doibase 10.1142/S0218271807011668} {\bibfield  {journal} {\bibinfo
  {journal} {Int. J. Mod. Phys.}\ }\textbf {\bibinfo {volume} {D16}},\ \bibinfo
  {pages} {2469} (\bibinfo {year} {2008})},\ \Eprint
  {http://arxiv.org/abs/hep-ph/0608083} {arXiv:hep-ph/0608083 [hep-ph]}
  \BibitemShut {NoStop}%
\bibitem [{\citenamefont {Liberati}(2013{\natexlab{a}})}]{Liberati:2012tb}%
  \BibitemOpen
  \bibfield  {author} {\bibinfo {author} {\bibfnamefont {S.}~\bibnamefont
  {Liberati}},\ }\bibfield  {booktitle} {\emph {\bibinfo {booktitle} {{Analogue
  Gravity Phenomenology}}},\ }\href {\doibase 10.1007/978-3-319-00266-8_13}
  {\bibfield  {journal} {\bibinfo  {journal} {Lect. Notes Phys.}\ }\textbf
  {\bibinfo {volume} {870}},\ \bibinfo {pages} {297} (\bibinfo {year}
  {2013}{\natexlab{a}})},\ \Eprint {http://arxiv.org/abs/1203.4105}
  {arXiv:1203.4105 [gr-qc]} \BibitemShut {NoStop}%
\bibitem [{\citenamefont {Liberati}(2013{\natexlab{b}})}]{Liberati:2013xla}%
  \BibitemOpen
  \bibfield  {author} {\bibinfo {author} {\bibfnamefont {S.}~\bibnamefont
  {Liberati}},\ }\href {\doibase 10.1088/0264-9381/30/13/133001} {\bibfield
  {journal} {\bibinfo  {journal} {Class. Quant. Grav.}\ }\textbf {\bibinfo
  {volume} {30}},\ \bibinfo {pages} {133001} (\bibinfo {year}
  {2013}{\natexlab{b}})},\ \Eprint {http://arxiv.org/abs/1304.5795}
  {arXiv:1304.5795 [gr-qc]} \BibitemShut {NoStop}%
\bibitem [{\citenamefont {Liberati}(2015)}]{Liberati:2015dja}%
  \BibitemOpen
  \bibfield  {author} {\bibinfo {author} {\bibfnamefont {S.}~\bibnamefont
  {Liberati}},\ }\bibfield  {booktitle} {\emph {\bibinfo {booktitle}
  {{Proceedings, 4th Symposium on Prospects in the Physics of Discrete
  Symmetries (DISCRETE 2014): London, UK, December 2-6, 2014}}},\ }\href
  {\doibase 10.1088/1742-6596/631/1/012011} {\bibfield  {journal} {\bibinfo
  {journal} {J. Phys. Conf. Ser.}\ }\textbf {\bibinfo {volume} {631}},\
  \bibinfo {pages} {012011} (\bibinfo {year} {2015})}\BibitemShut {NoStop}%
\bibitem [{\citenamefont {Kostelecky}\ and\ \citenamefont
  {Samuel}(1989{\natexlab{a}})}]{Kostelecky:1988zi}%
  \BibitemOpen
  \bibfield  {author} {\bibinfo {author} {\bibfnamefont {V.~A.}\ \bibnamefont
  {Kostelecky}}\ and\ \bibinfo {author} {\bibfnamefont {S.}~\bibnamefont
  {Samuel}},\ }\href {\doibase 10.1103/PhysRevD.39.683} {\bibfield  {journal}
  {\bibinfo  {journal} {Phys. Rev.}\ }\textbf {\bibinfo {volume} {D39}},\
  \bibinfo {pages} {683} (\bibinfo {year} {1989}{\natexlab{a}})}\BibitemShut
  {NoStop}%
\bibitem [{\citenamefont {Kostelecky}\ and\ \citenamefont
  {Samuel}(1989{\natexlab{b}})}]{Kostelecky:1989jp}%
  \BibitemOpen
  \bibfield  {author} {\bibinfo {author} {\bibfnamefont {V.~A.}\ \bibnamefont
  {Kostelecky}}\ and\ \bibinfo {author} {\bibfnamefont {S.}~\bibnamefont
  {Samuel}},\ }\href {\doibase 10.1103/PhysRevLett.63.224} {\bibfield
  {journal} {\bibinfo  {journal} {Phys. Rev. Lett.}\ }\textbf {\bibinfo
  {volume} {63}},\ \bibinfo {pages} {224} (\bibinfo {year}
  {1989}{\natexlab{b}})}\BibitemShut {NoStop}%
\bibitem [{\citenamefont {Kostelecky}\ and\ \citenamefont
  {Samuel}(1989{\natexlab{c}})}]{Kostelecky:1989jw}%
  \BibitemOpen
  \bibfield  {author} {\bibinfo {author} {\bibfnamefont {V.~A.}\ \bibnamefont
  {Kostelecky}}\ and\ \bibinfo {author} {\bibfnamefont {S.}~\bibnamefont
  {Samuel}},\ }\href {\doibase 10.1103/PhysRevD.40.1886} {\bibfield  {journal}
  {\bibinfo  {journal} {Phys. Rev.}\ }\textbf {\bibinfo {volume} {D40}},\
  \bibinfo {pages} {1886} (\bibinfo {year} {1989}{\natexlab{c}})}\BibitemShut
  {NoStop}%
\bibitem [{\citenamefont {Kostelecky}\ and\ \citenamefont
  {Potting}(1991)}]{Kostelecky:1991ak}%
  \BibitemOpen
  \bibfield  {author} {\bibinfo {author} {\bibfnamefont {V.~A.}\ \bibnamefont
  {Kostelecky}}\ and\ \bibinfo {author} {\bibfnamefont {R.}~\bibnamefont
  {Potting}},\ }\href {\doibase 10.1016/0550-3213(91)90071-5} {\bibfield
  {journal} {\bibinfo  {journal} {Nucl. Phys.}\ }\textbf {\bibinfo {volume}
  {B359}},\ \bibinfo {pages} {545} (\bibinfo {year} {1991})}\BibitemShut
  {NoStop}%
\bibitem [{\citenamefont {Kostelecky}\ and\ \citenamefont
  {Potting}(1996)}]{Kostelecky:1995qk}%
  \BibitemOpen
  \bibfield  {author} {\bibinfo {author} {\bibfnamefont {V.~A.}\ \bibnamefont
  {Kostelecky}}\ and\ \bibinfo {author} {\bibfnamefont {R.}~\bibnamefont
  {Potting}},\ }\href {\doibase 10.1016/0370-2693(96)00589-8} {\bibfield
  {journal} {\bibinfo  {journal} {Phys. Lett.}\ }\textbf {\bibinfo {volume}
  {B381}},\ \bibinfo {pages} {89} (\bibinfo {year} {1996})},\ \Eprint
  {http://arxiv.org/abs/hep-th/9605088} {arXiv:hep-th/9605088 [hep-th]}
  \BibitemShut {NoStop}%
\bibitem [{\citenamefont {Kostelecky}\ \emph {et~al.}(2000)\citenamefont
  {Kostelecky}, \citenamefont {Perry},\ and\ \citenamefont
  {Potting}}]{Kostelecky:1999mu}%
  \BibitemOpen
  \bibfield  {author} {\bibinfo {author} {\bibfnamefont {V.~A.}\ \bibnamefont
  {Kostelecky}}, \bibinfo {author} {\bibfnamefont {M.}~\bibnamefont {Perry}}, \
  and\ \bibinfo {author} {\bibfnamefont {R.}~\bibnamefont {Potting}},\ }\href
  {\doibase 10.1103/PhysRevLett.84.4541} {\bibfield  {journal} {\bibinfo
  {journal} {Phys. Rev. Lett.}\ }\textbf {\bibinfo {volume} {84}},\ \bibinfo
  {pages} {4541} (\bibinfo {year} {2000})},\ \Eprint
  {http://arxiv.org/abs/hep-th/9912243} {arXiv:hep-th/9912243 [hep-th]}
  \BibitemShut {NoStop}%
\bibitem [{\citenamefont {Sudarsky}\ \emph {et~al.}(2003)\citenamefont
  {Sudarsky}, \citenamefont {Urrutia},\ and\ \citenamefont
  {Vucetich}}]{Sudarsky:2002zy}%
  \BibitemOpen
  \bibfield  {author} {\bibinfo {author} {\bibfnamefont {D.}~\bibnamefont
  {Sudarsky}}, \bibinfo {author} {\bibfnamefont {L.}~\bibnamefont {Urrutia}}, \
  and\ \bibinfo {author} {\bibfnamefont {H.}~\bibnamefont {Vucetich}},\ }\href
  {\doibase 10.1103/PhysRevD.68.024010} {\bibfield  {journal} {\bibinfo
  {journal} {Phys. Rev.}\ }\textbf {\bibinfo {volume} {D68}},\ \bibinfo {pages}
  {024010} (\bibinfo {year} {2003})},\ \Eprint
  {http://arxiv.org/abs/gr-qc/0211101} {arXiv:gr-qc/0211101 [gr-qc]}
  \BibitemShut {NoStop}%
\bibitem [{\citenamefont {Jenkins}(2004)}]{Jenkins:2003hw}%
  \BibitemOpen
  \bibfield  {author} {\bibinfo {author} {\bibfnamefont {A.}~\bibnamefont
  {Jenkins}},\ }\href {\doibase 10.1103/PhysRevD.69.105007} {\bibfield
  {journal} {\bibinfo  {journal} {Phys. Rev.}\ }\textbf {\bibinfo {volume}
  {D69}},\ \bibinfo {pages} {105007} (\bibinfo {year} {2004})},\ \Eprint
  {http://arxiv.org/abs/hep-th/0311127} {arXiv:hep-th/0311127 [hep-th]}
  \BibitemShut {NoStop}%
\bibitem [{\citenamefont {Gambini}\ and\ \citenamefont
  {Pullin}(1999)}]{Gambini:1998it}%
  \BibitemOpen
  \bibfield  {author} {\bibinfo {author} {\bibfnamefont {R.}~\bibnamefont
  {Gambini}}\ and\ \bibinfo {author} {\bibfnamefont {J.}~\bibnamefont
  {Pullin}},\ }\href {\doibase 10.1103/PhysRevD.59.124021} {\bibfield
  {journal} {\bibinfo  {journal} {Phys. Rev.}\ }\textbf {\bibinfo {volume}
  {D59}},\ \bibinfo {pages} {124021} (\bibinfo {year} {1999})},\ \Eprint
  {http://arxiv.org/abs/gr-qc/9809038} {arXiv:gr-qc/9809038 [gr-qc]}
  \BibitemShut {NoStop}%
\bibitem [{\citenamefont {Alfaro}\ \emph {et~al.}(2000)\citenamefont {Alfaro},
  \citenamefont {Morales-Tecotl},\ and\ \citenamefont
  {Urrutia}}]{Alfaro:1999wd}%
  \BibitemOpen
  \bibfield  {author} {\bibinfo {author} {\bibfnamefont {J.}~\bibnamefont
  {Alfaro}}, \bibinfo {author} {\bibfnamefont {H.~A.}\ \bibnamefont
  {Morales-Tecotl}}, \ and\ \bibinfo {author} {\bibfnamefont {L.~F.}\
  \bibnamefont {Urrutia}},\ }\href {\doibase 10.1103/PhysRevLett.84.2318}
  {\bibfield  {journal} {\bibinfo  {journal} {Phys. Rev. Lett.}\ }\textbf
  {\bibinfo {volume} {84}},\ \bibinfo {pages} {2318} (\bibinfo {year}
  {2000})},\ \Eprint {http://arxiv.org/abs/gr-qc/9909079} {arXiv:gr-qc/9909079
  [gr-qc]} \BibitemShut {NoStop}%
\bibitem [{\citenamefont {Alfaro}\ \emph
  {et~al.}(2002{\natexlab{a}})\citenamefont {Alfaro}, \citenamefont
  {Morales-Tecotl},\ and\ \citenamefont {Urrutia}}]{Alfaro:2001rb}%
  \BibitemOpen
  \bibfield  {author} {\bibinfo {author} {\bibfnamefont {J.}~\bibnamefont
  {Alfaro}}, \bibinfo {author} {\bibfnamefont {H.~A.}\ \bibnamefont
  {Morales-Tecotl}}, \ and\ \bibinfo {author} {\bibfnamefont {L.~F.}\
  \bibnamefont {Urrutia}},\ }\href {\doibase 10.1103/PhysRevD.65.103509}
  {\bibfield  {journal} {\bibinfo  {journal} {Phys. Rev.}\ }\textbf {\bibinfo
  {volume} {D65}},\ \bibinfo {pages} {103509} (\bibinfo {year}
  {2002}{\natexlab{a}})},\ \Eprint {http://arxiv.org/abs/hep-th/0108061}
  {arXiv:hep-th/0108061 [hep-th]} \BibitemShut {NoStop}%
\bibitem [{\citenamefont {Alfaro}\ \emph
  {et~al.}(2002{\natexlab{b}})\citenamefont {Alfaro}, \citenamefont
  {Morales-Tecotl},\ and\ \citenamefont {Urrutia}}]{Alfaro:2002xz}%
  \BibitemOpen
  \bibfield  {author} {\bibinfo {author} {\bibfnamefont {J.}~\bibnamefont
  {Alfaro}}, \bibinfo {author} {\bibfnamefont {H.~A.}\ \bibnamefont
  {Morales-Tecotl}}, \ and\ \bibinfo {author} {\bibfnamefont {L.~F.}\
  \bibnamefont {Urrutia}},\ }\href {\doibase 10.1103/PhysRevD.66.124006}
  {\bibfield  {journal} {\bibinfo  {journal} {Phys. Rev.}\ }\textbf {\bibinfo
  {volume} {D66}},\ \bibinfo {pages} {124006} (\bibinfo {year}
  {2002}{\natexlab{b}})},\ \Eprint {http://arxiv.org/abs/hep-th/0208192}
  {arXiv:hep-th/0208192 [hep-th]} \BibitemShut {NoStop}%
\bibitem [{\citenamefont {Seiberg}\ and\ \citenamefont
  {Witten}(1999)}]{Seiberg:1999vs}%
  \BibitemOpen
  \bibfield  {author} {\bibinfo {author} {\bibfnamefont {N.}~\bibnamefont
  {Seiberg}}\ and\ \bibinfo {author} {\bibfnamefont {E.}~\bibnamefont
  {Witten}},\ }\href {\doibase 10.1088/1126-6708/1999/09/032} {\bibfield
  {journal} {\bibinfo  {journal} {JHEP}\ }\textbf {\bibinfo {volume} {09}},\
  \bibinfo {pages} {032} (\bibinfo {year} {1999})},\ \Eprint
  {http://arxiv.org/abs/hep-th/9908142} {arXiv:hep-th/9908142 [hep-th]}
  \BibitemShut {NoStop}%
\bibitem [{\citenamefont {Hayakawa}(2000)}]{Hayakawa:1999yt}%
  \BibitemOpen
  \bibfield  {author} {\bibinfo {author} {\bibfnamefont {M.}~\bibnamefont
  {Hayakawa}},\ }\href {\doibase 10.1016/S0370-2693(00)00242-2} {\bibfield
  {journal} {\bibinfo  {journal} {Phys. Lett.}\ }\textbf {\bibinfo {volume}
  {B478}},\ \bibinfo {pages} {394} (\bibinfo {year} {2000})},\ \Eprint
  {http://arxiv.org/abs/hep-th/9912094} {arXiv:hep-th/9912094 [hep-th]}
  \BibitemShut {NoStop}%
\bibitem [{\citenamefont {Mocioiu}\ \emph {et~al.}(2000)\citenamefont
  {Mocioiu}, \citenamefont {Pospelov},\ and\ \citenamefont
  {Roiban}}]{Mocioiu:2000ip}%
  \BibitemOpen
  \bibfield  {author} {\bibinfo {author} {\bibfnamefont {I.}~\bibnamefont
  {Mocioiu}}, \bibinfo {author} {\bibfnamefont {M.}~\bibnamefont {Pospelov}}, \
  and\ \bibinfo {author} {\bibfnamefont {R.}~\bibnamefont {Roiban}},\ }\href
  {\doibase 10.1016/S0370-2693(00)00928-X} {\bibfield  {journal} {\bibinfo
  {journal} {Phys. Lett.}\ }\textbf {\bibinfo {volume} {B489}},\ \bibinfo
  {pages} {390} (\bibinfo {year} {2000})},\ \Eprint
  {http://arxiv.org/abs/hep-ph/0005191} {arXiv:hep-ph/0005191 [hep-ph]}
  \BibitemShut {NoStop}%
\bibitem [{\citenamefont {Guralnik}\ \emph {et~al.}(2001)\citenamefont
  {Guralnik}, \citenamefont {Jackiw}, \citenamefont {Pi},\ and\ \citenamefont
  {Polychronakos}}]{Guralnik:2001ax}%
  \BibitemOpen
  \bibfield  {author} {\bibinfo {author} {\bibfnamefont {Z.}~\bibnamefont
  {Guralnik}}, \bibinfo {author} {\bibfnamefont {R.}~\bibnamefont {Jackiw}},
  \bibinfo {author} {\bibfnamefont {S.~Y.}\ \bibnamefont {Pi}}, \ and\ \bibinfo
  {author} {\bibfnamefont {A.~P.}\ \bibnamefont {Polychronakos}},\ }\href
  {\doibase 10.1016/S0370-2693(01)00986-8} {\bibfield  {journal} {\bibinfo
  {journal} {Phys. Lett.}\ }\textbf {\bibinfo {volume} {B517}},\ \bibinfo
  {pages} {450} (\bibinfo {year} {2001})},\ \Eprint
  {http://arxiv.org/abs/hep-th/0106044} {arXiv:hep-th/0106044 [hep-th]}
  \BibitemShut {NoStop}%
\bibitem [{\citenamefont {Douglas}\ and\ \citenamefont
  {Nekrasov}(2001)}]{Douglas:2001ba}%
  \BibitemOpen
  \bibfield  {author} {\bibinfo {author} {\bibfnamefont {M.~R.}\ \bibnamefont
  {Douglas}}\ and\ \bibinfo {author} {\bibfnamefont {N.~A.}\ \bibnamefont
  {Nekrasov}},\ }\href {\doibase 10.1103/RevModPhys.73.977} {\bibfield
  {journal} {\bibinfo  {journal} {Rev. Mod. Phys.}\ }\textbf {\bibinfo {volume}
  {73}},\ \bibinfo {pages} {977} (\bibinfo {year} {2001})},\ \Eprint
  {http://arxiv.org/abs/hep-th/0106048} {arXiv:hep-th/0106048 [hep-th]}
  \BibitemShut {NoStop}%
\bibitem [{\citenamefont {Jurco}\ \emph {et~al.}(2001)\citenamefont {Jurco},
  \citenamefont {Moller}, \citenamefont {Schraml}, \citenamefont {Schupp},\
  and\ \citenamefont {Wess}}]{Jurco:2001rq}%
  \BibitemOpen
  \bibfield  {author} {\bibinfo {author} {\bibfnamefont {B.}~\bibnamefont
  {Jurco}}, \bibinfo {author} {\bibfnamefont {L.}~\bibnamefont {Moller}},
  \bibinfo {author} {\bibfnamefont {S.}~\bibnamefont {Schraml}}, \bibinfo
  {author} {\bibfnamefont {P.}~\bibnamefont {Schupp}}, \ and\ \bibinfo {author}
  {\bibfnamefont {J.}~\bibnamefont {Wess}},\ }\href {\doibase
  10.1007/s100520100731} {\bibfield  {journal} {\bibinfo  {journal} {Eur. Phys.
  J.}\ }\textbf {\bibinfo {volume} {C21}},\ \bibinfo {pages} {383} (\bibinfo
  {year} {2001})},\ \Eprint {http://arxiv.org/abs/hep-th/0104153}
  {arXiv:hep-th/0104153 [hep-th]} \BibitemShut {NoStop}%
\bibitem [{\citenamefont {Carlson}\ \emph {et~al.}(2001)\citenamefont
  {Carlson}, \citenamefont {Carone},\ and\ \citenamefont
  {Lebed}}]{Carlson:2001sw}%
  \BibitemOpen
  \bibfield  {author} {\bibinfo {author} {\bibfnamefont {C.~E.}\ \bibnamefont
  {Carlson}}, \bibinfo {author} {\bibfnamefont {C.~D.}\ \bibnamefont {Carone}},
  \ and\ \bibinfo {author} {\bibfnamefont {R.~F.}\ \bibnamefont {Lebed}},\
  }\href {\doibase 10.1016/S0370-2693(01)01045-0} {\bibfield  {journal}
  {\bibinfo  {journal} {Phys. Lett.}\ }\textbf {\bibinfo {volume} {B518}},\
  \bibinfo {pages} {201} (\bibinfo {year} {2001})},\ \Eprint
  {http://arxiv.org/abs/hep-ph/0107291} {arXiv:hep-ph/0107291 [hep-ph]}
  \BibitemShut {NoStop}%
\bibitem [{\citenamefont {Carroll}\ \emph {et~al.}(2001)\citenamefont
  {Carroll}, \citenamefont {Harvey}, \citenamefont {Kostelecky}, \citenamefont
  {Lane},\ and\ \citenamefont {Okamoto}}]{Carroll:2001ws}%
  \BibitemOpen
  \bibfield  {author} {\bibinfo {author} {\bibfnamefont {S.~M.}\ \bibnamefont
  {Carroll}}, \bibinfo {author} {\bibfnamefont {J.~A.}\ \bibnamefont {Harvey}},
  \bibinfo {author} {\bibfnamefont {V.~A.}\ \bibnamefont {Kostelecky}},
  \bibinfo {author} {\bibfnamefont {C.~D.}\ \bibnamefont {Lane}}, \ and\
  \bibinfo {author} {\bibfnamefont {T.}~\bibnamefont {Okamoto}},\ }\href
  {\doibase 10.1103/PhysRevLett.87.141601} {\bibfield  {journal} {\bibinfo
  {journal} {Phys. Rev. Lett.}\ }\textbf {\bibinfo {volume} {87}},\ \bibinfo
  {pages} {141601} (\bibinfo {year} {2001})},\ \Eprint
  {http://arxiv.org/abs/hep-th/0105082} {arXiv:hep-th/0105082 [hep-th]}
  \BibitemShut {NoStop}%
\bibitem [{\citenamefont {Anisimov}\ \emph {et~al.}(2002)\citenamefont
  {Anisimov}, \citenamefont {Banks}, \citenamefont {Dine},\ and\ \citenamefont
  {Graesser}}]{Anisimov:2001zc}%
  \BibitemOpen
  \bibfield  {author} {\bibinfo {author} {\bibfnamefont {A.}~\bibnamefont
  {Anisimov}}, \bibinfo {author} {\bibfnamefont {T.}~\bibnamefont {Banks}},
  \bibinfo {author} {\bibfnamefont {M.}~\bibnamefont {Dine}}, \ and\ \bibinfo
  {author} {\bibfnamefont {M.}~\bibnamefont {Graesser}},\ }\href {\doibase
  10.1103/PhysRevD.65.085032} {\bibfield  {journal} {\bibinfo  {journal} {Phys.
  Rev.}\ }\textbf {\bibinfo {volume} {D65}},\ \bibinfo {pages} {085032}
  (\bibinfo {year} {2002})},\ \Eprint {http://arxiv.org/abs/hep-ph/0106356}
  {arXiv:hep-ph/0106356 [hep-ph]} \BibitemShut {NoStop}%
\bibitem [{\citenamefont {Polchinski}(1992)}]{Polchinski:1992ed}%
  \BibitemOpen
  \bibfield  {author} {\bibinfo {author} {\bibfnamefont {J.}~\bibnamefont
  {Polchinski}},\ }in\ \href@noop {} {\emph {\bibinfo {booktitle}
  {{Proceedings, Theoretical Advanced Study Institute (TASI 92): From Black
  Holes and Strings to Particles: Boulder, USA, June 1-26, 1992}}}}\ (\bibinfo
  {year} {1992})\ pp.\ \bibinfo {pages} {0235--276},\ \Eprint
  {http://arxiv.org/abs/hep-th/9210046} {arXiv:hep-th/9210046 [hep-th]}
  \BibitemShut {NoStop}%
\bibitem [{\citenamefont {Ng}\ and\ \citenamefont {Van~Dam}(1994)}]{Ng:1993jb}%
  \BibitemOpen
  \bibfield  {author} {\bibinfo {author} {\bibfnamefont {Y.~J.}\ \bibnamefont
  {Ng}}\ and\ \bibinfo {author} {\bibfnamefont {H.}~\bibnamefont {Van~Dam}},\
  }\href {\doibase 10.1142/S0217732394000356} {\bibfield  {journal} {\bibinfo
  {journal} {Mod. Phys. Lett.}\ }\textbf {\bibinfo {volume} {A9}},\ \bibinfo
  {pages} {335} (\bibinfo {year} {1994})}\BibitemShut {NoStop}%
\bibitem [{\citenamefont {Shiokawa}(2000)}]{Shiokawa:2000em}%
  \BibitemOpen
  \bibfield  {author} {\bibinfo {author} {\bibfnamefont {K.}~\bibnamefont
  {Shiokawa}},\ }\href {\doibase 10.1103/PhysRevD.62.024002} {\bibfield
  {journal} {\bibinfo  {journal} {Phys. Rev.}\ }\textbf {\bibinfo {volume}
  {D62}},\ \bibinfo {pages} {024002} (\bibinfo {year} {2000})},\ \Eprint
  {http://arxiv.org/abs/hep-th/0001088} {arXiv:hep-th/0001088 [hep-th]}
  \BibitemShut {NoStop}%
\bibitem [{\citenamefont {Jacobson}\ and\ \citenamefont
  {Mattingly}(2001)}]{Jacobson:2000xp}%
  \BibitemOpen
  \bibfield  {author} {\bibinfo {author} {\bibfnamefont {T.}~\bibnamefont
  {Jacobson}}\ and\ \bibinfo {author} {\bibfnamefont {D.}~\bibnamefont
  {Mattingly}},\ }\href {\doibase 10.1103/PhysRevD.64.024028} {\bibfield
  {journal} {\bibinfo  {journal} {Phys. Rev.}\ }\textbf {\bibinfo {volume}
  {D64}},\ \bibinfo {pages} {024028} (\bibinfo {year} {2001})},\ \Eprint
  {http://arxiv.org/abs/gr-qc/0007031} {arXiv:gr-qc/0007031 [gr-qc]}
  \BibitemShut {NoStop}%
\bibitem [{\citenamefont {Amelino-Camelia}(2001)}]{AmelinoCamelia:2000ge}%
  \BibitemOpen
  \bibfield  {author} {\bibinfo {author} {\bibfnamefont {G.}~\bibnamefont
  {Amelino-Camelia}},\ }\href {\doibase 10.1016/S0370-2693(01)00506-8}
  {\bibfield  {journal} {\bibinfo  {journal} {Phys. Lett.}\ }\textbf {\bibinfo
  {volume} {B510}},\ \bibinfo {pages} {255} (\bibinfo {year} {2001})},\ \Eprint
  {http://arxiv.org/abs/hep-th/0012238} {arXiv:hep-th/0012238 [hep-th]}
  \BibitemShut {NoStop}%
\bibitem [{\citenamefont {Amelino-Camelia}(2002)}]{AmelinoCamelia:2000mn}%
  \BibitemOpen
  \bibfield  {author} {\bibinfo {author} {\bibfnamefont {G.}~\bibnamefont
  {Amelino-Camelia}},\ }\href {\doibase 10.1142/S0218271802001330} {\bibfield
  {journal} {\bibinfo  {journal} {Int. J. Mod. Phys.}\ }\textbf {\bibinfo
  {volume} {D11}},\ \bibinfo {pages} {35} (\bibinfo {year} {2002})},\ \Eprint
  {http://arxiv.org/abs/gr-qc/0012051} {arXiv:gr-qc/0012051 [gr-qc]}
  \BibitemShut {NoStop}%
\bibitem [{\citenamefont {Bruno}\ \emph {et~al.}(2001)\citenamefont {Bruno},
  \citenamefont {Amelino-Camelia},\ and\ \citenamefont
  {Kowalski-Glikman}}]{Bruno:2001mw}%
  \BibitemOpen
  \bibfield  {author} {\bibinfo {author} {\bibfnamefont {N.~R.}\ \bibnamefont
  {Bruno}}, \bibinfo {author} {\bibfnamefont {G.}~\bibnamefont
  {Amelino-Camelia}}, \ and\ \bibinfo {author} {\bibfnamefont {J.}~\bibnamefont
  {Kowalski-Glikman}},\ }\href {\doibase 10.1016/S0370-2693(01)01264-3}
  {\bibfield  {journal} {\bibinfo  {journal} {Phys. Lett.}\ }\textbf {\bibinfo
  {volume} {B522}},\ \bibinfo {pages} {133} (\bibinfo {year} {2001})},\ \Eprint
  {http://arxiv.org/abs/hep-th/0107039} {arXiv:hep-th/0107039 [hep-th]}
  \BibitemShut {NoStop}%
\bibitem [{\citenamefont {Dowker}\ \emph {et~al.}(2004)\citenamefont {Dowker},
  \citenamefont {Henson},\ and\ \citenamefont {Sorkin}}]{Dowker:2003hb}%
  \BibitemOpen
  \bibfield  {author} {\bibinfo {author} {\bibfnamefont {F.}~\bibnamefont
  {Dowker}}, \bibinfo {author} {\bibfnamefont {J.}~\bibnamefont {Henson}}, \
  and\ \bibinfo {author} {\bibfnamefont {R.~D.}\ \bibnamefont {Sorkin}},\
  }\href {\doibase 10.1142/S0217732304015026} {\bibfield  {journal} {\bibinfo
  {journal} {Mod. Phys. Lett.}\ }\textbf {\bibinfo {volume} {A19}},\ \bibinfo
  {pages} {1829} (\bibinfo {year} {2004})},\ \Eprint
  {http://arxiv.org/abs/gr-qc/0311055} {arXiv:gr-qc/0311055 [gr-qc]}
  \BibitemShut {NoStop}%
\bibitem [{\citenamefont {Kowalski-Glikman}(2005)}]{KowalskiGlikman:2004qa}%
  \BibitemOpen
  \bibfield  {author} {\bibinfo {author} {\bibfnamefont {J.}~\bibnamefont
  {Kowalski-Glikman}},\ }\bibfield  {booktitle} {\emph {\bibinfo {booktitle}
  {{Planck scale effects in astrophysics and cosmology. Proceedings, 40th
  Karpacs Winter School, Ladek Zdroj, Poland, February 4-14, 2004}}},\ }\href
  {\doibase 10.1007/11377306_5} {\bibfield  {journal} {\bibinfo  {journal}
  {Lect. Notes Phys.}\ }\textbf {\bibinfo {volume} {669}},\ \bibinfo {pages}
  {131} (\bibinfo {year} {2005})},\ \bibinfo {note} {[,131(2004)]},\ \Eprint
  {http://arxiv.org/abs/hep-th/0405273} {arXiv:hep-th/0405273 [hep-th]}
  \BibitemShut {NoStop}%
\bibitem [{\citenamefont {Lim}(2005)}]{Lim:2004js}%
  \BibitemOpen
  \bibfield  {author} {\bibinfo {author} {\bibfnamefont {E.~A.}\ \bibnamefont
  {Lim}},\ }\href {\doibase 10.1103/PhysRevD.71.063504} {\bibfield  {journal}
  {\bibinfo  {journal} {Phys. Rev.}\ }\textbf {\bibinfo {volume} {D71}},\
  \bibinfo {pages} {063504} (\bibinfo {year} {2005})},\ \Eprint
  {http://arxiv.org/abs/astro-ph/0407437} {arXiv:astro-ph/0407437 [astro-ph]}
  \BibitemShut {NoStop}%
\bibitem [{\citenamefont {Jacobson}\ and\ \citenamefont
  {Mattingly}(2004)}]{Jacobson:2004ts}%
  \BibitemOpen
  \bibfield  {author} {\bibinfo {author} {\bibfnamefont {T.}~\bibnamefont
  {Jacobson}}\ and\ \bibinfo {author} {\bibfnamefont {D.}~\bibnamefont
  {Mattingly}},\ }\href {\doibase 10.1103/PhysRevD.70.024003} {\bibfield
  {journal} {\bibinfo  {journal} {Phys. Rev.}\ }\textbf {\bibinfo {volume}
  {D70}},\ \bibinfo {pages} {024003} (\bibinfo {year} {2004})},\ \Eprint
  {http://arxiv.org/abs/gr-qc/0402005} {arXiv:gr-qc/0402005 [gr-qc]}
  \BibitemShut {NoStop}%
\bibitem [{\citenamefont {Colladay}\ and\ \citenamefont
  {Kostelecky}(1997)}]{Colladay:1996iz}%
  \BibitemOpen
  \bibfield  {author} {\bibinfo {author} {\bibfnamefont {D.}~\bibnamefont
  {Colladay}}\ and\ \bibinfo {author} {\bibfnamefont {V.~A.}\ \bibnamefont
  {Kostelecky}},\ }\href {\doibase 10.1103/PhysRevD.55.6760} {\bibfield
  {journal} {\bibinfo  {journal} {Phys. Rev.}\ }\textbf {\bibinfo {volume}
  {D55}},\ \bibinfo {pages} {6760} (\bibinfo {year} {1997})},\ \Eprint
  {http://arxiv.org/abs/hep-ph/9703464} {arXiv:hep-ph/9703464 [hep-ph]}
  \BibitemShut {NoStop}%
\bibitem [{\citenamefont {Colladay}\ and\ \citenamefont
  {Kostelecky}(1998)}]{Colladay:1998fq}%
  \BibitemOpen
  \bibfield  {author} {\bibinfo {author} {\bibfnamefont {D.}~\bibnamefont
  {Colladay}}\ and\ \bibinfo {author} {\bibfnamefont {V.~A.}\ \bibnamefont
  {Kostelecky}},\ }\href {\doibase 10.1103/PhysRevD.58.116002} {\bibfield
  {journal} {\bibinfo  {journal} {Phys. Rev.}\ }\textbf {\bibinfo {volume}
  {D58}},\ \bibinfo {pages} {116002} (\bibinfo {year} {1998})},\ \Eprint
  {http://arxiv.org/abs/hep-ph/9809521} {arXiv:hep-ph/9809521 [hep-ph]}
  \BibitemShut {NoStop}%
\bibitem [{\citenamefont {Kostelecky}(2004)}]{Kostelecky:2003fs}%
  \BibitemOpen
  \bibfield  {author} {\bibinfo {author} {\bibfnamefont {V.~A.}\ \bibnamefont
  {Kostelecky}},\ }\href {\doibase 10.1103/PhysRevD.69.105009} {\bibfield
  {journal} {\bibinfo  {journal} {Phys. Rev.}\ }\textbf {\bibinfo {volume}
  {D69}},\ \bibinfo {pages} {105009} (\bibinfo {year} {2004})},\ \Eprint
  {http://arxiv.org/abs/hep-th/0312310} {arXiv:hep-th/0312310 [hep-th]}
  \BibitemShut {NoStop}%
\bibitem [{\citenamefont {Bluhm}(2006)}]{Bluhm:2005uj}%
  \BibitemOpen
  \bibfield  {author} {\bibinfo {author} {\bibfnamefont {R.}~\bibnamefont
  {Bluhm}},\ }\bibfield  {booktitle} {\emph {\bibinfo {booktitle} {{339th WE
  Heraeus Seminar on Special Relativity: Will It Survive the Next 100 Years?
  Potsdam, Germany, February 13-18, 2005}}},\ }\href {\doibase
  10.1007/3-540-34523-X_8} {\bibfield  {journal} {\bibinfo  {journal} {Lect.
  Notes Phys.}\ }\textbf {\bibinfo {volume} {702}},\ \bibinfo {pages} {191}
  (\bibinfo {year} {2006})},\ \bibinfo {note} {[,191(2005)]},\ \Eprint
  {http://arxiv.org/abs/hep-ph/0506054} {arXiv:hep-ph/0506054 [hep-ph]}
  \BibitemShut {NoStop}%
\bibitem [{\citenamefont {Greenberg}(2002)}]{Greenberg:2002uu}%
  \BibitemOpen
  \bibfield  {author} {\bibinfo {author} {\bibfnamefont {O.~W.}\ \bibnamefont
  {Greenberg}},\ }\href {\doibase 10.1103/PhysRevLett.89.231602} {\bibfield
  {journal} {\bibinfo  {journal} {Phys. Rev. Lett.}\ }\textbf {\bibinfo
  {volume} {89}},\ \bibinfo {pages} {231602} (\bibinfo {year} {2002})},\
  \Eprint {http://arxiv.org/abs/hep-ph/0201258} {arXiv:hep-ph/0201258 [hep-ph]}
  \BibitemShut {NoStop}%
\bibitem [{\citenamefont {Greenberg}(2011)}]{Greenberg:2011cw}%
  \BibitemOpen
  \bibfield  {author} {\bibinfo {author} {\bibfnamefont {O.~W.~A.}\
  \bibnamefont {Greenberg}},\ }\href@noop {} {\  (\bibinfo {year} {2011})},\
  \Eprint {http://arxiv.org/abs/1105.0927} {arXiv:1105.0927 [hep-ph]}
  \BibitemShut {NoStop}%
\bibitem [{\citenamefont {Kostelecky}\ and\ \citenamefont
  {Russell}(2011)}]{Kostelecky:2008ts}%
  \BibitemOpen
  \bibfield  {author} {\bibinfo {author} {\bibfnamefont {V.~A.}\ \bibnamefont
  {Kostelecky}}\ and\ \bibinfo {author} {\bibfnamefont {N.}~\bibnamefont
  {Russell}},\ }\href {\doibase 10.1103/RevModPhys.83.11} {\bibfield  {journal}
  {\bibinfo  {journal} {Rev. Mod. Phys.}\ }\textbf {\bibinfo {volume} {83}},\
  \bibinfo {pages} {11} (\bibinfo {year} {2011})},\ \Eprint
  {http://arxiv.org/abs/0801.0287} {arXiv:0801.0287 [hep-ph]} \BibitemShut
  {NoStop}%
\bibitem [{\citenamefont {Arkani-Hamed}\ \emph {et~al.}(2004)\citenamefont
  {Arkani-Hamed}, \citenamefont {Cheng}, \citenamefont {Luty},\ and\
  \citenamefont {Mukohyama}}]{ArkaniHamed:2003uy}%
  \BibitemOpen
  \bibfield  {author} {\bibinfo {author} {\bibfnamefont {N.}~\bibnamefont
  {Arkani-Hamed}}, \bibinfo {author} {\bibfnamefont {H.-C.}\ \bibnamefont
  {Cheng}}, \bibinfo {author} {\bibfnamefont {M.~A.}\ \bibnamefont {Luty}}, \
  and\ \bibinfo {author} {\bibfnamefont {S.}~\bibnamefont {Mukohyama}},\ }\href
  {\doibase 10.1088/1126-6708/2004/05/074} {\bibfield  {journal} {\bibinfo
  {journal} {JHEP}\ }\textbf {\bibinfo {volume} {05}},\ \bibinfo {pages} {074}
  (\bibinfo {year} {2004})},\ \Eprint {http://arxiv.org/abs/hep-th/0312099}
  {arXiv:hep-th/0312099 [hep-th]} \BibitemShut {NoStop}%
\bibitem [{\citenamefont {Kostelecky}\ and\ \citenamefont
  {Mewes}(2009)}]{Kostelecky:2009zp}%
  \BibitemOpen
  \bibfield  {author} {\bibinfo {author} {\bibfnamefont {V.~A.}\ \bibnamefont
  {Kostelecky}}\ and\ \bibinfo {author} {\bibfnamefont {M.}~\bibnamefont
  {Mewes}},\ }\href {\doibase 10.1103/PhysRevD.80.015020} {\bibfield  {journal}
  {\bibinfo  {journal} {Phys. Rev.}\ }\textbf {\bibinfo {volume} {D80}},\
  \bibinfo {pages} {015020} (\bibinfo {year} {2009})},\ \Eprint
  {http://arxiv.org/abs/0905.0031} {arXiv:0905.0031 [hep-ph]} \BibitemShut
  {NoStop}%
\bibitem [{\citenamefont {Unruh}(1995)}]{PhysRevD.51.2827}%
  \BibitemOpen
  \bibfield  {author} {\bibinfo {author} {\bibfnamefont {W.~G.}\ \bibnamefont
  {Unruh}},\ }\href {\doibase 10.1103/PhysRevD.51.2827} {\bibfield  {journal}
  {\bibinfo  {journal} {Phys. Rev. D}\ }\textbf {\bibinfo {volume} {51}},\
  \bibinfo {pages} {2827} (\bibinfo {year} {1995})}\BibitemShut {NoStop}%
\bibitem [{\citenamefont {Corley}\ and\ \citenamefont
  {Jacobson}(1996)}]{Corley:1996ar}%
  \BibitemOpen
  \bibfield  {author} {\bibinfo {author} {\bibfnamefont {S.}~\bibnamefont
  {Corley}}\ and\ \bibinfo {author} {\bibfnamefont {T.}~\bibnamefont
  {Jacobson}},\ }\href {\doibase 10.1103/PhysRevD.54.1568} {\bibfield
  {journal} {\bibinfo  {journal} {Phys. Rev.}\ }\textbf {\bibinfo {volume}
  {D54}},\ \bibinfo {pages} {1568} (\bibinfo {year} {1996})},\ \Eprint
  {http://arxiv.org/abs/hep-th/9601073} {arXiv:hep-th/9601073 [hep-th]}
  \BibitemShut {NoStop}%
\bibitem [{\citenamefont {Husain}\ and\ \citenamefont
  {Louko}(2016)}]{Husain:2015tna}%
  \BibitemOpen
  \bibfield  {author} {\bibinfo {author} {\bibfnamefont {V.}~\bibnamefont
  {Husain}}\ and\ \bibinfo {author} {\bibfnamefont {J.}~\bibnamefont {Louko}},\
  }\href {\doibase 10.1103/PhysRevLett.116.061301} {\bibfield  {journal}
  {\bibinfo  {journal} {Phys. Rev. Lett.}\ }\textbf {\bibinfo {volume} {116}},\
  \bibinfo {pages} {061301} (\bibinfo {year} {2016})},\ \Eprint
  {http://arxiv.org/abs/1508.05338} {arXiv:1508.05338 [gr-qc]} \BibitemShut
  {NoStop}%
\bibitem [{\citenamefont {Galaverni}\ and\ \citenamefont
  {Sigl}(2008)}]{Galaverni:2008yj}%
  \BibitemOpen
  \bibfield  {author} {\bibinfo {author} {\bibfnamefont {M.}~\bibnamefont
  {Galaverni}}\ and\ \bibinfo {author} {\bibfnamefont {G.}~\bibnamefont
  {Sigl}},\ }\href {\doibase 10.1103/PhysRevD.78.063003} {\bibfield  {journal}
  {\bibinfo  {journal} {Phys. Rev.}\ }\textbf {\bibinfo {volume} {D78}},\
  \bibinfo {pages} {063003} (\bibinfo {year} {2008})},\ \Eprint
  {http://arxiv.org/abs/0807.1210} {arXiv:0807.1210 [astro-ph]} \BibitemShut
  {NoStop}%
\bibitem [{\citenamefont {Stecker}(2011)}]{Stecker:2011ps}%
  \BibitemOpen
  \bibfield  {author} {\bibinfo {author} {\bibfnamefont {F.~W.}\ \bibnamefont
  {Stecker}},\ }\href {\doibase 10.1016/j.astropartphys.2011.06.007} {\bibfield
   {journal} {\bibinfo  {journal} {Astropart. Phys.}\ }\textbf {\bibinfo
  {volume} {35}},\ \bibinfo {pages} {95} (\bibinfo {year} {2011})},\ \Eprint
  {http://arxiv.org/abs/1102.2784} {arXiv:1102.2784 [astro-ph.HE]} \BibitemShut
  {NoStop}%
\bibitem [{\citenamefont {Vasileiou}\ \emph {et~al.}(2013)\citenamefont
  {Vasileiou}, \citenamefont {Jacholkowska}, \citenamefont {Piron},
  \citenamefont {Bolmont}, \citenamefont {Couturier}, \citenamefont {Granot},
  \citenamefont {Stecker}, \citenamefont {Cohen-Tanugi},\ and\ \citenamefont
  {Longo}}]{Vasileiou:2013vra}%
  \BibitemOpen
  \bibfield  {author} {\bibinfo {author} {\bibfnamefont {V.}~\bibnamefont
  {Vasileiou}}, \bibinfo {author} {\bibfnamefont {A.}~\bibnamefont
  {Jacholkowska}}, \bibinfo {author} {\bibfnamefont {F.}~\bibnamefont {Piron}},
  \bibinfo {author} {\bibfnamefont {J.}~\bibnamefont {Bolmont}}, \bibinfo
  {author} {\bibfnamefont {C.}~\bibnamefont {Couturier}}, \bibinfo {author}
  {\bibfnamefont {J.}~\bibnamefont {Granot}}, \bibinfo {author} {\bibfnamefont
  {F.~W.}\ \bibnamefont {Stecker}}, \bibinfo {author} {\bibfnamefont
  {J.}~\bibnamefont {Cohen-Tanugi}}, \ and\ \bibinfo {author} {\bibfnamefont
  {F.}~\bibnamefont {Longo}},\ }\href {\doibase 10.1103/PhysRevD.87.122001}
  {\bibfield  {journal} {\bibinfo  {journal} {Phys. Rev.}\ }\textbf {\bibinfo
  {volume} {D87}},\ \bibinfo {pages} {122001} (\bibinfo {year} {2013})},\
  \Eprint {http://arxiv.org/abs/1305.3463} {arXiv:1305.3463 [astro-ph.HE]}
  \BibitemShut {NoStop}%
\bibitem [{\citenamefont {Kislat}\ and\ \citenamefont
  {Krawczynski}(2017)}]{Kislat:2016ehx}%
  \BibitemOpen
  \bibfield  {author} {\bibinfo {author} {\bibfnamefont {F.}~\bibnamefont
  {Kislat}}\ and\ \bibinfo {author} {\bibfnamefont {H.}~\bibnamefont
  {Krawczynski}},\ }in\ \href {\doibase 10.1142/9789813148505_0029} {\emph
  {\bibinfo {booktitle} {{Proceedings, 7th Meeting on CPT and Lorentz Symmetry
  (CPT 16): Bloomington, Indiana, USA, June 20-24, 2016}}}}\ (\bibinfo {year}
  {2017})\ pp.\ \bibinfo {pages} {113--116},\ \Eprint
  {http://arxiv.org/abs/1608.03631} {arXiv:1608.03631 [astro-ph.HE]}
  \BibitemShut {NoStop}%
\bibitem [{\citenamefont {Wei}\ \emph {et~al.}(2017)\citenamefont {Wei},
  \citenamefont {Wu}, \citenamefont {Zhang}, \citenamefont {Shao},
  \citenamefont {Mészáros},\ and\ \citenamefont {Kostelecký}}]{Wei:2017zuu}%
  \BibitemOpen
  \bibfield  {author} {\bibinfo {author} {\bibfnamefont {J.-J.}\ \bibnamefont
  {Wei}}, \bibinfo {author} {\bibfnamefont {X.-F.}\ \bibnamefont {Wu}},
  \bibinfo {author} {\bibfnamefont {B.-B.}\ \bibnamefont {Zhang}}, \bibinfo
  {author} {\bibfnamefont {L.}~\bibnamefont {Shao}}, \bibinfo {author}
  {\bibfnamefont {P.}~\bibnamefont {Mészáros}}, \ and\ \bibinfo {author}
  {\bibfnamefont {V.~A.}\ \bibnamefont {Kostelecký}},\ }\href {\doibase
  10.3847/1538-4357/aa7630} {\bibfield  {journal} {\bibinfo  {journal}
  {Astrophys. J.}\ }\textbf {\bibinfo {volume} {842}},\ \bibinfo {pages} {115}
  (\bibinfo {year} {2017})},\ \Eprint {http://arxiv.org/abs/1704.05984}
  {arXiv:1704.05984 [astro-ph.HE]} \BibitemShut {NoStop}%
\bibitem [{\citenamefont {Lang}\ \emph {et~al.}(2019)\citenamefont {Lang},
  \citenamefont {Martínez-Huerta},\ and\ \citenamefont
  {de~Souza}}]{Lang:2018yog}%
  \BibitemOpen
  \bibfield  {author} {\bibinfo {author} {\bibfnamefont {R.~G.}\ \bibnamefont
  {Lang}}, \bibinfo {author} {\bibfnamefont {H.}~\bibnamefont
  {Martínez-Huerta}}, \ and\ \bibinfo {author} {\bibfnamefont
  {V.}~\bibnamefont {de~Souza}},\ }\href {\doibase 10.1103/PhysRevD.99.043015}
  {\bibfield  {journal} {\bibinfo  {journal} {Phys. Rev.}\ }\textbf {\bibinfo
  {volume} {D99}},\ \bibinfo {pages} {043015} (\bibinfo {year} {2019})},\
  \Eprint {http://arxiv.org/abs/1810.13215} {arXiv:1810.13215 [astro-ph.HE]}
  \BibitemShut {NoStop}%
\bibitem [{\citenamefont {Kostelecky}\ and\ \citenamefont
  {Lane}(1999)}]{Kostelecky:1999mr}%
  \BibitemOpen
  \bibfield  {author} {\bibinfo {author} {\bibfnamefont {V.~A.}\ \bibnamefont
  {Kostelecky}}\ and\ \bibinfo {author} {\bibfnamefont {C.~D.}\ \bibnamefont
  {Lane}},\ }\href {\doibase 10.1103/PhysRevD.60.116010} {\bibfield  {journal}
  {\bibinfo  {journal} {Phys. Rev.}\ }\textbf {\bibinfo {volume} {D60}},\
  \bibinfo {pages} {116010} (\bibinfo {year} {1999})},\ \Eprint
  {http://arxiv.org/abs/hep-ph/9908504} {arXiv:hep-ph/9908504 [hep-ph]}
  \BibitemShut {NoStop}%
\bibitem [{\citenamefont {Mewes}(2003)}]{Mewes:2003mw}%
  \BibitemOpen
  \bibfield  {author} {\bibinfo {author} {\bibfnamefont {M.}~\bibnamefont
  {Mewes}},\ }in\ \href@noop {} {\emph {\bibinfo {booktitle} {{2003 NASA / JPL
  Workshop on Fundamental Physics in Space Oxnard, California, April 14-16,
  2003}}}}\ (\bibinfo {year} {2003})\ \Eprint
  {http://arxiv.org/abs/hep-ph/0307161} {arXiv:hep-ph/0307161 [hep-ph]}
  \BibitemShut {NoStop}%
\bibitem [{\citenamefont {Russell}(2004)}]{Russell:2004ne}%
  \BibitemOpen
  \bibfield  {author} {\bibinfo {author} {\bibfnamefont {N.}~\bibnamefont
  {Russell}},\ }\bibfield  {booktitle} {\emph {\bibinfo {booktitle} {{HYPER
  Symposium Paris, France, November 4-6, 2002}}},\ }\href {\doibase
  10.1023/B:GERG.0000046187.55942.fa} {\  (\bibinfo {year} {2004}),\
  10.1023/B:GERG.0000046187.55942.fa},\ \bibinfo {note} {[Gen. Rel.
  Grav.36,2341(2004)]},\ \Eprint {http://arxiv.org/abs/hep-ph/0407256}
  {arXiv:hep-ph/0407256 [hep-ph]} \BibitemShut {NoStop}%
\bibitem [{\citenamefont {Kostelecký}\ and\ \citenamefont
  {Vargas}(2018)}]{Kostelecky:2018fmc}%
  \BibitemOpen
  \bibfield  {author} {\bibinfo {author} {\bibfnamefont {V.~A.}\ \bibnamefont
  {Kostelecký}}\ and\ \bibinfo {author} {\bibfnamefont {A.~J.}\ \bibnamefont
  {Vargas}},\ }\href {\doibase 10.1103/PhysRevD.98.036003} {\bibfield
  {journal} {\bibinfo  {journal} {Phys. Rev.}\ }\textbf {\bibinfo {volume}
  {D98}},\ \bibinfo {pages} {036003} (\bibinfo {year} {2018})},\ \Eprint
  {http://arxiv.org/abs/1805.04499} {arXiv:1805.04499 [hep-ph]} \BibitemShut
  {NoStop}%
\bibitem [{\citenamefont {Bluhm}\ \emph {et~al.}(1997)\citenamefont {Bluhm},
  \citenamefont {Kostelecky},\ and\ \citenamefont {Russell}}]{Bluhm:1997ci}%
  \BibitemOpen
  \bibfield  {author} {\bibinfo {author} {\bibfnamefont {R.}~\bibnamefont
  {Bluhm}}, \bibinfo {author} {\bibfnamefont {V.~A.}\ \bibnamefont
  {Kostelecky}}, \ and\ \bibinfo {author} {\bibfnamefont {N.}~\bibnamefont
  {Russell}},\ }\href {\doibase 10.1103/PhysRevLett.79.1432} {\bibfield
  {journal} {\bibinfo  {journal} {Phys. Rev. Lett.}\ }\textbf {\bibinfo
  {volume} {79}},\ \bibinfo {pages} {1432} (\bibinfo {year} {1997})},\ \Eprint
  {http://arxiv.org/abs/hep-ph/9707364} {arXiv:hep-ph/9707364 [hep-ph]}
  \BibitemShut {NoStop}%
\bibitem [{\citenamefont {Bluhm}\ \emph {et~al.}(1998)\citenamefont {Bluhm},
  \citenamefont {Kostelecky},\ and\ \citenamefont {Russell}}]{Bluhm:1997qb}%
  \BibitemOpen
  \bibfield  {author} {\bibinfo {author} {\bibfnamefont {R.}~\bibnamefont
  {Bluhm}}, \bibinfo {author} {\bibfnamefont {V.~A.}\ \bibnamefont
  {Kostelecky}}, \ and\ \bibinfo {author} {\bibfnamefont {N.}~\bibnamefont
  {Russell}},\ }\href {\doibase 10.1103/PhysRevD.57.3932} {\bibfield  {journal}
  {\bibinfo  {journal} {Phys. Rev.}\ }\textbf {\bibinfo {volume} {D57}},\
  \bibinfo {pages} {3932} (\bibinfo {year} {1998})},\ \Eprint
  {http://arxiv.org/abs/hep-ph/9809543} {arXiv:hep-ph/9809543 [hep-ph]}
  \BibitemShut {NoStop}%
\bibitem [{\citenamefont {Brillet}\ and\ \citenamefont
  {Hall}(1979)}]{Brillet:1979zz}%
  \BibitemOpen
  \bibfield  {author} {\bibinfo {author} {\bibfnamefont {A.}~\bibnamefont
  {Brillet}}\ and\ \bibinfo {author} {\bibfnamefont {J.~L.}\ \bibnamefont
  {Hall}},\ }\href {\doibase 10.1103/PhysRevLett.42.549} {\bibfield  {journal}
  {\bibinfo  {journal} {Phys. Rev. Lett.}\ }\textbf {\bibinfo {volume} {42}},\
  \bibinfo {pages} {549} (\bibinfo {year} {1979})}\BibitemShut {NoStop}%
\bibitem [{\citenamefont {Hils}\ and\ \citenamefont
  {Hall}(1990)}]{Hils:1990nrd}%
  \BibitemOpen
  \bibfield  {author} {\bibinfo {author} {\bibfnamefont {D.}~\bibnamefont
  {Hils}}\ and\ \bibinfo {author} {\bibfnamefont {J.~L.}\ \bibnamefont
  {Hall}},\ }\href {\doibase 10.1103/PhysRevLett.64.1697} {\bibfield  {journal}
  {\bibinfo  {journal} {Phys. Rev. Lett.}\ }\textbf {\bibinfo {volume} {64}},\
  \bibinfo {pages} {1697} (\bibinfo {year} {1990})}\BibitemShut {NoStop}%
\bibitem [{\citenamefont {Braxmaier}\ \emph {et~al.}(2002)\citenamefont
  {Braxmaier}, \citenamefont {Muller}, \citenamefont {Pradl}, \citenamefont
  {Mlynek}, \citenamefont {Peters},\ and\ \citenamefont
  {Schiller}}]{Braxmaier:2001wu}%
  \BibitemOpen
  \bibfield  {author} {\bibinfo {author} {\bibfnamefont {C.}~\bibnamefont
  {Braxmaier}}, \bibinfo {author} {\bibfnamefont {H.}~\bibnamefont {Muller}},
  \bibinfo {author} {\bibfnamefont {O.}~\bibnamefont {Pradl}}, \bibinfo
  {author} {\bibfnamefont {J.}~\bibnamefont {Mlynek}}, \bibinfo {author}
  {\bibfnamefont {A.}~\bibnamefont {Peters}}, \ and\ \bibinfo {author}
  {\bibfnamefont {S.}~\bibnamefont {Schiller}},\ }\href {\doibase
  10.1103/PhysRevLett.88.010401} {\bibfield  {journal} {\bibinfo  {journal}
  {Phys. Rev. Lett.}\ }\textbf {\bibinfo {volume} {88}},\ \bibinfo {pages}
  {010401} (\bibinfo {year} {2002})}\BibitemShut {NoStop}%
\bibitem [{\citenamefont {Kostelecky}\ and\ \citenamefont
  {Mewes}(2002)}]{Kostelecky:2002hh}%
  \BibitemOpen
  \bibfield  {author} {\bibinfo {author} {\bibfnamefont {V.~A.}\ \bibnamefont
  {Kostelecky}}\ and\ \bibinfo {author} {\bibfnamefont {M.}~\bibnamefont
  {Mewes}},\ }\href {\doibase 10.1103/PhysRevD.66.056005} {\bibfield  {journal}
  {\bibinfo  {journal} {Phys. Rev.}\ }\textbf {\bibinfo {volume} {D66}},\
  \bibinfo {pages} {056005} (\bibinfo {year} {2002})},\ \Eprint
  {http://arxiv.org/abs/hep-ph/0205211} {arXiv:hep-ph/0205211 [hep-ph]}
  \BibitemShut {NoStop}%
\bibitem [{\citenamefont {Muller}\ \emph
  {et~al.}(2003{\natexlab{a}})\citenamefont {Muller}, \citenamefont
  {Braxmaier}, \citenamefont {Herrmann}, \citenamefont {Peters},\ and\
  \citenamefont {Laemmerzahl}}]{Muller:2002uk}%
  \BibitemOpen
  \bibfield  {author} {\bibinfo {author} {\bibfnamefont {H.}~\bibnamefont
  {Muller}}, \bibinfo {author} {\bibfnamefont {C.}~\bibnamefont {Braxmaier}},
  \bibinfo {author} {\bibfnamefont {S.}~\bibnamefont {Herrmann}}, \bibinfo
  {author} {\bibfnamefont {A.}~\bibnamefont {Peters}}, \ and\ \bibinfo {author}
  {\bibfnamefont {C.}~\bibnamefont {Laemmerzahl}},\ }\href {\doibase
  10.1103/PhysRevD.67.056006} {\bibfield  {journal} {\bibinfo  {journal} {Phys.
  Rev.}\ }\textbf {\bibinfo {volume} {D67}},\ \bibinfo {pages} {056006}
  (\bibinfo {year} {2003}{\natexlab{a}})},\ \Eprint
  {http://arxiv.org/abs/hep-ph/0212289} {arXiv:hep-ph/0212289 [hep-ph]}
  \BibitemShut {NoStop}%
\bibitem [{\citenamefont {Muller}\ \emph
  {et~al.}(2003{\natexlab{b}})\citenamefont {Muller}, \citenamefont {Herrmann},
  \citenamefont {Braxmaier}, \citenamefont {Schiller},\ and\ \citenamefont
  {Peters}}]{Muller:2003zzc}%
  \BibitemOpen
  \bibfield  {author} {\bibinfo {author} {\bibfnamefont {H.}~\bibnamefont
  {Muller}}, \bibinfo {author} {\bibfnamefont {S.}~\bibnamefont {Herrmann}},
  \bibinfo {author} {\bibfnamefont {C.}~\bibnamefont {Braxmaier}}, \bibinfo
  {author} {\bibfnamefont {S.}~\bibnamefont {Schiller}}, \ and\ \bibinfo
  {author} {\bibfnamefont {A.}~\bibnamefont {Peters}},\ }\href {\doibase
  10.1103/PhysRevLett.91.020401} {\bibfield  {journal} {\bibinfo  {journal}
  {Phys. Rev. Lett.}\ }\textbf {\bibinfo {volume} {91}},\ \bibinfo {pages}
  {020401} (\bibinfo {year} {2003}{\natexlab{b}})},\ \Eprint
  {http://arxiv.org/abs/physics/0305117} {arXiv:physics/0305117
  [physics.class-ph]} \BibitemShut {NoStop}%
\bibitem [{\citenamefont {Muller}(2005)}]{Muller:2004zp}%
  \BibitemOpen
  \bibfield  {author} {\bibinfo {author} {\bibfnamefont {H.}~\bibnamefont
  {Muller}},\ }\href {\doibase 10.1103/PhysRevD.71.045004} {\bibfield
  {journal} {\bibinfo  {journal} {Phys. Rev.}\ }\textbf {\bibinfo {volume}
  {D71}},\ \bibinfo {pages} {045004} (\bibinfo {year} {2005})},\ \Eprint
  {http://arxiv.org/abs/hep-ph/0412385} {arXiv:hep-ph/0412385 [hep-ph]}
  \BibitemShut {NoStop}%
\bibitem [{\citenamefont {Wolf}\ \emph {et~al.}(2004)\citenamefont {Wolf},
  \citenamefont {Bize}, \citenamefont {Clairon}, \citenamefont {Santarelli},
  \citenamefont {Tobar},\ and\ \citenamefont {Luiten}}]{Wolf:2004gg}%
  \BibitemOpen
  \bibfield  {author} {\bibinfo {author} {\bibfnamefont {P.}~\bibnamefont
  {Wolf}}, \bibinfo {author} {\bibfnamefont {S.}~\bibnamefont {Bize}}, \bibinfo
  {author} {\bibfnamefont {A.}~\bibnamefont {Clairon}}, \bibinfo {author}
  {\bibfnamefont {G.}~\bibnamefont {Santarelli}}, \bibinfo {author}
  {\bibfnamefont {M.~E.}\ \bibnamefont {Tobar}}, \ and\ \bibinfo {author}
  {\bibfnamefont {A.~N.}\ \bibnamefont {Luiten}},\ }\href {\doibase
  10.1103/PhysRevD.70.051902} {\bibfield  {journal} {\bibinfo  {journal} {Phys.
  Rev.}\ }\textbf {\bibinfo {volume} {D70}},\ \bibinfo {pages} {051902}
  (\bibinfo {year} {2004})},\ \Eprint {http://arxiv.org/abs/hep-ph/0407232}
  {arXiv:hep-ph/0407232 [hep-ph]} \BibitemShut {NoStop}%
\bibitem [{\citenamefont {Herrmann}\ \emph {et~al.}(2005)\citenamefont
  {Herrmann}, \citenamefont {Senger}, \citenamefont {Kovalchuk}, \citenamefont
  {Muller},\ and\ \citenamefont {Peters}}]{Herrmann:2005qe}%
  \BibitemOpen
  \bibfield  {author} {\bibinfo {author} {\bibfnamefont {S.}~\bibnamefont
  {Herrmann}}, \bibinfo {author} {\bibfnamefont {A.}~\bibnamefont {Senger}},
  \bibinfo {author} {\bibfnamefont {E.}~\bibnamefont {Kovalchuk}}, \bibinfo
  {author} {\bibfnamefont {H.}~\bibnamefont {Muller}}, \ and\ \bibinfo {author}
  {\bibfnamefont {A.}~\bibnamefont {Peters}},\ }\href {\doibase
  10.1103/PhysRevLett.95.150401} {\bibfield  {journal} {\bibinfo  {journal}
  {Phys. Rev. Lett.}\ }\textbf {\bibinfo {volume} {95}},\ \bibinfo {pages}
  {150401} (\bibinfo {year} {2005})},\ \Eprint
  {http://arxiv.org/abs/physics/0508097} {arXiv:physics/0508097
  [physics.class-ph]} \BibitemShut {NoStop}%
\bibitem [{\citenamefont {Antonini}\ \emph {et~al.}(2005)\citenamefont
  {Antonini}, \citenamefont {Okhapkin}, \citenamefont {Goklu},\ and\
  \citenamefont {Schiller}}]{Antonini:2005yb}%
  \BibitemOpen
  \bibfield  {author} {\bibinfo {author} {\bibfnamefont {P.}~\bibnamefont
  {Antonini}}, \bibinfo {author} {\bibfnamefont {M.}~\bibnamefont {Okhapkin}},
  \bibinfo {author} {\bibfnamefont {E.}~\bibnamefont {Goklu}}, \ and\ \bibinfo
  {author} {\bibfnamefont {S.}~\bibnamefont {Schiller}},\ }\href {\doibase
  10.1103/PhysRevA.71.050101} {\bibfield  {journal} {\bibinfo  {journal} {Phys.
  Rev.}\ }\textbf {\bibinfo {volume} {A71}},\ \bibinfo {pages} {050101}
  (\bibinfo {year} {2005})},\ \Eprint {http://arxiv.org/abs/gr-qc/0504109}
  {arXiv:gr-qc/0504109 [gr-qc]} \BibitemShut {NoStop}%
\bibitem [{\citenamefont {Stanwix}\ \emph {et~al.}(2006)\citenamefont
  {Stanwix}, \citenamefont {Tobar}, \citenamefont {Wolf}, \citenamefont
  {Locke},\ and\ \citenamefont {Ivanov}}]{Stanwix:2006jb}%
  \BibitemOpen
  \bibfield  {author} {\bibinfo {author} {\bibfnamefont {P.~L.}\ \bibnamefont
  {Stanwix}}, \bibinfo {author} {\bibfnamefont {M.~E.}\ \bibnamefont {Tobar}},
  \bibinfo {author} {\bibfnamefont {P.}~\bibnamefont {Wolf}}, \bibinfo {author}
  {\bibfnamefont {C.~R.}\ \bibnamefont {Locke}}, \ and\ \bibinfo {author}
  {\bibfnamefont {E.~N.}\ \bibnamefont {Ivanov}},\ }\href {\doibase
  10.1103/PhysRevD.74.081101} {\bibfield  {journal} {\bibinfo  {journal} {Phys.
  Rev.}\ }\textbf {\bibinfo {volume} {D74}},\ \bibinfo {pages} {081101}
  (\bibinfo {year} {2006})},\ \Eprint {http://arxiv.org/abs/gr-qc/0609072}
  {arXiv:gr-qc/0609072 [gr-qc]} \BibitemShut {NoStop}%
\bibitem [{\citenamefont {Herrmann}\ \emph {et~al.}(2006)\citenamefont
  {Herrmann}, \citenamefont {Senger}, \citenamefont {Peters}, \citenamefont
  {Kovalchuk},\ and\ \citenamefont {Muller}}]{Herrmann:2006zza}%
  \BibitemOpen
  \bibfield  {author} {\bibinfo {author} {\bibfnamefont {S.}~\bibnamefont
  {Herrmann}}, \bibinfo {author} {\bibfnamefont {A.}~\bibnamefont {Senger}},
  \bibinfo {author} {\bibfnamefont {A.}~\bibnamefont {Peters}}, \bibinfo
  {author} {\bibfnamefont {E.}~\bibnamefont {Kovalchuk}}, \ and\ \bibinfo
  {author} {\bibfnamefont {A.}~\bibnamefont {Muller}},\ }\bibfield  {booktitle}
  {\emph {\bibinfo {booktitle} {{339th WE Heraeus Seminar on Special
  Relativity: Will It Survive the Next 100 Years? Potsdam, Germany, February
  13-18, 2005}}},\ }\href {\doibase 10.1007/3-540-34523-X_13} {\bibfield
  {journal} {\bibinfo  {journal} {Lect. Notes Phys.}\ }\textbf {\bibinfo
  {volume} {702}},\ \bibinfo {pages} {385} (\bibinfo {year}
  {2006})}\BibitemShut {NoStop}%
\bibitem [{\citenamefont {Herrmann}\ \emph {et~al.}(2009)\citenamefont
  {Herrmann}, \citenamefont {Senger}, \citenamefont {Mohle}, \citenamefont
  {Nagel}, \citenamefont {Kovalchuk},\ and\ \citenamefont
  {Peters}}]{Herrmann:2009zzb}%
  \BibitemOpen
  \bibfield  {author} {\bibinfo {author} {\bibfnamefont {S.}~\bibnamefont
  {Herrmann}}, \bibinfo {author} {\bibfnamefont {A.}~\bibnamefont {Senger}},
  \bibinfo {author} {\bibfnamefont {K.}~\bibnamefont {Mohle}}, \bibinfo
  {author} {\bibfnamefont {M.}~\bibnamefont {Nagel}}, \bibinfo {author}
  {\bibfnamefont {E.~V.}\ \bibnamefont {Kovalchuk}}, \ and\ \bibinfo {author}
  {\bibfnamefont {A.}~\bibnamefont {Peters}},\ }\href {\doibase
  10.1103/PhysRevD.80.105011} {\bibfield  {journal} {\bibinfo  {journal} {Phys.
  Rev.}\ }\textbf {\bibinfo {volume} {D80}},\ \bibinfo {pages} {105011}
  (\bibinfo {year} {2009})},\ \Eprint {http://arxiv.org/abs/1002.1284}
  {arXiv:1002.1284 [physics.class-ph]} \BibitemShut {NoStop}%
\bibitem [{\citenamefont {Nagel}\ \emph {et~al.}(2015)\citenamefont {Nagel},
  \citenamefont {Parker}, \citenamefont {Kovalchuk}, \citenamefont {Stanwix},
  \citenamefont {Hartnett}, \citenamefont {Ivanov}, \citenamefont {Peters},\
  and\ \citenamefont {Tobar}}]{Nagel:2014aga}%
  \BibitemOpen
  \bibfield  {author} {\bibinfo {author} {\bibfnamefont {M.}~\bibnamefont
  {Nagel}}, \bibinfo {author} {\bibfnamefont {S.~R.}\ \bibnamefont {Parker}},
  \bibinfo {author} {\bibfnamefont {E.~V.}\ \bibnamefont {Kovalchuk}}, \bibinfo
  {author} {\bibfnamefont {P.~L.}\ \bibnamefont {Stanwix}}, \bibinfo {author}
  {\bibfnamefont {J.~G.}\ \bibnamefont {Hartnett}}, \bibinfo {author}
  {\bibfnamefont {E.~N.}\ \bibnamefont {Ivanov}}, \bibinfo {author}
  {\bibfnamefont {A.}~\bibnamefont {Peters}}, \ and\ \bibinfo {author}
  {\bibfnamefont {M.~E.}\ \bibnamefont {Tobar}},\ }\href {\doibase
  10.1038/ncomms9174} {\bibfield  {journal} {\bibinfo  {journal} {Nature
  Commun.}\ }\textbf {\bibinfo {volume} {6}},\ \bibinfo {pages} {8174}
  (\bibinfo {year} {2015})},\ \Eprint {http://arxiv.org/abs/1412.6954}
  {arXiv:1412.6954 [hep-ph]} \BibitemShut {NoStop}%
\bibitem [{\citenamefont {Mewes}(2012)}]{Mewes:2012sm}%
  \BibitemOpen
  \bibfield  {author} {\bibinfo {author} {\bibfnamefont {M.}~\bibnamefont
  {Mewes}},\ }\href {\doibase 10.1103/PhysRevD.85.116012} {\bibfield  {journal}
  {\bibinfo  {journal} {Phys. Rev.}\ }\textbf {\bibinfo {volume} {D85}},\
  \bibinfo {pages} {116012} (\bibinfo {year} {2012})},\ \Eprint
  {http://arxiv.org/abs/1203.5331} {arXiv:1203.5331 [hep-ph]} \BibitemShut
  {NoStop}%
\bibitem [{\citenamefont {Parker}\ \emph {et~al.}(2015)\citenamefont {Parker},
  \citenamefont {Mewes}, \citenamefont {Baynes},\ and\ \citenamefont
  {Tobar}}]{Parker:2015ena}%
  \BibitemOpen
  \bibfield  {author} {\bibinfo {author} {\bibfnamefont {S.~R.}\ \bibnamefont
  {Parker}}, \bibinfo {author} {\bibfnamefont {M.}~\bibnamefont {Mewes}},
  \bibinfo {author} {\bibfnamefont {F.~N.}\ \bibnamefont {Baynes}}, \ and\
  \bibinfo {author} {\bibfnamefont {M.~E.}\ \bibnamefont {Tobar}},\ }\href
  {\doibase 10.1016/j.physleta.2015.08.006} {\bibfield  {journal} {\bibinfo
  {journal} {Phys. Lett.}\ }\textbf {\bibinfo {volume} {A379}},\ \bibinfo
  {pages} {2681} (\bibinfo {year} {2015})},\ \Eprint
  {http://arxiv.org/abs/1508.02490} {arXiv:1508.02490 [hep-ph]} \BibitemShut
  {NoStop}%
\bibitem [{\citenamefont {Kostelecky}\ and\ \citenamefont
  {Lehnert}(2001)}]{Kostelecky:2000mm}%
  \BibitemOpen
  \bibfield  {author} {\bibinfo {author} {\bibfnamefont {V.~A.}\ \bibnamefont
  {Kostelecky}}\ and\ \bibinfo {author} {\bibfnamefont {R.}~\bibnamefont
  {Lehnert}},\ }\href {\doibase 10.1103/PhysRevD.63.065008} {\bibfield
  {journal} {\bibinfo  {journal} {Phys. Rev.}\ }\textbf {\bibinfo {volume}
  {D63}},\ \bibinfo {pages} {065008} (\bibinfo {year} {2001})},\ \Eprint
  {http://arxiv.org/abs/hep-th/0012060} {arXiv:hep-th/0012060 [hep-th]}
  \BibitemShut {NoStop}%
\bibitem [{\citenamefont {{Xiao}}(2009)}]{2009arXiv0908.0787X}%
  \BibitemOpen
  \bibfield  {author} {\bibinfo {author} {\bibfnamefont {M.-w.}\ \bibnamefont
  {{Xiao}}},\ }\href@noop {} {\bibfield  {journal} {\bibinfo  {journal} {ArXiv
  e-prints}\ } (\bibinfo {year} {2009})},\ \Eprint
  {http://arxiv.org/abs/0908.0787} {arXiv:0908.0787 [math-ph]} \BibitemShut
  {NoStop}%
\bibitem [{\citenamefont {Kostelecky}(2001)}]{Kostelecky:2001xz}%
  \BibitemOpen
  \bibfield  {author} {\bibinfo {author} {\bibfnamefont {V.~A.}\ \bibnamefont
  {Kostelecky}},\ }in\ \href@noop {} {\emph {\bibinfo {booktitle} {{The role of
  neutrinos, strings, gravity, and variable cosmological constant in elementary
  particle physics. Proceedings, 29th International Conference on high energy
  physics and cosmology, Orbis Scientiae, Coral Gables, USA, December 14-17,
  2000}}}}\ (\bibinfo {year} {2001})\ pp.\ \bibinfo {pages} {57--68},\ \Eprint
  {http://arxiv.org/abs/hep-ph/0104227} {arXiv:hep-ph/0104227 [hep-ph]}
  \BibitemShut {NoStop}%
\bibitem [{\citenamefont {Lehnert}(2002)}]{Lehnert:2002yi}%
  \BibitemOpen
  \bibfield  {author} {\bibinfo {author} {\bibfnamefont {R.}~\bibnamefont
  {Lehnert}},\ }in\ \href {\doibase 10.1142/9789812778123_0021} {\emph
  {\bibinfo {booktitle} {{CPT and Lorentz symmetry. Proceedings: 2nd Meeting,
  Bloomington, USA, Aug 15-18, 2001}}}}\ (\bibinfo {year} {2002})\ pp.\
  \bibinfo {pages} {190--199},\ \Eprint {http://arxiv.org/abs/hep-th/0201238}
  {arXiv:hep-th/0201238 [hep-th]} \BibitemShut {NoStop}%
\bibitem [{\citenamefont {Colladay}(2017)}]{Colladay:2017bon}%
  \BibitemOpen
  \bibfield  {author} {\bibinfo {author} {\bibfnamefont {D.}~\bibnamefont
  {Colladay}},\ }\href {\doibase 10.1016/j.physletb.2017.07.027} {\bibfield
  {journal} {\bibinfo  {journal} {Phys. Lett.}\ }\textbf {\bibinfo {volume}
  {B772}},\ \bibinfo {pages} {694} (\bibinfo {year} {2017})},\ \Eprint
  {http://arxiv.org/abs/1706.06637} {arXiv:1706.06637 [hep-th]} \BibitemShut
  {NoStop}%
\bibitem [{\citenamefont {Colladay}(2018)}]{Colladay:2017cey}%
  \BibitemOpen
  \bibfield  {author} {\bibinfo {author} {\bibfnamefont {D.}~\bibnamefont
  {Colladay}},\ }\bibfield  {booktitle} {\emph {\bibinfo {booktitle}
  {{Proceedings, Workshop CPT and Lorentz Symmetry in Field Theory: Faro,
  Portugal, July 6-7, 2017}}},\ }\href {\doibase
  10.1088/1742-6596/952/1/012011} {\bibfield  {journal} {\bibinfo  {journal}
  {J. Phys. Conf. Ser.}\ }\textbf {\bibinfo {volume} {952}},\ \bibinfo {pages}
  {012011} (\bibinfo {year} {2018})},\ \Eprint
  {http://arxiv.org/abs/1710.02479} {arXiv:1710.02479 [hep-ph]} \BibitemShut
  {NoStop}%
\bibitem [{\citenamefont {Husain}\ \emph {et~al.}(2013)\citenamefont {Husain},
  \citenamefont {Seahra},\ and\ \citenamefont {Webster}}]{Husain:2013zda}%
  \BibitemOpen
  \bibfield  {author} {\bibinfo {author} {\bibfnamefont {V.}~\bibnamefont
  {Husain}}, \bibinfo {author} {\bibfnamefont {S.~S.}\ \bibnamefont {Seahra}},
  \ and\ \bibinfo {author} {\bibfnamefont {E.~J.}\ \bibnamefont {Webster}},\
  }\href {\doibase 10.1103/PhysRevD.88.024014} {\bibfield  {journal} {\bibinfo
  {journal} {Phys. Rev.}\ }\textbf {\bibinfo {volume} {D88}},\ \bibinfo {pages}
  {024014} (\bibinfo {year} {2013})},\ \Eprint {http://arxiv.org/abs/1305.2814}
  {arXiv:1305.2814 [hep-th]} \BibitemShut {NoStop}%
\bibitem [{\citenamefont {Bachmann}\ and\ \citenamefont
  {Kempf}(2008)}]{1751-8121-41-16-164021}%
  \BibitemOpen
  \bibfield  {author} {\bibinfo {author} {\bibfnamefont {S.}~\bibnamefont
  {Bachmann}}\ and\ \bibinfo {author} {\bibfnamefont {A.}~\bibnamefont
  {Kempf}},\ }\href {http://stacks.iop.org/1751-8121/41/i=16/a=164021}
  {\bibfield  {journal} {\bibinfo  {journal} {Journal of Physics A:
  Mathematical and Theoretical}\ }\textbf {\bibinfo {volume} {41}},\ \bibinfo
  {pages} {164021} (\bibinfo {year} {2008})}\BibitemShut {NoStop}%
\end{thebibliography}%
\end{document}